\documentclass[english,final]{IEEEtran}
 
\makeatletter

\usepackage[T1]{fontenc}
\usepackage[latin9]{inputenc}
\usepackage{amsthm}
\usepackage{amsmath}
\usepackage{amssymb}
\usepackage{graphicx}
\usepackage{dsfont}
\usepackage{cite}
\usepackage{amsmath}
\usepackage{array}
\usepackage{multirow}
\usepackage{caption}
\usepackage{color}

\usepackage{multicol}

\theoremstyle{plain}

\theoremstyle{plain}

\theoremstyle{plain}
\newtheorem{lem}{\protect\lemmaname}
\theoremstyle{plain}
\newtheorem{thm}{\protect\theoremname}
\theoremstyle{plain}
\newtheorem{cor}{\protect\corollaryname}  
\theoremstyle{definition}
\newtheorem{defn}{\protect\definitionname}
\theoremstyle{definition}

\theoremstyle{definition}
\newtheorem{rem}{\protect\remarkname}

\makeatother
  
\usepackage{babel} 

\providecommand{\claimname}{Claim}
\providecommand{\lemmaname}{Lemma}
\providecommand{\propositionname}{Proposition}
\providecommand{\theoremname}{Theorem}
\providecommand{\corollaryname}{Corollary} 
\providecommand{\definitionname}{Definition}
\providecommand{\assumptionname}{Assumption}
\providecommand{\remarkname}{Remark}

\newcommand{\overbar}[1]{\mkern 1.25mu\overline{\mkern-1.25mu#1\mkern-0.25mu}\mkern 0.25mu}

\DeclareMathOperator*{\argmax}{arg\,max}

\newcommand{\Runder}[1]{\underbar{R}^{({\rm #1})}}
\newcommand{\Rbar}[1]{\overbar{R}^{({\rm #1})}}
\newcommand{\NDhat}{\widehat{\mathrm{ND}}}
\newcommand{\PDhat}{\widehat{\mathrm{PD}}}
\newcommand{\pdhat}{\widehat{\mathrm{pd}}}
\newcommand{\ND}{\mathrm{ND}}
\newcommand{\PD}{\mathrm{PD}}
\newcommand{\Shat}{\widehat{S}}
\newcommand{\peid}{P_{\mathrm{e},1}^{(\mathrm{D})}}
\newcommand{\peind}{P_{\mathrm{e},1}^{(\mathrm{ND})}}
\newcommand{\peiid}{P_{\mathrm{e},2}^{(\mathrm{D})}}
\newcommand{\peiind}{P_{\mathrm{e},2}^{(\mathrm{ND})}}
\newcommand{\epsid}{\epsilon_{1}^{(\mathrm{D})}}
\newcommand{\epsind}{\epsilon_{1}^{(\mathrm{ND})}}
\newcommand{\epsiid}{\epsilon_{2}^{(\mathrm{D})}}
\newcommand{\epsiind}{\epsilon_{2}^{(\mathrm{ND})}}
\newcommand{\nid}{n_1^{(\mathrm{D})}}
\newcommand{\nind}{n_1^{(\mathrm{ND})}}
\newcommand{\niid}{n_2^{(\mathrm{D})}}
\newcommand{\niind}{n_2^{(\mathrm{ND})}}
\newcommand{\Npd}{N_{\mathrm{pos}}^{(\mathrm{D})}}
\newcommand{\Npnd}{N_{\mathrm{pos}}^{(\mathrm{ND})}}
\newcommand{\Nnnd}{N_{\mathrm{neg}}^{(\mathrm{ND})}}
\newcommand{\npd}{n_{\mathrm{pos}}^{(\mathrm{D})}}
\newcommand{\npnd}{n_{\mathrm{pos}}^{(\mathrm{ND})}}

\newcommand{\nneg}{n_{\mathrm{neg}}}

\newcommand{\Npposj}{N'_{\mathrm{pos},j}}

\newcommand{\Ntilposi}{\widetilde{N}_{\mathrm{pos},i}}
\newcommand{\Ntilnegi}{\widetilde{N}_{\mathrm{neg},i}}
\newcommand{\Ntilposj}{\widetilde{N}_{\mathrm{pos},j}}
\newcommand{\Ntilnegj}{\widetilde{N}_{\mathrm{neg},j}}
\newcommand{\Nother}{N_{\mathrm{other}}}
\newcommand{\qother}{q_{\mathrm{other}}}
\newcommand{\qpposj}{q'_{\mathrm{pos},j}}
\newcommand{\qneg}{q_{\mathrm{neg}}}

\newcommand{\qtil}{\widetilde{q}}

\newcommand{\Nnegj}{N_{\mathrm{neg},j}}

\newcommand{\mutild}{\tilde{\mu}^{(\mathrm{D})}}
\newcommand{\mutilnd}{\tilde{\mu}^{(\mathrm{ND})}}

\newcommand{\Nneg}{N_{\mathrm{neg}}}
\newcommand{\Npos}{N_{\mathrm{pos}}}

\newcommand{\Bernoulli}{\mathrm{Bernoulli}}
\newcommand{\Binomial}{\mathrm{Binomial}}

\newcommand{\pe}{P_{\mathrm{e}}}

\newcommand{\Xv}{\mathbf{X}}
\newcommand{\yv}{\mathbf{y}}
\newcommand{\Yv}{\mathbf{Y}}

\newcommand{\Ec}{\mathcal{E}}

\newcommand{\EE}{\mathbb{E}}
\newcommand{\PP}{\mathbb{P}}

\newcommand{\var}{\mathrm{Var}}
\newcommand{\cov}{\mathrm{Cov}}

\usepackage{hyperref} 
\usepackage{enumitem}
\usepackage{comment}
\usepackage{float}

\floatstyle{ruled}
\newfloat{algorithm}{tbp}{loa}
\providecommand{\algorithmname}{Algorithm}
\floatname{algorithm}{\protect\algorithmname}

\newcommand{\ep}{{\mathbb E}}
\newcommand{\pr}{{\mathbb P}}

\newcommand{\TT}{{\mathcal{T}}}

\newcommand{\thetacrit}{\theta^{(\mathrm{RZ})}_{\rm crit}}
\newcommand{\thetaopt}{\theta_{\rm opt}}
\newcommand{\thetacritZ}{\theta^{(\mathrm{Z})}_{\rm crit}}

\newcommand{\ncomp}{n_{\rm COMP}}
\newcommand{\ndd}{n_{\rm DD}}

\setenumerate[1]{label=\arabic*.}

\makeatletter
\newcommand{\manuallabel}[2]{\def\@currentlabel{#2}\label{#1}}
\makeatother

\begin{document} 

\title{Noisy Non-Adaptive Group Testing: \\ A (Near-)Definite Defectives Approach}
\author{Jonathan Scarlett and Oliver Johnson}
\maketitle

\begin{abstract}
    The group testing problem consists of determining a small set of defective items from a larger set of items based on a number of possibly-noisy tests, and is relevant in applications such as medical testing, communication protocols, pattern matching, and more.  We study the noisy version of this problem, where the outcome of each standard noiseless group test is subject to independent noise, corresponding to passing the noiseless result through a binary channel. We introduce a class of algorithms that we refer to as Near-Definite Defectives (NDD), and study bounds on the required number of tests for asymptotically vanishing error probability under Bernoulli random test designs. In addition, we study algorithm-independent converse results, giving lower bounds on the required number of tests under Bernoulli test designs.  Under reverse Z-channel noise, the achievable rates and converse results match in a broad range of sparsity regimes, and under Z-channel noise, the two match in a narrower range of dense/low-noise regimes.  We observe that although these two channels have the same Shannon capacity when viewed as a communication channel, they can behave quite differently when it comes to group testing.  Finally, we extend our analysis of these noise models to a general binary noise model (including symmetric noise), and show improvements over known existing bounds in broad scaling regimes.
\end{abstract}
\begin{IEEEkeywords}
    Group testing, performance bounds, sparsity, Z channel, information-theoretic limits
\end{IEEEkeywords}

\long\def\symbolfootnote[#1]#2{\begingroup\def\thefootnote{\fnsymbol{footnote}}\footnote[#1]{#2}\endgroup}

\symbolfootnote[0]{ Jonathan Scarlett is with the Department of Computer Science, National University of Singapore, Singapore, and also with the Department of Mathematics, National University of Singapore, Singapore (e-mail: scarlett@comp.nus.edu.sg).}
\symbolfootnote[0]{ Oliver Johnson is with the School of Mathematics, University of Bristol, UK (e-mail: maotj@bristol.ac.uk).}
\symbolfootnote[0]{ Copyright (c) 2021 IEEE. Personal use of this material is permitted.  However, permission to use this material for any other purposes must be obtained from the IEEE by sending a request to pubs-permissions@ieee.org.}
\vspace*{-0.5cm}

%
% INTRODUCTION
%
\section{Introduction} \label{sec:intro}

The group testing problem consists of determining a small subset of ``defective'' items within a larger set of items, based on a number of possibly-noisy tests. As described in more detail in the survey monograph \cite{Ald19}, this problem has a history in medical testing \cite{Dor43}, and has regained significant attention with subsequent applications in areas such as communication protocols \cite{Ant11},  DNA sequencing \cite{erlich2}, data forensics \cite{goodrich2}, pattern matching \cite{Cli10}, and database systems \cite{Cor05}, as well as new connections with compressive sensing \cite{Gil08,Gil07}.  The general setup involves a sequence of tests, each of which acts on a particular subset (or ``pool'') of items and produces an outcome $Y$ that can be a deterministic or random function of the defectivity status of the items in the pool. 

In recent years, the information-theoretic limits and performance limits of practical algorithms for {\em noiseless} group testing have become increasingly well-understood \cite{Ati12,Bal13,Ald14a,Sca15b,Ald15,Joh16,Coj19}.  By comparison, random noise settings are somewhat less well-understood despite ongoing advances \cite{Cha14,Laa14,Sca15b,Sca17b,Sca18}.  In particular, the algorithm that gives the best known {\em noiseless} performance guarantees in most sparsity regimes (among practical algorithms), known as Definite Defectives (DD) \cite{Ald14a,Joh16}, has no previous noisy counterpart.  In this paper, we address this gap by introducing and studying noisy variants of DD, and showing that they provide the best known performance bounds in a wide range of settings depending on the sparsity and noise level.

\subsection{Overview of Noiseless Group Testing}

Let $p$ denote the number of items, and let $S \subseteq \{1,\dotsc,p\}$ denote the set of defective items.  In the standard noiseless setting first introduced in \cite{Dor43}, the outcome of each test takes the form
\begin{equation}
    Y = \bigvee_{j \in S} X_j, \label{eq:gt_noiseless_model}
\end{equation}
where the test vector $X = (X_1,\dotsc,X_p) \in \{0,1\}^p$ indicates which items are included in the test.  That is, the resulting outcome $Y = 1$ if and only if  at least one defective item was included in the test.   We refer to tests with $Y=1$ as positive, and tests with $Y=0$ as negative.

Given the tests and their outcomes, a \emph{decoder} forms an estimate $\Shat$ of $S$.  One wishes to design a sequence of tests $X^{(1)},\dotsc,X^{(n)}$, with $n$ ideally as small as possible, such that the decoder recovers $S$ with probability arbitrarily close to one.  The error probability is given by 
\begin{equation}
    \pe := \PP[\Shat \ne S], \label{eq:pe}
\end{equation}
and is taken over the randomness of the defective set $S$, the tests $X^{(1)},\dotsc,X^{(n)}$ (if randomized), and the test outcomes $Y^{(1)},\dotsc,Y^{(n)}$ (if noisy).  For convenience, we represent the tests as a matrix $\Xv \in \{0,1\}^{n \times p}$, where the $i$-th row is $X^{(i)}$ and represents the $i$-th test.

In this paper, we consider the case that, for a given sparsity level $k$, the defective set $S$ is chosen uniformly on the $\binom{p}{k}$ subsets of $\{1,\dotsc,p\}$ of cardinality $k$. 
% A standard information-theoretic argument (see for example, \cite{Cha11}) shows that in order to have error probability of zero, we require that the number of tests $n \geq \log_2 \binom{p}{k}$. 
Following recent works such as \cite{Bal13}, we define the rate (in bits/test) of a group-testing algorithm using $n$ tests to be\footnote{Throughout the paper, $\log$ refers to natural logarithms taken to base $e$, and we write $\log_2$ for base 2 logarithms.}
\begin{equation} \label{eq:defrate} 
    R := \frac{ \log_2 \binom{p}{k}}{n}, 
\end{equation}
which we can think of as the number of bits of information about the defective set learned per test.  We consider the asymptotic regime where $p \rightarrow \infty$ with $k \asymp p^{\theta}$ for some $\theta \in (0,1)$,\footnote{Here and subsequently, $k \asymp p^{\theta}$ means that $\frac{k}{p^{\theta}}$ is bounded away from both $0$ and $\infty$ in the limit as $p \to \infty$.} so we will often use the equivalent limiting definition that
\begin{equation} \label{eq:defrateeq} 
    R \sim \frac{ k \log_2(p/k)}{n},
\end{equation}
where $\sim$ denotes asymptotic equality up to a multiplicative $1+o(1)$ term.

It is well known from standard information-theoretic arguments (e.g., \cite{Bal13}) that no algorithm with rate above $1$ bit/test can have vanishing error probability. For noiseless adaptive group testing (where the choice of test $X^{(i+1)}$ can depend on the previous tests $X^{(1)}, \ldots, X^{(i)}$ and their outcomes $Y^{(1)}, \ldots, Y^{(i)}$), Hwang's algorithm \cite{hwang} has error probability tending to zero with rate approaching one in any regime where $k = o(p)$, and is therefore asymptotically optimal.

In this paper, we study non-adaptive group testing, where the entire collection of tests $\Xv \in \{0,1\}^{n \times p}$ is fixed in advance.  
%These tests are described by a binary matrix $\Xv \in \{0,1\}^{n \times p}$, with rows corresponding to tests and columns corresponding to items. Here matrix element $X_{ij} = 1$ if item $j$ appears in the $i$-th test, and $X_{ij} = 0$ otherwise. 
We focus in particular on {\em Bernoulli testing}, where the entries of $\Xv$ are independently drawn at random from $\Bernoulli\big( \frac{\nu}{k} \big)$ for some parameter $\nu > 0$.

In the noiseless case, it is known that the non-adaptive definite defectives (DD) algorithm \cite{Ald14} is both practically implementable (in terms of storage and processing requirements) and performs well in terms of rate.  Specifically, under Bernoulli testing in the regime $k \asymp p^{\theta}$, the DD algorithm achieves $\pe \rightarrow 0$ when \cite[Theorem 12]{Ald14a}
\begin{equation} \label{eq:ddnoiselessrate} 
    R < \frac{1}{e \log 2} \min \left\{ 1, \frac{1-\theta}{\theta} \right\}.
\end{equation}
Furthermore, the DD algorithm is known to be rate-optimal for sufficiently dense problems (specifically, for $\theta > 1/2$) under Bernoulli testing, in the sense that {\em any} algorithm has $\pe$ bounded away from zero if the rate satisfies \cite{Ald14a,Ald15}
\begin{equation} \label{eq:ddnoiselesscap} 
    R > \frac{1}{e \log 2}  \frac{1-\theta}{\theta}.
\end{equation}

The contribution of this paper is to extend the bounds of the form \eqref{eq:ddnoiselessrate} and \eqref{eq:ddnoiselesscap} to noisy
group testing models, by introducing and referring to a class of algorithms which we refer to as noisy DD (NDD).  For clarity, we first review both (noiseless) DD \cite{Ald14} and a related algorithm called COMP \cite{Cha11}, which forms the first stage of DD.
% Since these algorithms are designed as generalizations of the DD algorithm of \cite{Ald14}, for completeness we briefly review both DD and a related algorithm \cite{Cha11} called COMP, which forms the first stage of  DD.

\begin{defn} \mbox{ }
The Combinatorial Orthogonal Matching Pursuit (COMP) and Definite Defectives (DD) algorithms for {\em noiseless} non-adaptive group testing are defined as follows:
\begin{enumerate}
\item Since $Y=1$ if and only if the test pool contains a defective item, we can be sure that each item that appears in a negative test is not defective. We can form a list of such items formed from all tests, which we refer to as $\ND$; the rest of the items $\PD :=  \{1,\dotsc,p\} \backslash \ND$ are considered ``possibly defective''. The COMP algorithm simply estimates $S$ using the set of possible defective items, $\Shat = \PD$.
\item The DD algorithm starts with the possible defective items $\PD$. Since every positive test must contain at least one 
defective item, if a test with $Y=1$ contains exactly one item from $\PD$, then we can be certain that the  item in question is defective.  The
DD algorithm outputs $\Shat$ equaling the set of PD items that appear in a positive test with no other PD item.
\end{enumerate}
\end{defn}

%\begin{example}
%Consider the following group testing problem, where items $S = \{1,6 \}$  are defective (corresponding to bold columns).
%$$
%\Xv =  \left( \mbox{\begin{tabular}{ccccccc}
%   {\bf 1} & 0  & 0 & 0 & 0  & {\bf 0} & 1  \\
% {\bf 1} & 0 & 1 & 1  & 0 & {\bf 0} &  0 \\
% {\bf 0} & 0 & 1 & 0  & 1 & {\bf 0} &  0  \\
% {\bf 0} & 0 & 1 & 0  & 0 & {\bf 1} & 0  \\
%  {\bf 0} & 1 & 0 & 0  & 0 & {\bf 0} & 1   \\ 
%\end{tabular}}
%\right)
%\mbox{ \;\;\;\;\;\;\;\;}
%Y = \left( \mbox{\begin{tabular}{c}
%1 \\ 1 \\ 0 \\ 1 \\ 0 \\
%\end{tabular}}
%\right). $$
%
%\begin{enumerate}
%\item From test 3 (negative), we deduce that items 3 and  5 are non-defective. From test 5  (negative), we deduce that items 2 and 7
%are non-defective. Hence $ND = \{ 2, 3, 5, 7 \}$, and COMP estimates the defective set by $\Shat = \PD = \{ 1, 4, 6 \}$.
%\item Test 1 (positive) contains item 1 and no other PD item, so item 1 is definitely defective. Similarly test 4 (positive) contains item 6 and no other PD
%item, so item 6 is definitely defective. (Although Test 2 is positive, it contains two PD items, items 1 and 4, so the DD algorithm ignores it). Hence, DD estimates the defective set by $\Shat = \{ 1, 6 \}$. 
%\end{enumerate}
%\end{example}

\subsection{Noisy Group Testing} \label{sec:setup}
 
Generalizing \eqref{eq:gt_noiseless_model}, we consider noisy models that correspond to passing the quantity $U = \vee_{j \in S} X_j$ through a noisy channel $P_{Y|U}$.  We focus in particular on the following special cases, each of which depends on one or two noise parameters that can be set to zero to recover the noiseless model.

\begin{defn}\label{def:chanmod} We define the following noise models, illustrated in Figure \ref{fig:Zchannels}:
\begin{enumerate}
 \item The general binary channel model
 % \footnote{Noise models can also be {\em non-binary} in the sense of depending on the number of defectives in the test, with the dilution noise model being a common example \cite[Sec.~3.1]{Ald19}.  However, such models are beyond the scope of this paper.} 
 is given by
\begin{gather} 
    P_{Y|U}(0|0) = 1-\rho_{01}, \quad P_{Y|U}(1|0) = \rho_{01}, \nonumber \\
    P_{Y|U}(0|1) = \rho_{10}, \quad P_{Y|U}(1|1) = 1-\rho_{10} \label{eq:gt_gen_model}
    \end{gather}
    for some noise levels $\rho_{01}$ and $\rho_{10}$ both in $[0,1]$.
    \item The Z-channel model is obtained by taking $\rho_{01} = 0$ and $\rho_{10} = \rho$ in the general model, yielding
so that 
    \begin{gather}
    P_{Y|U}(0|0) = 1, \quad P_{Y|U}(1|0) = 0, \nonumber \\
    P_{Y|U}(0|1) = \rho, \quad P_{Y|U}(1|1) = 1-\rho
    \end{gather}
    for some noise level $\rho \in [0,1]$.  Hence, if $Y=1$, then the test must contain a defective item.
    \item The reverse Z-channel (RZ-channel) model, also known as the addition noise model \cite{Ati12}, is obtained by taking $\rho_{01} = \rho$ and $\rho_{10} = 0$, yielding
    \begin{gather}
        P_{Y|U}(0|0) = 1-\rho, \quad P_{Y|U}(1|0) = \rho, \nonumber \\
        P_{Y|U}(0|1) = 0, \quad P_{Y|U}(1|1) = 1
    \end{gather}
    for some  noise level  $\rho \in [0,1]$.  Hence, if $Y=0$, then the test must contain no defective items.
    \item The symmetric noise model is obtained by taking $\rho_{01} = \rho_{10} = \rho$ in the general model, yielding
    \begin{equation} \label{eq:gt_symm_model}
        P_{Y|U}(y|u) =
        \begin{cases}
            1-\rho & y = u \\
            \rho & y \ne u
        \end{cases}
    \end{equation}
    for some  noise level  $\rho \in [0,1]$.
\end{enumerate}
\end{defn}

\begin{figure}
    \begin{centering}
        \includegraphics[width=\columnwidth]{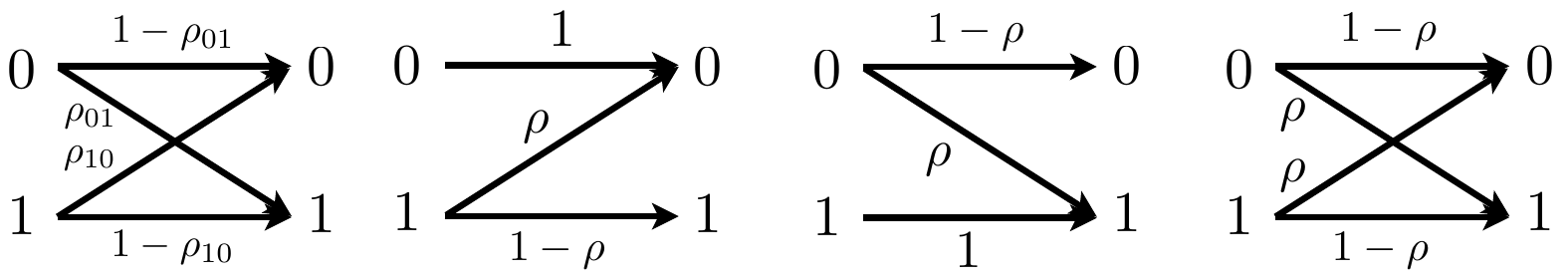}
        \par
    \end{centering}
    
    \caption{General binary channel, Z-channel, reverse Z-channel, and binary symmetric channel. \label{fig:Zchannels}}
\end{figure}
   
In this paper, we will focus primarily on the Z and reverse Z-channel models, applications of which are discussed in Section \ref{sec:applications}.  As well as being important in their own right for applications, these models serve as useful stepping stones towards the general binary model (including the symmetric model), which is handled in Appendix \ref{app:symmetric}.

Note that for the general binary channel model, we can assume without loss of generality that $\rho_{01} + \rho_{10} \leq 1$. This is because (as for a standard binary symmetric channel) if this were not true, we could flip the outcome of all tests as a pre-processing step, producing a new channel with $\widetilde{\rho}_{01} = 1 - \rho_{01}$ and $\widetilde{\rho}_{10} = 1 - \rho_{10}$ satisfying $\widetilde{\rho}_{01} + \widetilde{\rho}_{01} \leq 1$. Additionally, observe that the case $\rho_{01} + \rho_{10} = 1$
is a degenerate one, under which the conditional distribution $P_{Y|U}(y | u)$ does not depend on $u$, meaning that $U$ and $Y$ are independent, so we cannot hope to recover useful information about the defective set from $Y$.

When considered to define a standard noisy communication channel, both the Z-channel and reverse Z-channel have Shannon capacity (in bits/use) given by \cite{tallini2002capacity}
\begin{equation} \label{eq:zcap}
    C_{\mathrm{Z}}(\rho) = \log_2 \left( 1 + (1-\rho) \rho^{\rho/(1-\rho)} \right),
\end{equation}
and the symmetric noise model has Shannon capacity (in bits/use) given by
\begin{equation} \label{eq:bsccap}
    C_{\mathrm{BSC}}(\rho) = 1 - h(\rho)
\end{equation}
where $h(\rho) = - \rho \log_2 \rho - (1-\rho) \log_2 (1-\rho)$ is the binary entropy in bits.
Since the Z-channel and reverse Z-channel have the same Shannon capacity, the information-theoretic results of \cite{Mal80,Sca15b} suggest that they may require the same asymptotic number of tests, at least for sufficiently sparse settings.  On the other hand, when adopting an NDD approach, it is unclear {\em a priori} which model requires more tests.  See Section \ref{sec:Z_vs_RZ} for further discussion.
 
Except where stated otherwise, we assume that the noise levels $\rho_{10}, \rho_{01}$ and number of defectives $k$ are known; our analysis can also be applied to cases where only bounds are known, but the details become more tedious.  Our main goal is to provide explicit achievable rates and converse bounds for noisy group testing under Bernoulli designs. % We describe a process as `i.i.d. Bernoulli testing with parameter $\nu$' if for each $i$ and $t$, the probability that item $i$ appears in test $t$ is $\nu/k$, independently of all other tests and items. 

\begin{rem} \label{rem:noise}
    While the general binary model \eqref{eq:gt_gen_model} captures several symmetric and non-symmetric noise models, there are  other noise models of interest that it does not capture.  As discussed in \cite[Sec.~3.1]{Ald19}, some noise models of interest depend on the {\em number of defectives} in the test, and not just the presence vs.~absence of any defectives.  A prominent example is dilution noise \cite{Ati12}, in which each defective item in the test is independently ``diluted'' (and hence behaves as though it were non-defective) with some probability $u$.  Hence, if there are $\ell$ defectives in the test, the probability they are all diluted is $u^{\ell}$. 
    While it may be possible to handle noise models of this kind using our techniques, the analysis appears to become significantly more complicated.  This is primarily due to the different conditional distribution of $Y$ for all different values of the number of defectives in the test, $\ell \in \{0,1,\dotsc,k\}$, in contrast with \eqref{eq:gt_gen_model} in which only need to distinguish between $\ell = 0$ and $\ell \ge 1$.  Due to these complications, we leave these further generalizations for future work.
    %
    %Noise models can also be {\em non-binary} in the sense of depending on the number of defectives in the test, with the dilution noise model being a common example \cite[Sec.~3.1]{Ald19}.  However, such models are beyond the scope of this paper.
\end{rem}

\subsection{Related Work}

The information-theoretic limits of noiseless and noisy non-adaptive group testing were initially studied in the Russian literature \cite{Mal78,Mal80}, and have recently become increasingly well-understood \cite{Ati12,Ald14,Sca15,Sca15b,Sca16b,Ald15}.  Among the existing works, the results most relevant to the present paper are as follows:
\begin{itemize}
    \item For both the adaptive and non-adaptive settings, it was shown by Baldassini {\em et al.} \cite{Bal13} that if the outcome $Y$ is produced by passing the noiseless outcome $U = \vee_{j \in S} X_j$ through a channel $P_{Y|U}$, then any group testing achieving $\pe \rightarrow 0$ the must have rate $R \leq C$, where $R$ is defined in \eqref{eq:defrate} and $C$ is the Shannon capacity of $P_{Y|U}$. Equivalently,
    the number of tests must satisfy $n \ge \big(\frac{1}{C}k\log_2 \frac{p}{k}\big)(1-o(1))$.   
    For instance, under the symmetric noise model \eqref{eq:gt_symm_model}, this yields
        \begin{equation}
            n \ge \frac{k\log_2 \frac{p}{k}}{1 - h(\rho)} (1-o(1)). \label{eq:mi_conv}
        \end{equation}
    It has recently been shown that under the RZ and symmetric noise models, this converse is not tight (i.e., it can be improved) when $\theta \in (0,1)$ is sufficiently close to one, even in the adaptive setting \cite{Sca18}.
    \item In the non-adaptive setting with symmetric noise, it was shown in \cite{Sca15,Sca15b} that an information-theoretic threshold decoder attains the bound \eqref{eq:mi_conv} when $k \asymp p^{\theta}$ for {\em sufficiently small} $\theta > 0$.  The analysis of \cite[Appendix A]{Sca18} shows that analogous findings also hold for the Z and RZ noise models.
\end{itemize}
Several non-adaptive noisy group testing algorithms have been shown to come with rigorous guarantees.
\begin{itemize}
    \item The {\em Noisy Combinatorial Orthogonal Matching Pursuit} (NCOMP) algorithm checks, for each item, the proportion of tests it was included in that returned positive, and declares the item to be defective if this number exceeds a suitably-chosen threshold.  This is known to provide optimal scaling laws for the regime $k \asymp p^{\theta}$ ($\theta \in (0,1)$) \cite{Cha11,Cha14}, albeit with somewhat suboptimal constants. That is, in the terminology of \eqref{eq:defrate}, NCOMP has a non-zero but suboptimal rate under symmetric noise. Similar results are also obtained for the general binary noise channel using a linear programming based algorithm in \cite[Thm.~7]{Cha14}.
    \item The method of {\em separate decoding of items}, also known as {\em separate testing of inputs} \cite{Mal80,Sca17b}, also considers the items separately, but uses all of the tests.  Specifically, a given item's status is selected via a binary hypothesis test.  This method was studied for $k = O(1)$ in \cite{Mal80}, and for $k \asymp p^{\theta}$ in \cite{Sca17b}.  In particular, it was shown that for the symmetric noise model, the number of tests is within a factor $\log 2$ of the optimal information-theoretic threshold as $\theta \to 0$.  However, the rate quickly become weaker as $\theta$ increases away from zero; see Appendix \ref{app:symmetric} for an example.
\end{itemize}
Since other works on noisy group testing are less related to the present paper, we only provide a brief outline.  Some heuristic algorithms have been proposed for noisy settings without theoretical guarantees, including belief propagation \cite{Sed10} and a noisy linear programming relaxation \cite{Mal12}.  Sublinear-time algorithms with guarantees on the number of samples and runtime have been proposed \cite{Cai13,Lee15a,Ina19,Bon19a} (see also the earlier works of \cite{Che13a,Ngo11,Ind10}), but the constants (and sometimes logarithmic factors) in the sample complexity bounds are far from optimal.  The complementary viewpoint of {\em adversarial noise} has also been explored \cite{Mac97,Ngo11,Che13a}.

\subsection{Applications of Noise Models} \label{sec:applications}

Group testing has been applied in a wide range of contexts, including biology, communications, information technology and data science, as outlined in \cite[Section 1.7]{Ald19}. While the noiseless model has been widely studied from a theoretical point of view, in many of these applications it is unrealistic to assume that all tests will return a perfectly accurate answer. Many papers have dealt with this issue by studying the symmetric noise model described in \eqref{eq:gt_symm_model}.

However, we believe that in many applications, this assumption of symmetry is itself also unrealistic. Since there are often different mechanisms operating that may ``flip'' positive tests to negative and vice versa, there is no {\em a priori} reason to believe that these two types of error should be equally likely. This argument motivates the general binary channel  model of Definition \ref{def:chanmod}.  In addition, to motivate the study of the RZ and Z channels in their own right, we proceed by giving examples where these models naturally arise.

% but we will briefly describe some application settings where the assumption of RZ or Z channel noise is  particularly natural, motivating the study of those channel models in their own right.

In the {\em file comparison} problem, we wish to carry out data forensics to determine which out of a collection of computer files have been changed.  One way to do this, described in \cite{goodrich2}, is to store a number of hashes of various concatenated collections of files. By comparing the hashes before and after any possible tampering, if the hash has changed, we know that at least one file in the group has been altered. This can be thought of as a group testing scenario:  An altered file corresponds to a defective item,
the collection of files corresponds to the testing pool, and a changed hash corresponds to a positive test. However, as discussed in \cite{madej}, it is possible % (if perhaps unlikely) 
that  the files may be altered in a way that does not change the value of the hash. In this sense, a test that should be positive may fail to be detected as such with a certain probability -- this is exactly the Z-channel of Definition \ref{def:chanmod}. While \cite{madej} argues that this effect can be minimized by taking arbitrarily long hashes, it may be that from an efficiency point of view it is preferable to store shorter hashes and take into account the effect of the Z channel noise in identifying the modified files. 

% {\bf [TODO: Option 1]} An application for which the RZ-channel provides a natural model is that of DNA testing; see\cite{du-hwang2,erlich2} for reviews of group testing in this context.  Roughly speaking, group testing offers an efficient way to test for rare genetic conditions, by identifying a substring of DNA that (in a certain position on a chromosome) is associated with it. A pool of individuals can simultaneously be screened by breaking their DNA into fragments, and checking whether the substring is present in any fragment. In an idealized sense, this acts like a group test, in that if any member of the pool has the condition, this substring will be present in the sample, and the test will be positive. However, since a particular target substring may also be found at a different position in the DNA, a positive test does not guarantee that a member of the population has the condition.  That is, we may have false positives in accordance with the RZ channel model, with a noise parameter depending on the length of the substring in question.

An application for which the RZ-channel can serve as a natural model is that of {\em multiple-access communication}, e.g., see \cite{wolf,berger,Ant11}.  Roughly speaking, group testing permits the detection of a small subset of active users by requesting each user to transmit a signal at the times corresponding to $1$'s in the group testing matrix.  Then, obtaining the group testing outcomes only requires detecting whether or not there is any signal present at each time instant.  If this detection procedure is reliable, then there should be no false negatives; however, in the presence of an {\em interfering signal}, one is prone to false positives, in agreement with the RZ channel model of Definition \ref{def:chanmod}.  Alternatively, if the detection procedure is not perfectly reliable, then we may be subject to the general binary noise model with suitably-chosen values of $\rho_{01}$ and $\rho_{10}$.

\section{Summary of Main Results}

As mentioned above, we will provide achievable rates for noisy group testing using noisy variants of the COMP and DD algorithms, as well as providing algorithm-independent converse bounds.  We summarize our main results in the following subsections.

\subsection{Highlights} \label{sec:highlights}

Since the statements of our main results are somewhat technical, we begin by highlighting some key special cases, focusing primarily on scenarios in which our bounds are tight in a certain sense.  The relevant rates are plotted in Figures \ref{fig:RZ_rates} and \ref{fig:Z_rates} for the RZ- and Z-channel models respectively, and in Figure \ref{fig:RatesBSC} in Appendix \ref{app:symmetric} for the symmetric noise model.  We have the following:
\begin{itemize}
    \item For the RZ-channel model, in Theorem \ref{thm:RZrate} we establish an achievable rate and an algorithm-independent converse (for Bernoulli testing) that match in broad scaling regimes, e.g., for all $\theta > 0.212$ in the case that $\rho = 0.1$.  As the noise level $\rho$ increases, the bounds match for a broader range of $\theta$.
    \item Also for the RZ-channel model,  in Appendix \ref{sec:high_noise_opt} we show that for $\theta$ close to zero and $\rho$ close to one (i.e., the noisy and sparse setting), the achievable rate is approximately $\frac{1-\rho}{e \log 2}$; according to a simple capacity-based converse \cite{Bal13}, this cannot be improved by any {\em arbitrary and possibly adaptive} algorithm.
    \item For the Z-channel model, in Theorem \ref{thm:Zrate} we establish an achievable rate and an algorithm-independent converse (for Bernoulli testing) that match in certain dense scaling regimes when $\rho$ is small enough, e.g., for all $\theta > 0.729$ in the case that $\rho = 0.001$.  However, as the noise level $\rho$ increases, the bounds in fact match for a narrower range of $\theta$ (and eventually for no $\theta$).
    \item For each noise model, we show that we recover the rate of \cite{Ald14a} for the noiseless setting in the limit as the noise parameters tend to zero; this rate is known to be tight for Bernoulli testing whenever $\theta > \frac{1}{2}$.
\end{itemize}

\subsection{Preliminaries}

{\bf Notation.} First, we establish some additional notation.
\begin{defn} \label{def:dgamma}
For any $\gamma > 0$, we define the function 
\begin{equation} \label{eq:ddef} D_\gamma(t) = t \log \left( \frac{t}{\gamma} \right) - t + \gamma\mbox{ \;\;\;\;\; for $t \geq 0$.} \end{equation}
\end{defn}
Note that $D_\gamma(\gamma) = D'_\gamma(\gamma) = 0$, and that $D''_\gamma(t) = 1/t \geq 0$, so $D_\gamma$ is convex.  Hence,
$D_\gamma(t) \geq 0$ for all $t \geq 0$, and $D_\gamma(t)$ is strictly increasing for $t > \gamma $. Note further that 
\begin{equation} \label{eq:handy} 
    D_a(t) = a D_1(t/a)  \mbox{ \;\;\;\; for all $a > 0$ and $t \geq 0$.} 
\end{equation}

\begin{figure}
    \begin{centering}
        \includegraphics[width=0.9\columnwidth]{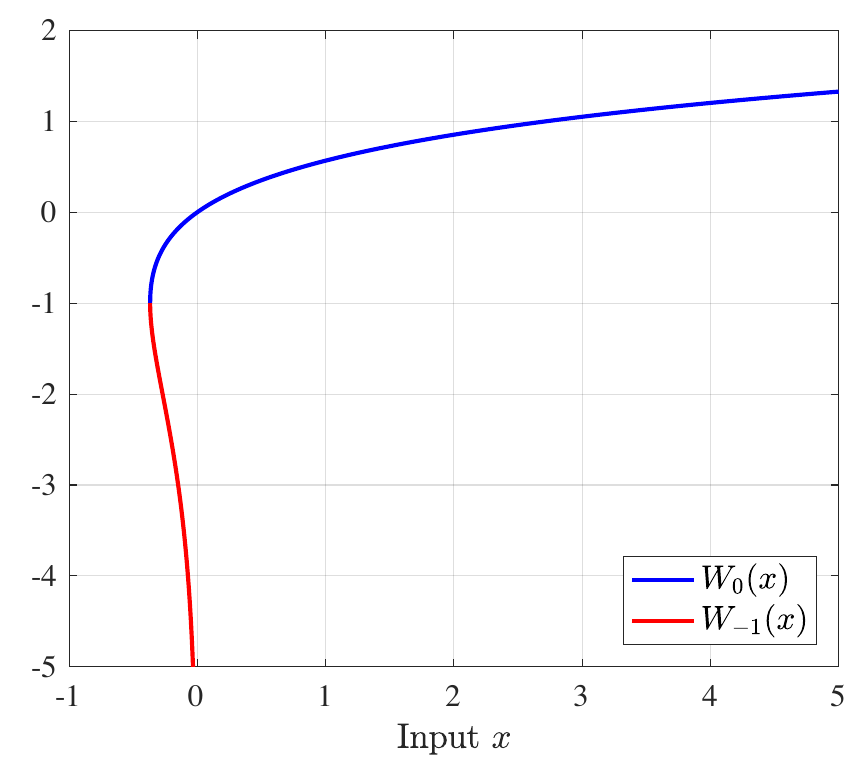}
        \par
    \end{centering}
    
    \caption{Upper and lower branches of the Lambert $W$-function. \label{fig:LambertW}}
\end{figure}

At several points in the paper, we require explicit values for the intersection of two functions related to the $D_{\gamma}(t)$ function (for different choices of $\gamma$). This intersection is found in Lemma \ref{lem:tech} of Appendix \ref{app:lemtech}, and can be expressed in terms of the Lambert $W$-function  (see for example \cite{corless}). This function gives the solution to the equation $W(x) e^{W(x)} = x$ for $x \geq -1/e$, and has two real branches at the point $(-1/e,-1)$; see Figure \ref{fig:LambertW}. We shall write $W_0$ and 
$W_{-1}$ respectively for the principal branch ($W_0(x) \geq -1$) and lower branch $(W_{-1}(x) \leq -1$). The key properties
of $W$ that we shall require are the derivative \cite[Eq. (3.2)]{corless}
\begin{equation} \label{eq:Wderivative} W'(x) = \frac{ W(x)}{x(1 + W(x))} \mbox{ \;\;for $x \notin \{ 0, -1/e \}$}
\end{equation}
(which holds on either branch), and the following asymptotic expansions, \cite[Section 4]{corless}:
\begin{gather} W_0(x) = \log x - \log \log x + o(1), ~~ x \to \infty \label{eq:Wlimit} \\
W_{-1}(x) = \log(-x) - \log(- \log (-x)) + o(1), ~~ x \to 0^{-}, \label{eq:Wlimit2} 
\end{gather}
where $x \to 0^{-}$ means approaching zero from below.
Some intuition behind these expansion is as follows: Direct calculation using the fact that $W(x) e^{W(x)} = x$ means we can deduce that when $W(x) \geq 1$,
we have $x \geq e^{W(x)}$, and hence $W_0(x) \leq \log x$ for $x \geq e$. Similarly, when $W(x) \leq - 1$, we have $e^{W(x)} \leq -x$, and hence $W_{-1}(x) \leq \log(-x)$ for all $x$.

{\bf Concentration Bounds.}
We consider $D_\gamma$ of Definition \ref{def:dgamma} because it naturally arises in tail bounds on binomial
random variables.  Specifically, we will use the following \cite[Ch.~4]{Mot10}: For $Z \sim \Binomial(N,q)$, we have that for any $\epsilon > 0$ that
\begin{align}
    \PP[ Z \le Nq(1-\epsilon) ] &\le \exp\Big( -Nq  D_1(1-\epsilon) % \big( (1-\epsilon) \log (1-\epsilon) + \epsilon\big) 
\Big), \label{eq:bino_le} \\
    \PP[ Z \ge Nq(1+\epsilon) ] &\le \exp\Big( -Nq D_1(1+\epsilon) %\big( (1+\epsilon) \log (1+\epsilon) - \epsilon\big) 
\Big). \label{eq:bino_ge}
\end{align}
The bounds of \eqref{eq:bino_le} and \eqref{eq:bino_ge} are asymptotically tight in certain regimes.  To establish this fact, we will make use of the following binomial coefficient bound \cite[Lemma 4.7.1]{ash}:
$$ \binom{N}{ \delta N} \ge \frac{1}{\sqrt{ 8 N \delta(1-\delta)}} \exp( N h_e(\delta)),$$
where $h_e$ is the binary entropy function in nats. Using this bound, we deduce that for given $\phi \in (0,1)$ such that $N q \phi$ is an integer, we have
\begin{align}
&\PP[ Z = N q \phi]\nonumber \\
& = \binom{N}{ N q \phi} q^{N q \phi} (1-q)^{N (1- q \phi)}  \\
& \ge \frac{1}{\sqrt{ 8 N q \phi(1- q \phi)}} \exp \bigg(  -  N q D_1(\phi)  \nonumber \\ 
    &\hspace*{0.5cm} + N\bigg( q (1-\phi) + (1- q \phi) \log \Big( \frac{ 1-q}{1-q \phi}  \Big) \bigg) \bigg) \label{eq:bino_sharp0}
\\
& \ge \frac{\exp \big(- N (1-\phi)^2 q^2/(1-q) \big)}{\sqrt{ 8 N q \phi(1- q \phi)}} \exp \left(  -  N q D_1(\phi) \right),
 \label{eq:bino_gesharp}
\end{align}
where \eqref{eq:bino_sharp0} uses $qD_1(\phi) = q\phi \log \phi + q(1-\phi)$, and \eqref{eq:bino_gesharp} follows from substituting the value $u = \frac{ 1-q}{1-q \phi}$ in the bound $\log u \geq 1 - \frac{1}{u}$ and rearranging.

%\footnote{More specifically, \cite[Lemma 4.7.2]{ash} gives a lower bound of 
 %    $$ \PP[ Z \ge Nq(1+\epsilon) ] \ge \frac{1}{\sqrt{ 8 N q (1+ \epsilon) (1 - q(1+ \epsilon)) }} \exp \left( - N D( q(1+\epsilon) \| q)
 %    \right),$$
 %    where $D(p \| q)$ is the relative entropy from a Bernoulli$(p)$ to a Bernoulli$(q)$ variable, and we then obtain \eqref{eq:bino_gesharp} using
 %    \begin{align*}  
  %   D( q(1+\epsilon) \| q) &= q (1+ \epsilon) \log(1 + \epsilon) + (1 - q(1+\epsilon)) \log \left( \frac{ 1 - q (1+\epsilon)}{1-q} \right) \\
 %    &\leq q (1+ \epsilon) \log(1 + \epsilon) - q \epsilon + \frac{\epsilon^2 q^2}{1-q}.
  %   \end{align*}}
%\begin{equation}
% \PP[ Z \ge Nq(1+\epsilon) ] \ge \frac{ \exp \big(- \frac{N \epsilon^2 q^2}{1-q} \big)}{\sqrt{ 8 N q (1+ \epsilon) (1 - q(1+ \epsilon)) }} \exp\Big( -Nq D_1(1+\epsilon)
% \big( (1+\epsilon) \log (1+\epsilon) - \epsilon\big)
% \Big). \label{eq:bino_gesharp}
%\end{equation}
We refer to the term preceding the exponential in \eqref{eq:bino_gesharp} as the ``sharpness factor''. Observe that if $\phi$ is constant and we have $N q = \big( a \log p \big)(1+o(1))$, $q = \frac{b}{k} (1+o(1))$, and $k \asymp p^{\theta}$ for some positive constants $a,b,\theta$, then for any $\epsilon' > 0$ this sharpness factor is lower bounded by $p^{-\epsilon'}$ for $p$ sufficiently large. By picking $\phi = 1 \pm \epsilon$, and bounding the tail by the respective point probability, we deduce that \eqref{eq:bino_le} and \eqref{eq:bino_ge} are each tight to within this sharpness factor.

\subsection{Reverse Z-Channel}

Our main result for the reverse Z-channel is written in terms of the following technical definitions.  First, we define 
\begin{equation} \label{eq:kappadef}
    \kappa = \kappa(\theta) :=  - W_{-1} \left( - e^{-1} \rho^{\theta/(1-\theta)} \right),
\end{equation}
where $W_{-1}$ denotes the lower branch of the Lambert $W$-function; that is, $\kappa > 0$ is a solution to the equation $\kappa e^{-\kappa} =   e^{-1} \rho^{\theta/(1-\theta)}$.  Moreover, we write
\begin{align}
    \thetacrit &:= \thetacrit(\rho) = 1 + \frac{ \rho \log \rho}{1-\rho} \label{eq:theta_crit}\\
    \thetaopt &:=\thetaopt(\rho) =  \frac{t(\rho)}{\log \rho + t(\rho)}, \label{eq:theta_opt}
\end{align}
where $t(\rho) := -\log(1-\rho) + \log( -\log(\rho)) + \frac{\log(\rho)}{1-\rho} + 1$.

%\begin{defn} \label{def:values}
%    (i) We define 
%    \begin{equation} \label{eq:kappadef}
%        \kappa = \kappa(\theta) :=  - W_{-1} \left( - e^{-1} \rho^{\theta/(1-\theta)} \right),
%    \end{equation}
%    where $W_{-1}$ denotes the lower branch of the Lambert $W$-function; that is, $\kappa > 0$ is a solution to the equation $\kappa e^{-\kappa} =   e^{-1} \rho^{\theta/(1-\theta)}$. 
%    
%    (ii) We write 
%    \begin{align}
%        \thetacrit &:= \thetacrit(\rho) = 1 + \frac{ \rho \log \rho}{1-\rho} \label{eq:theta_crit}\\
%        \thetaopt &:=\thetaopt(\rho) =  \frac{t(\rho)}{\log \rho + t(\rho)}, \label{eq:theta_opt}
%    \end{align}
%    where $t(\rho) := -\log(1-\rho) + \log( -\log(\rho)) + \log(\rho)/(1-\rho) + 1$.
%\end{defn}

\begin{thm} \label{thm:RZrate} {\em (Reverse Z-Channel)} For noisy group testing under reverse Z-channel noise with parameter $\rho \in (0,1)$, in the regime where $k \asymp p^{\theta}$ with $\theta \in (0,1)$, we have the following under Bernoulli testing:
\begin{enumerate}
\item {\bf [Achievability]} \label{it:RZrate1} 
Under the Bernoulli testing parameter $\nu = 1$, there exists a practical algorithm achieving error probability $\pe \rightarrow 0$ with rate
\begin{equation}
    \hspace*{-2ex} \Runder{RZ}(\theta, \rho) =  \left\{ \begin{array}{ll} 
     \frac{1-\rho}{e \log 2}, & \mbox{ \;\;\;\;  $\theta \leq \thetaopt$,}  \\
      \frac{-\log \rho}{\kappa(\theta) e \log 2}, & \mbox{ \;\;\;\; $\thetaopt \leq \theta \leq \thetacrit$,}  \\
     \frac{(1-\theta)(1- \rho)}{e \log 2}, & \mbox{ \;\;\;\; $  \theta \ge \thetacrit$.}  \\
    \end{array} \right. \label{eq:RZ_ach}
\end{equation}
\item {\bf [Converse]} \label{it:RZrate2}  Under any Bernoulli testing parameter $\nu > 0$, if $\rho < 1/2$, then no algorithm can achieve $\pe \rightarrow 0$ with a rate higher than
\begin{align} \label{eq:RZconverse}
    &\Rbar{RZ}(\theta, \rho) \nonumber \\ &=  \left\{ \begin{array}{ll} 
     \min \left\{ C_Z(\rho), \frac{-\log \rho}{\kappa(\theta) e \log 2} \right\}, & \mbox{$\theta \leq \thetacrit$,}  \\
     \frac{(1-\theta)(1- \rho)}{e \log 2}, & \mbox{$\theta \ge \thetacrit$.} \\
    \end{array} \right. 
\end{align}
where $C_Z(\rho)$ is the Shannon capacity of the channel, given by \eqref{eq:zcap}.
\end{enumerate}
\end{thm}
\begin{proof}
    See Section \ref{sec:RZ_proofs}.
\end{proof}

\begin{figure}
    \begin{centering}
        \includegraphics[width=0.95\columnwidth]{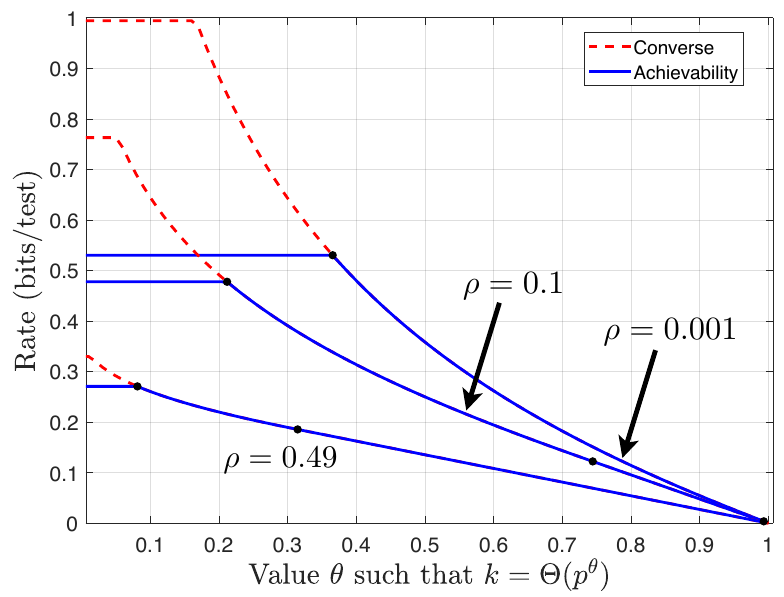}
        \par
    \end{centering}
    
    \caption{Achievable rate and algorithm-independent lower bound under reverse Z-channel noise and Bernoulli testing.  The dots indicate the thresholds $\thetaopt$ and $\thetacrit$.} \label{fig:RZ_rates}
\end{figure}

These rates are illustrated for three different noise levels in Figure \ref{fig:RZ_rates}. 

\begin{rem} \mbox{ } 
\begin{enumerate}
    \item As in the noiseless case, for sufficiently dense problems (i.e., for $\theta \geq \thetaopt$),  we obtain a sharp result, with the achievable and converse rates coinciding.  In this case, this optimal performance is achieved by either a noisy version of the DD algorithm  (for $\thetaopt \leq \theta \leq \thetacrit$) or a trivial extension of the COMP algorithm (for  $\theta \geq \thetacrit$).  See Section \ref{sec:RZ_proofs} for details.
    \item As $\rho \rightarrow 0$, we have $\thetaopt(\rho) \rightarrow 1/2$, and we recover the fact that the optimal performance is achieved by practical algorithms for $\theta \geq 1/2$ \cite{Ald14a}.  However, as soon as we increase the noise level even slightly, the achievability and converse match over a noticeably wider parameter range. For example, for $\rho = 0.001$, this is the case for 
    $\theta \geq \thetaopt(\rho) = 0.3656$, and for $\rho = 0.1$ this widens to $\theta \geq \thetaopt(\rho) = 0.2119$.
    \item As $\rho \rightarrow 0$, we have $\thetacrit \rightarrow 1$, and so for any fixed $\theta$ the final case in \eqref{eq:RZ_ach} does not apply in this limit. Furthermore, as $\rho \rightarrow 0$, \eqref{eq:Wlimit2} gives $\frac{-\log \rho}{\kappa(\theta)} \rightarrow \frac{1-\theta}{\theta}$, and we recover the noiseless results \eqref{eq:ddnoiselessrate}--\eqref{eq:ddnoiselesscap}.
    \item Both the achievable rate and converse provide a curve which is continuous in $\theta$. We can establish continuity at $\thetacrit$ using the fact that $\kappa(\thetacrit) = \frac{1}{\rho}$. This, in turn, follows because we can verify that $\rho^{\thetacrit/(\thetacrit-1)} e = \rho e^{1/\rho}$ (as \eqref{eq:theta_crit} gives $\frac{\thetacrit}{\thetacrit-1} = \frac{1-\rho}{\rho \log \rho} + 1$ and in addition we have $\rho^{\frac{1-\rho}{\rho\log\rho}} = e^{1/\rho - 1}$), and the choice $\kappa = \frac{1}{\rho}$ makes $\kappa e^{-\kappa} =  \big( \rho e^{1/\rho} \big)^{-1}$ in agreement with the definition $\kappa e^{-\kappa} = e^{-1} \rho^{\theta/(1-\theta)}$.
    % \item While a noisy variant of DD is used to prove the achievability part, in the regime $\theta \geq \thetacrit$ the same rate can be achieved via a COMP approach.  In fact, since items in a negative test are definitely negative, one can  follow the analysis of \cite{Ald14a} with almost no modification, and achieve the noiseless rate multiplied by $1-\rho$, since this is the fraction of negative tests that are flipped to positive.
\end{enumerate}
\end{rem}

\noindent In the appendices, we provide two further claims pertaining to converse results under RZ noise:
\begin{enumerate}
    \item (Appendix \ref{sec:dd_spec_conv}) When the noisy DD algorithm is used in conjunction with Bernoulli testing, no rate higher than $\frac{1-\rho}{e \log 2}$ can be achieved.  Therefore, one cannot hope to improve on the first case in \eqref{eq:RZ_ach} (nor on the other cases where an algorithm-independent converse holds) without moving to a different test design and/or a different decoding algorithm.  To our knowledge, this result is new even when specialized to the noiseless case.
    \item (Appendix \ref{sec:high_noise_opt}) In the limit $\theta \to 0$, the achievable rate approaches $\frac{1-\rho}{e \log 2}$, which is also the first-order term in the Z-channel capacity as $\rho \to 1$.  Therefore, under the order of limits $n \to \infty$, $\theta \to 0$, and then $\rho \to 1$, the limiting behavior of the noisy DD rate cannot be improved on even by an adaptive algorithm (since the capacity-based converse holds even for adaptive algorithms \cite{Bal13}).  We support this claim with the rate plot in Figure \ref{fig:VaryingRhoRZ} for $\theta = 0$, where we observe nearly tight bounds for $\rho$ close to one.
\end{enumerate}

\begin{figure}
    \begin{centering}
        \includegraphics[width=0.95\columnwidth]{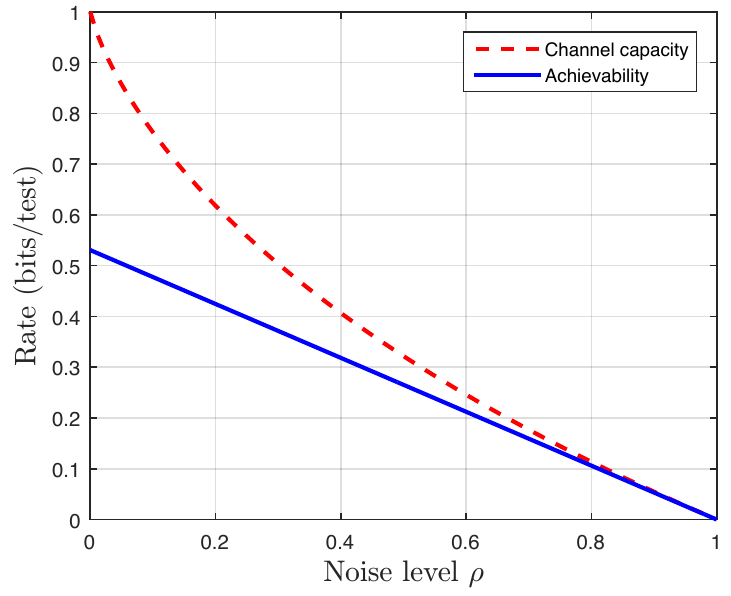}
        \par
    \end{centering}
    
    \caption{Reverse Z-channel model: Limiting achievable rate and algorithm-independent converse as $\theta \to 0$, plotted as a function of the noise level $\rho$.} \label{fig:VaryingRhoRZ}
\end{figure}

\subsection{Z-Channel}

Our rates for the Z-channel are written in terms of the following.  For given values of $\nu$ and $\rho$, define $s = (1-\rho) e^{-\nu}/\rho > 0$ and for $\theta \ne \frac{1}{2}$ write the ratio
\begin{equation} \label{eq:gratio} g(s, \theta) =
    \frac{ 1 + \frac{\theta}{2 \theta -1} s}{ (1+s)^{\theta/(2 \theta-1)}},
\end{equation}
and introduce
\begin{equation} \label{eq:lambdadef}
    \lambda := \lambda(\theta) = \left\{
    \begin{array}{ll}  
    W_{0} \left( - e^{-1} g(s, \theta) \right), & \mbox{ \;\;\;\;\; for $\theta < \frac{1}{2}$,} \\
    W_{-1} \left(  - e^{-1} g(s, \theta) \right), & \mbox{ \;\;\;\;\; for $\theta > \frac{1}{2}$.} \\
    \end{array} \right.
\end{equation}
Moreover, define 
\begin{equation} \label{eq:alphadef}
    \begin{array}{ll} \alpha^*(\theta) = \rho \frac{s}{\log(1+s)} & \mbox{ \;\;\;\;\; for $\theta = \frac{1}{2}$,} \\
    \alpha^*(\theta) = - \frac{\rho}{\lambda(\theta)} \big( 1 +  \frac{\theta s}{2 \theta - 1} \big)    & \mbox{ \;\;\;\;\; for $\theta \neq  \frac{1}{2}$.} \\ \end{array} 
\end{equation}

%\begin{defn} 
%(i) For given values of $\nu$ and $\rho$,
%define $s = (1-\rho) e^{-\nu}/\rho > 0$ and write the ratio
%\begin{equation} \label{eq:gratio} g(s, \theta) =
%\frac{ 1 + \frac{\theta}{2 \theta -1} s}{ (1+s)^{\theta/(2 \theta-1)}}.
%\end{equation}
%
%(ii) Then introduce
%\begin{equation} \label{eq:lambdadef}
% \lambda := \lambda(\theta) = \left\{
%\begin{array}{ll}  
%W_{0} \left( - e^{-1} g(s, \theta) \right), & \mbox{ \;\;\;\;\; for $\theta < 1/2$,} \\
%W_{-1} \left(  - e^{-1} g(s, \theta) \right), & \mbox{ \;\;\;\;\; for $\theta > 1/2$.} \\
%\end{array} \right.
%\end{equation}
%(iii) Finally, define 
%\begin{equation} \label{eq:alphadef}
%\begin{array}{ll} \alpha^*(\theta) = \rho s/\log(1+s)   & \mbox{ \;\;\;\;\; for $\theta = 1/2$,} \\
%\alpha^*(\theta) = - \rho \left( 1 +  \frac{\theta s}{2 \theta - 1} \right)/\lambda(\theta)    & \mbox{ \;\;\;\;\; for $\theta \neq  1/2$.} \\ \end{array} \end{equation}
%\end{defn}

In the two cases $\theta < \frac{1}{2}$ and $\theta > \frac{1}{2}$, the value of $\frac{\theta}{2 \theta - 1}$ is below $0$ and above $1$, respectively. Since these are the two ranges of values of $r$ for which the Bernoulli inequality $(1+ r s) \leq (1+s)^r$ holds, in each case we can deduce that
$g(s, \theta) \leq 1$, which implies that the Lambert functions in \eqref{eq:lambdadef} are well-defined (see Figure \ref{fig:LambertW}).

\begin{thm} \label{thm:Zrate} {\em (Z-Channel)}  For noisy group testing under Z-channel noise with parameter $\rho$, in the regime where $k \asymp p^{\theta}$ with $\theta \in (0,1)$, we have the following under Bernoulli testing with parameter $\nu > 0$:
\begin{enumerate}
\item {\bf [Achieveable rate]} \label{it:Zrate1} 
There exists a practical algorithm with error probability $\pe \rightarrow 0$ with
\begin{align} 
    &\hspace*{-1.2ex}\Runder{Z}(\theta, \rho) = \frac{(1-\theta)(1- \rho) \nu e^{-\nu}}{\theta \log 2 } \nonumber \\
        & \times \min \bigg\{  \frac{\theta}{2 \theta-1} \left(   \frac{\alpha^*(\theta)}{\alpha^*(1/2)} - 1 \right), 1
        \bigg\} ~\mbox{ if }~ \theta \neq \frac{1}{2}, \tag{32a} \label{eq:Zratedoable1}
\end{align}
and
\begin{align} 
    &\Runder{Z}(\theta, \rho) \nonumber \\
     &= \frac{(1- \rho) \nu e^{-\nu}}{ \log 2 }  \min \bigg\{ \frac{1}{\log(1+s)} \log  \frac{ s}{\log(1+s)} \nonumber \\    
        & \qquad - \frac{1}{\log(1+s)} + \frac{1}{s}, 1
        \bigg\} \mbox{ ~ if~ } \theta = \frac{1}{2}.  \tag{32b} \label{eq:Zratedoable2}
\end{align}
\setcounter{equation}{32}
%\begin{align} 
%    \Runder{Z}(\theta, \rho) = 
%    \begin{cases}
%        \frac{(1-\theta)(1- \rho) \nu e^{-\nu}}{\theta \log 2 } 
%        \min \left\{  \frac{\theta}{2 \theta-1} \left(   \frac{\alpha^*(\theta)}{\alpha^*(1/2)} - 1 \right), 1
%        \right\} & \theta \neq \frac{1}{2} \\
%        \frac{(1- \rho) \nu e^{-\nu}}{ \log 2 } 
%        \min \left\{ \frac{1}{\log(1+s)} \log  \frac{ s}{\log(1+s)}  - \frac{1}{\log(1+s)} + \frac{1}{s}, 1
%        \right\} & \theta = \frac{1}{2}.
%    \end{cases} \label{eq:Zratedoable}
%\end{align}
\item {\bf [Converse]} \label{it:Zrate2}  If $\rho < 1/2$, then no algorithm can achieve $\pe \rightarrow 0$ with a rate higher than
\begin{equation} \label{eq:Zrateconverse}
    \Rbar{Z}(\theta, \rho) =  \min \left\{C_Z(\rho), \frac{ (1-\theta)(1-\rho)}{\theta e \log 2} \right\}.
\end{equation}
\end{enumerate}
\end{thm}
\begin{proof}
    See Section \ref{sec:Z_proofs}.
\end{proof}

\begin{figure}
    \begin{centering}
        \includegraphics[width=0.95\columnwidth]{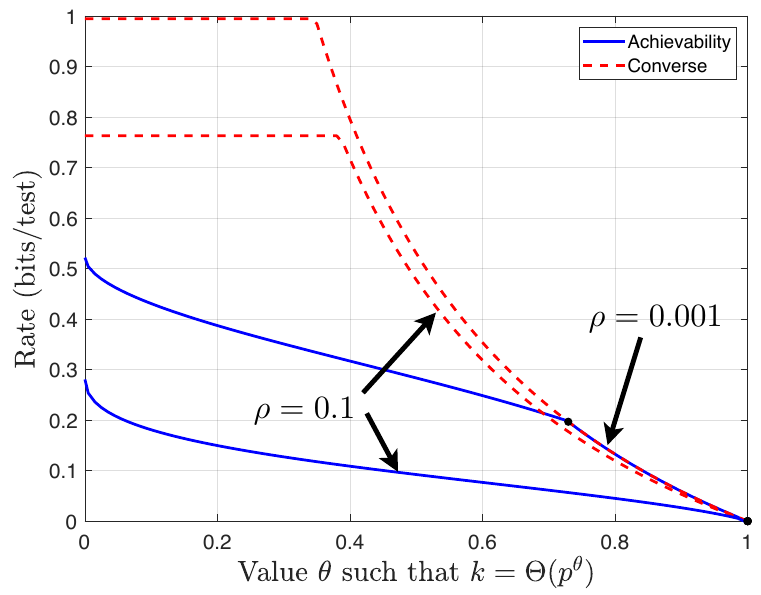}
        \par
    \end{centering}
    
    \caption{Achievable rate and algorithm-independent lower bound under Z-channel noise and Bernoulli testing.  The dots indicate the threshold $\thetacritZ$ (see Remark \ref{rem:Zratedetailsshort}).  \label{fig:Z_rates}}
\end{figure}

These rates are illustrated for two different noise levels in Figure \ref{fig:Z_rates}.  The noisy DD algorithm used for the achievability part is described in Section \ref{sec:Z_proofs}.  We proceed by giving some properties of these rates; see Remark \ref{rem:Zratedetails} in the proof for more details.

\bigskip
\begin{rem} \label{rem:Zratedetailsshort}  \mbox{ }
\begin{enumerate}
\item It can be shown that both the achievable and converse rates are continuous and non-increasing in $\theta$.  
\item There exists some $\thetacritZ$ (which may be equal to $1$) such that the first term achieves the $\min\{\cdot,1\}$ in (32) if and only if $\theta \leq \thetacritZ$.  For $\nu = 1$ and $\theta > \thetacritZ$, the achievability and converse bounds coincide.  As we see in Figure \ref{fig:Z_rates}, this is indeed observed for small $\rho$; however, even for $\rho = 0.1$, the bounds do not match for any value of $\theta < 1$, due to the fact that $\thetacritZ = 1$.
%\item For given $s$, the function $\alpha^*(\theta)$ is continuous at $\theta = 1/2$ and
% is increasing in $\theta$ for all $\theta$.
%\item It can be shown that $\frac{\theta}{2 \theta-1} \left(   \frac{\alpha^*(\theta)}{\alpha^*(1/2)} - 1 \right)$ is increasing in $\theta$, and so the first term in the minimum applies for all $\theta \leq \thetacritZ$, where $\thetacritZ$ may equal 1.
\item As $\theta \rightarrow 1$, we have $g(s, \theta) \rightarrow 1$, so that $\lambda(\theta) \rightarrow -1$ and 
$\alpha^*(\theta) \rightarrow \rho(1+s)$.  The implies that the first term inside the minimum in \eqref{eq:Zratedoable1} converges
to $ ((1+s) \log(1+s) - s)/s$, which is less than $1$ if and only if $s > 3.922$. In the case $\nu = 1$, this corresponds to the
fact that the first term in (32) always gives the minimum (i.e., $\thetacritZ = 1$) for $\rho < 0.0858$.
\item As $\rho \rightarrow 0$, the converse result \eqref{eq:Zrateconverse} clearly tends to the noiseless converse result of
\eqref{eq:ddnoiselesscap}. The corresponding argument for the achievability rate is more delicate. However,
as described in Remark \ref{rem:Zratedetails} in the proof, we indeed recover the noiseless achievable rate  of \eqref{eq:ddnoiselessrate} as $\rho \rightarrow 0$.
\item The achievability result includes the Bernoulli testing parameter $\nu > 0$, which can be optimized. Since $s$ depends on $\nu$, there appears to be no closed form expression for the optimal choice. However, our numerical findings suggest that the value $\nu=1$  is near-optimal, particularly for small $\rho$.
Comparing the achievability and converse parts, we see that $\nu = 1$ is certainly optimal when $\theta \geq \thetacritZ$.
\end{enumerate}
\end{rem}

\subsection{Comparison of the Channels} \label{sec:Z_vs_RZ}

In Figure \ref{fig:both_rates}, we compare the rates for the Z and reverse Z-channel noise models.  For $\theta$ close to one, we observe that the former is provably easier to handle: The Z achievability curve lies above the RZ converse curve.  On the other hand, RZ noise appears to be easier to handle for small to moderate $\theta$ (though we cannot say this conclusively, as we have not yet verified whether the Z achievability curve is the best possible).

Some intuition behind this behavior can be obtained by noting that, like the noiseless version, the noisy DD algorithms first use negative tests to find a set of ``possible defectives'' (PD), and then estimate the defective set based on positive tests containing a single PD.  The rate turns out to be dictated by the first step for small $\theta$, and by the second step for large $\theta$.   Given that this is the case, the behavior in Figure \ref{fig:both_rates} is to be expected:
\begin{itemize}
    \item The first step is easier under RZ noise, since negative test outcomes are perfectly reliable.
    \item The second step is easier under Z noise, since positive test outcomes are perfectly reliable.
\end{itemize}
The fact that Z noise is preferable for $\theta$ close to one was also observed in the {\em adaptive} setting in \cite{Sca18}.

\begin{figure}
    \begin{centering}
        \includegraphics[width=0.95\columnwidth]{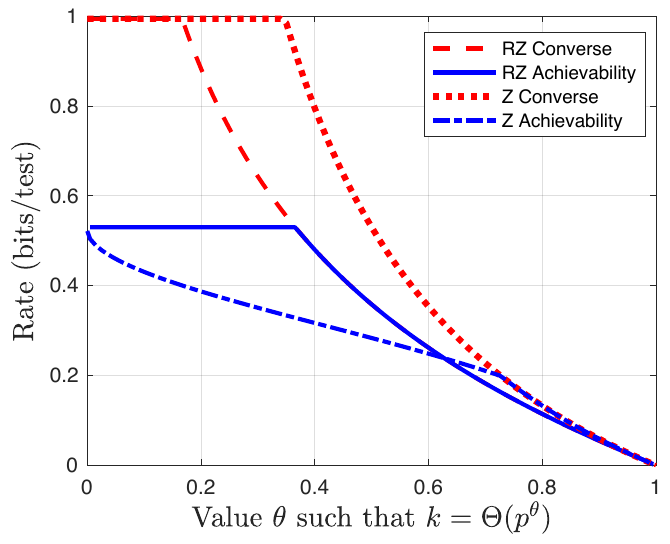}
        \par
    \end{centering}
    
    \caption{Comparison of rates under Z and reverse Z-channel noise. \label{fig:both_rates}}
\end{figure}

\subsection{General Binary Noise and Symmetric Noise}

Our techniques can be extended to general binary channels ({\em cf.}, \eqref{eq:gt_gen_model}), including the widely-considered symmetric model ({\em cf.}, \eqref{eq:gt_symm_model}).  However, for such channels, we have not yet proved a matching achievability and converse in any regime, other than the low-noise limit.  We therefore defer our results on these models to Appendix \ref{app:symmetric}.

\section{Proof of Theorem \ref{thm:RZrate} (Reverse Z-Channel)} \label{sec:RZ_proofs}

\subsection{Achievability via Noisy DD and COMP}

Recall the reverse Z-channel of Definition \ref{def:chanmod}. We first describe a noisy DD algorithm for this model, exploiting the fact that if an item appears in a test with $Y=0$ then we can be certain that it is non-defective. \\

\noindent\fbox{
    \parbox{0.95\columnwidth}{
        \textbf{Noisy DD algorithm for RZ channel noise:}
        \begin{enumerate}
            \item For each $j \in [p]$, let $\Nnegj$ be the number of negative tests in which item $j$ is included.  In the first step, we construct the following set of items that are definitely non-defective:
            \begin{equation}
            \NDhat = \bigg\{ j \in [p] \,:\, \Nnegj > 0 \bigg\}. \label{eq:NDhat_RZ}
            \end{equation}
            The remaining items, $\PDhat = [p] \setminus \NDhat$, are called ``possibly defective''.
            \item For each $j \in \PDhat$, let $\Npposj$ be the number of positive tests that include item $j$ and no other item from $\PDhat$.  In the second step, we fix a constant $\beta \in (\rho,1)$, and estimate the defective set as follows:
            \begin{equation}
            \Shat = \bigg\{ j \in \PDhat \,:\, \Npposj \ge \frac{\beta n \nu e^{-\nu}}{k} \bigg\}.  \label{eq:Shat_RZ}
            \end{equation}
        \end{enumerate}
    }
} \medskip

In addition to this noisy variation of DD, we can directly apply the noiseless COMP algorithm (e.g., see \cite{Cha11}) to the reverse Z-channel model.  As in the noiseless case, we know that even if reverse Z-channel noise is present, an item appearing in a negative test is definitive proof that it is not defective.  Therefore, we can consider taking $\PDhat$ above to be the estimate of $S$.  We analyze the performance of this algorithm using essentially the same argument as in \cite{Cha11}. We wish to ensure that each non-defective item has some test where it is tested with no defective present, and that the resulting test outcome is not changed by the reverse Z-channel (we say that the non-defective item {\em survives} this test).

\begin{lem} \label{lem:COMP_RZ}
    {\em (COMP under RZ noise)}
    Consider the reverse Z-channel noisy group testing setup with parameter $\rho \in (0,1)$, number of defectives $k \asymp p^{\theta}$ (where $\theta \in (0,1)$), and i.i.d.~Bernoulli testing with parameter $\nu > 0$. Under the COMP algorithm, we have $\pe \to 0$ as long as $n \geq \ncomp (1+\eta)$ for arbitrarily small $\eta > 0$, where 
    \begin{equation} 
       \ncomp =  \frac{1}{(1-\rho) \nu e^{-\nu}}  k \log p. \label{eq:ncomp_bound_RZ}
    \end{equation}
\end{lem}
\begin{proof}  
The probability that any particular non-defective item (is included in and) survives any particular test is $(1-\rho) \nu/k (1- \nu/k)^k$.
Using $n = \gamma k \log p$ tests, we can control $\pe$ using the union bound by
{\allowdisplaybreaks
\begin{align}
\pe & \leq \pr \left( \bigcup_{j \in S^c} \{ \mbox{item $j$ doesn't survive any test} \} \right) \\
    &\leq  p \pr \left( \mbox{item $j$ doesn't survive any test} \right) \\
    &= p \bigg( 1 -  (1-\rho) \frac{\nu}{k} \Big(1- \frac{\nu}{k}\Big)^k \bigg)^n \\
    &\leq  p \exp \left( \left(  - (1-\rho) \nu (1- \nu/k)^k \gamma \right) \log p \right) \label{eq:COMP_Step4} \\
    &\leq  \exp \left( \left( 1 - (1-\rho) \nu (1- \nu/k)^k \gamma \right) \log p \right),
\end{align}}
where \eqref{eq:COMP_Step4} uses $1-z \le e^{-z}$.  Since $\big(1-\frac{\nu}{k}\big)^k = e^{-\nu}(1+o(1))$, we deduce that if $n \geq \ncomp(1+\eta)$ (or equivalently
$  (1- \rho) \nu e^{-\nu} \gamma \ge 1+\eta$), then $(1-\rho) \nu (1- \nu/k)^k \gamma \ge 1 + \eta/2$ for $k$ sufficiently large,
so $\pe \leq p^{-\eta/2}$, which converges to 0. 
\end{proof}

The main step towards proving the achievability part of Theorem \ref{thm:RZrate} is to establish the following.

\begin{thm} \label{thm:ndd_RZ}
    {\em (NDD under RZ noise)}
    Consider the reverse Z-channel noisy group testing setup with parameter $\rho \in (0,1)$, number of defectives $k \asymp p^{\theta}$ (where $\theta \in (0,1)$), and i.i.d.~Bernoulli testing with parameter $\nu > 0$.  For any $\beta \in (\rho,1)$ and $\xi \in (0,\theta)$, the noisy DD algorithm achieves $\pe \to 0$  as long as $n \ge \ndd(1+\eta)$ for arbitrarily small $\eta > 0$, where
    \begin{equation}
    \ndd = \max\big\{ \nind, \niid, \niind \big\}, \label{eq:ndd_bound_RZ}
    \end{equation}
    and where
    \begin{align}
        \nind &= \frac{1-\xi}{(1-\rho) \nu e^{-\nu}} \cdot k \log p,  \label{eq:nind_RZ} \\
        \niid &= \frac{1}{\nu e^{-\nu} D_1(\beta)} %\cdot \frac{1}{(1-\epsiid) \log (1-\epsiid) + \epsiid} 
\cdot k \log k,  \label{eq:niid_RZ} \\
        \niind &= \frac{\xi}{\nu e^{-\nu} \rho D_1(\beta/\rho)} %\cdot \frac{1}{(1+\epsiind) \log (1+\epsiind) - \epsiind} 
\cdot k \log p.  \label{eq:niind_RZ}
    \end{align}
 %   and 
%    \begin{gather}
%        \epsiid = 1-\beta, \quad
%        \epsiind = \frac{\beta - \rho}{\rho}. 
%    \end{gather}
\end{thm}

Once this result and Lemma \ref{lem:COMP_RZ} are in place, proving the achievability of the rate \eqref{eq:RZ_ach} boils down to algebraic manipulations.  Since these are somewhat tedious, they are deferred to Appendix \ref{app:manip_ach_RZ}.  In the remainder of this subsection, we focus on the proof of Theorem \ref{thm:ndd_RZ}.

Let $\Nneg$ and $\Npos$ respectively denote the number of negative and positive tests.  Both of these quantities follow a binomial distribution; the probability of a given test $i$ being positive is given by
\begin{align}
\PP[ Y_i = 1 ] 
&= \bigg(1 - \Big( 1 - \frac{\nu}{k} \Big)^k \bigg) + \rho\Big( 1 - \frac{\nu}{k} \Big)^k \\
&= \Big( (1-e^{-\nu}) + \rho e^{-\nu} \Big) (1+o(1)),
\end{align}
since $\big( 1 - \frac{\nu}{k} \big)^k \to e^{-\nu}$ as $k \to \infty$, and we similarly have
\begin{align}
\PP[ Y_i = 0 ] =  (1-\rho) e^{-\nu} \cdot (1+o(1)). \label{eq:p_neg}
\end{align}
Hence, and using \eqref{eq:bino_le}--\eqref{eq:bino_ge}, we have with probability approaching one that
\begin{align}
\Nneg &= n \cdot (1-\rho) e^{-\nu} \cdot (1+o(1)), \label{eq:n0_conc_RZ} \\
\Npos &= n \cdot \Big( (1-e^{-\nu}) + \rho e^{-\nu} \Big) \cdot (1+o(1)). \label{eq:n1_conc_RZ}
\end{align}
It will be useful to split the $\Npos$ positive tests into two types: those that contain a defective item and whose noise does not flip the outcome, and those that contain no defective items but whose noise flips the outcome.  The number of such tests are denoted by $\Npd$ and $\Npnd$, respectively.  A given test falls into the first category with probability $(1-e^{-\nu}) (1+o(1))$, and the second category with probability $\rho e^{-\nu}(1+o(1))$. Therefore, the concentration bounds \eqref{eq:bino_le}--\eqref{eq:bino_ge} yield
\begin{align}
\Npd &= n \cdot (1-e^{-\nu}) \cdot (1+o(1)), \label{eq:n1d_conc_RZ} \\
\Npnd &= n \cdot \rho e^{-\nu} \cdot  (1+o(1)). \label{eq:n1nd_conc_RZ}
\end{align}
with probability approaching one.

Throughout the remainder of the analysis, we implicitly condition on the defective set $S$ taking a fixed value, say $S = \{1,\dotsc,k\}$.  By the symmetry of the random test design, the conditional error probability is the same for any such realization of cardinality $k$.

{\bf Analysis of First Step.}
Since negative tests can never include a defective item for the reverse Z-channel, the set $\PDhat = [p] \setminus \NDhat$ contains all of the defective items.  We proceed by establishing a sufficient condition such that with high probability, it also contains at most $p^\xi$ non-defectives, for some constant $\xi \in (0,\theta)$.   We denote the complement of this event by $\peind$, where the subscript and superscript respectively denote the step number and the consideration of non-defective items. 

{\em Analysis of non-defective items.} Let $\peind(\nneg)$ be defined similarly to $\peind$, but conditioned on $\Nneg$ taking a given value $\nneg$.  It suffices to establish that  $\peind(\nneg) \to 0$ for all $\nneg$ satisfying the concentration bound \eqref{eq:n0_conc_RZ}, since this bound holds with probability approaching one.

Since the test outcomes depend only on the defective items, one can envision the non-defective items as being placed in each test with probability $\frac{\nu}{k}$ {\em after} the test outcomes have been produced.  As a result, given $\Nneg = \nneg$, the number of negative tests $\Nnegj$ including a given item $j \notin S$ is distributed as $(\Nnegj | \nneg) \sim \Binomial\big(\nneg,\frac{\nu}{k}\big)$.  Hence, for any $j \notin S$, we have
\begin{align}
\PP[j \notin \NDhat \,|\, \nneg] 
&= \PP\big[ \Nnegj = 0 \,\Big|\, \nneg \big] \\
&= \bigg(1 - \frac{\nu}{k}\bigg)^{\nneg} \\
&\le e^{-\frac{\nneg \nu}{k}}, \label{eq:pe1nd_4_RZ}
\end{align}
where we have applied $1 - \zeta \le e^{-\zeta}$.  As a result, letting $G$ denote the number of non-defective items in $\PDhat = [p] \setminus \NDhat$, we find that $\EE[G \,|\,\nneg]$ is upper bounded by $p-k$ times the right-hand side of \eqref{eq:pe1nd_4_RZ}.  Applying Markov's inequality, we obtain
\begin{equation}
\PP[G \ge p^\xi \,|\,\nneg] \le p^{1-\xi} e^{-\frac{\nneg \nu}{k}}.
\end{equation}
Since $\nneg$ satisfies \eqref{eq:n0_conc_RZ}, we find that we can achieve $\peind \to 0$ under the condition
\begin{equation}
n \ge \bigg(\frac{1-\xi}{(1-\rho)\nu e^{-\nu}} \cdot k \log p \bigg) (1+o(1)). \label{eq:n_final_1nd_RZ}
\end{equation}

{\bf Analysis of Second Step.} Recall that the final estimate includes all $j \in \PDhat$ for which $\Npposj \ge \frac{\beta n \nu e^{-\nu}}{k}$, where $\Npposj$ is the number of tests containing item $j$ and no other item from $\PDhat$.  To characterize the distribution of $\Npposj$, we make use of a multinomial conditioning argument analogous to that used in the noiseless DD algorithm \cite{Ald14a}.  Since this is used multiple times throughout the paper, we first state the relevant result in generic notation.

\begin{lem} \label{lem:multi_cond}
    {\em \cite[Lemma C.1]{Ald14a}}
    Fix a positive integer $m$, and let $(W_0,W_1,W_2)$ have a multinomial distribution with $m$ trials and probabilities $(r_0,r_1,r_2)$.  Associate an observation $(W_0,W_1,W_2) = (w_0,w_1,w_2)$ with an unordered list of $m$ class labels (class $0$, $1$, or $2$), and suppose that each label in class $1$ is independently changed to some sub-class $1'$ with probability $\gamma \in [0,1]$, and to some sub-class $1''$ with probability $1-\gamma$ (where $\gamma$ may depend on $w_0$).  Then, conditioned on $W_0 = w_0$, the corresponding random variables $(W'_1,W''_1,W_2)$ counting the transformed class labels have a multinomial distribution with $m-w_0$ trials and the following probability parameters:
    \begin{equation}
        \bigg( \frac{r_1 \gamma}{1-r_0}, \frac{r_1 (1-\gamma)}{1-r_0}, \frac{r_2}{1-r_0} \bigg).
    \end{equation}
\end{lem}

%Using this lemma, we deduce the the following corollary for group testing with a general noise distribution $P_{Y|U}$ (see Section \ref{sec:setup}). 
%
%\begin{cor}
%    In the noisy group testing problem, consider classifying the $n$ tests into three groups of size $(\Nneg,N_1,N_2)$, where $\Nneg$ is the number of negative tests, $N_1$ is the number of tests satisfying a fixed logical formula on the items (e.g., containing all items in some set $S_1 \subseteq \{1,\dotsc,p\}$ but containing no items from some other subset $S_2 \subseteq \{1,\dotsc,p\}$), and $N_2$ being the number of remaining tests.
%\end{cor}

{\em Analysis of defective items.} To apply Lemma \ref{lem:multi_cond}, we first fix a defective $j \in S$ and consider the triplet $(\Nneg,\Ntilposj,\Nother)$, where:
\begin{itemize}
    \item $\Nneg$ is the number of negative tests; recall from \eqref{eq:p_neg} that the probability of a negative test is $\qneg = \big((1-\rho) e^{-\nu}\big)(1+o(1))$.
    \item $\Ntilposj$ is the number of tests containing $j$ but no other defective item; the probability of a given test satisfying this condition is $\qtil_j = \frac{\nu}{k}\big(1 - \frac{\nu}{k}\big)^{k-1} = \big( \frac{\nu e^{-\nu}}{k} \big)(1+o(1))$.
    \item $\Nother$ is the number of remaining tests, with associated probability $\qother = 1 - \qneg - \qtil_j$.
\end{itemize}
Hence, $(\Nneg,\Ntilposj,\Nother)$ has a multinomial distribution with $n$ trials and parameters $(\qneg,\qtil_j,\qother)$.

We now consider conditioning on the negative tests, and consequently on $\Nneg = \nneg$, $\PDhat = \pdhat$, and $G = g$, which are determined from the first step using only these tests.  We assume that $\pdhat$ contains all defective items and $g \le p^{\xi}$ non-defectives, and that $\nneg$ satisfies the concentration bound \eqref{eq:n0_conc_RZ}; recall that these events all occur with probability approaching one.

Conditioned on $(\pdhat, g, \nneg)$, we consider ``splitting'' the $\Ntilposj$ tests mentioned above into two sub-classes with counts $(\Npposj, \Ntilposj - \Npposj)$.  By the definitions of $\Ntilposj$ and $\Npposj$, each test will fall in the second sub-class if it contains a non-defective from $\pdhat$, and in the first class otherwise.  Hence, the associated conditional probability of falling in the first sub-class is\footnote{In Lemma \ref{lem:multi_cond} we only condition on $W_0 = w_0$, which plays the role of $\Nneg = \nneg$ here.  Since the tests are independent and the conditioning on $\PDhat = \pdhat$ and $G = g$ only depends on the negative tests, this additional conditioning does not affect the conditional joint distribution of $(\Ntilposj,\Nother)$ given $\Nneg = \nneg$; it only affects the probability $\gamma$ appearing in Lemma \ref{lem:multi_cond}. }
\begin{equation}
    \gamma = \bigg(1 - \frac{\nu}{k}\bigg)^g = 1 - o(1), \label{eq:gamma_calc}
\end{equation}
where we used the fact that $g = o(k)$ (since $g \le p^{\xi}$ with $\xi < \theta$, whereas $k \asymp p^{\theta}$).

Combining the above observations and applying Lemma \ref{lem:multi_cond}, we find that the joint distribution of $(\Npposj, \Ntilposj - \Npposj, \Nother)$ given $(\pdhat, g, \nneg)$ is multinomial with $n - \nneg$ trials, with the first probability parameter being 
\begin{align}
    \qpposj 
        &= \frac{\qtil_j \gamma}{1 - \qneg} \\
        &= \bigg( \frac{\nu e^{-\nu}}{k (1 - \qneg)} \bigg)(1+o(1))
\end{align}
by the above calculations of $\qtil_j$ and $\gamma$.
Since the marginal of a multinomial distribution is binomial and $n - \nneg = n(1-\qneg)(1+o(1))$ (see \eqref{eq:n0_conc_RZ}), we deduce that
\begin{multline}
    \hspace*{-1.5ex} (\Npposj \,|\, \pdhat, g, \nneg) \sim \Binomial\bigg( n(1-\qneg)(1+o(1)), \\ \frac{\nu e^{-\nu}}{k (1 - \qneg)} (1+o(1)) \bigg). \label{eq:Ntil1_dist_RZ}
\end{multline}
Observe that the mean $\mutild_1$ of this distribution satisfies
\begin{equation}
\mutild_1 = \bigg(\frac{n}{k} \nu e^{-\nu} \bigg) (1+o(1)). \label{eq:mutil_RZ}
\end{equation}
Recall from \eqref{eq:Shat_RZ} that a given item $j$ is included in the final estimate $\Shat$ if $\Npposj \ge \frac{\beta n \nu e^{-\nu}}{k}$.  Under the definition $\epsiid = 1 - \beta$, we find that the threshold $\frac{\beta n \nu e^{-\nu}}{k}$ equals $\frac{n}{k} \nu e^{-\nu}(1 - \epsiid)$, and hence, the probability that a given defective item $j$ is incorrectly excluded from $\Shat$ satisfies
\begin{align}
&\PP[ j \notin \Shat \,|\, \pdhat, g, \nneg ] \nonumber \\
&= \PP\bigg[ \Npposj < \frac{n}{k} \nu e^{-\nu}(1 - \epsiid) \,\Big|\,\pdhat, g, \nneg \bigg] \\
&\le \exp\bigg( -\frac{n}{k} \nu e^{-\nu} \cdot D_1(1-\epsiid) \cdot (1+o(1)) \bigg) \label{eq:step2d_8_RZ} 
\end{align}
by \eqref{eq:Ntil1_dist_RZ}--\eqref{eq:mutil_RZ} and the concentration bound in \eqref{eq:bino_le}.

By the union bound, the probability of there existing some $j \in S$ failing to be included in $\Shat$ is at most $k$ times the right-hand side of \eqref{eq:step2d_8_RZ}.  By re-arranging, we deduce that $\peiid \to 0$ as $p \to \infty$ under the condition
\begin{equation}
n \ge \bigg(\frac{1}{\nu e^{-\nu}} \cdot \frac{1}{D_1(\beta)} \cdot k \log k\bigg) (1+o(1)), \label{eq:n_final_2d_RZ}
\end{equation}
where we have substituted the choice $\epsiid = 1 - \beta$.

{\em Analysis of non-defective items.} To characterize the probability of a non-defective incorrectly being included in the final estimate, we apply a similar argument to the one following Lemma \ref{lem:multi_cond}.  Due to the level of similarity, we omit some details and focus on the main differences.

We consider the triplet $(\Nneg,\Ntilposj,\Nother)$ defined in the same way as above, but now with $\Ntilposj$ being the number of {\em positive} tests containing a given {\em non-defective} $j$ and none of the defective items.  The associated probability is $\qtil_j = \rho \frac{\nu}{k}\big(1 - \frac{\nu}{k}\big)^{k} = \big( \frac{\rho \nu e^{-\nu}}{k} \big)(1+o(1))$, and the probability $\qother = 1 - \qneg - \qtil_j$ changes accordingly.

We again condition on $(\pdhat, g, \nneg)$, and consider the splitting the $\Ntilposj$ tests into two sub-classes with counts $(\Npposj, \Ntilposj - \Npposj)$.  The same argument as \eqref{eq:gamma_calc} gives $\gamma = 1 - o(1)$, and we obtain the following analog of \eqref{eq:Ntil1_dist_RZ}:
\begin{multline}
\hspace*{-1.5ex}(\Npposj \,|\, \nneg,\pdhat,g) \sim \Binomial\bigg(  n(1-\qneg)(1+o(1)), \\ \frac{\rho\nu e^{-\nu}}{k (1-\qneg)} (1+o(1)) \bigg), \label{eq:Ntil1_dist_nd_RZ}
\end{multline}
with the mean satisfying
\begin{equation}
\mutilnd_1 = \bigg(\frac{n}{k} \rho \nu e^{-\nu} \bigg) (1+o(1)). \label{eq:mutilnd_RZ}
\end{equation}
Defining $\epsiind = \frac{\beta - \rho}{\rho}$, we find that the threshold $\frac{\beta n \nu e^{-\nu}}{k}$ equals $\frac{n}{k} \rho \nu e^{-\nu}(1 + \epsiind)$, and hence, the probability of $j \notin S$ incorrectly being included in $\Shat$ satisfies
\begin{align}
&\PP[ j \in \Shat \,|\, \Ec_1^c, g, \npd ] \nonumber \\
&= \PP\bigg[ \Npposj \ge \frac{n}{k} \rho \nu e^{-\nu}(1 + \epsiind) \,\Big|\, \Ec_1^c, g, \npd \bigg] \\
&\le \exp\bigg( -\frac{n}{k} \rho \nu e^{-\nu} \cdot D_1(1+\epsiind) \cdot (1+o(1)) \bigg) \label{eq:step2nd_8_RZ} 
\end{align}
by \eqref{eq:Ntil1_dist_nd_RZ}--\eqref{eq:mutilnd_RZ} and the concentration bound in \eqref{eq:bino_ge}.

By the union bound, the probability of any $j \in \PDhat \setminus S$ incorrectly being included in $\Shat$ is at most $g \le p^{\xi}$ times the right-hand side of \eqref{eq:step2nd_8_RZ}.  By re-arranging, we deduce that $\peiind \to 0$  as $p \to \infty$ under the condition
\begin{equation}
n \ge \bigg(\frac{\xi}{\rho\nu e^{-\nu}} \cdot \frac{1}{D_1(\beta/\rho)} \cdot k \log p\bigg) (1+o(1)), \label{eq:n_final_2nd_RZ}
\end{equation}
where we have substituted $\epsiind = \frac{\beta - \rho}{\rho} = \frac{\beta}{\rho} - 1$.

{\em Wrapping up.} Combining the conditions in \eqref{eq:n_final_1nd_RZ}, \eqref{eq:n_final_2d_RZ}, and \eqref{eq:n_final_2nd_RZ}, we deduce Theorem \ref{thm:ndd_RZ}.

\subsection{Algorithm-Independent Converse}

We will prove a converse that holds for all algorithms under Bernoulli testing.  Since $S$ is assumed to be uniform over the subsets of cardinality $k$, the decoding rule that minimizes $\pe$ is maximum likelihood:
\begin{equation}
    \Shat_{\rm ML}(\yv,\Xv) = \argmax_{S \,:\, |S| = k} \, \PP[ \Yv = \yv | \Xv, S ].\label{eq:ml}
\end{equation}
Hence, it suffices to lower bound the number of tests required for vanishing error probability under this decoding rule.  We first define two families of random variables:
\begin{enumerate} 
\item Given a defective item $i \in S$, as in \cite{Ald14a}, we write $M_i$ for the number of tests containing defective item $i$ but no other defective. 
\item We consider the set $\TT$ of intruding possible defectives, defined as follows: $j \in \TT$ if item $j$ is non-defective and does not appear in any
negative tests.  Given an intruding possible defective $j \in \TT$, we write $N_j$ for the number of false positive tests containing $j$.
\end{enumerate}
Let $\Npnd$ denote the overall total number of false positive tests introduced by the reverse Z-channel.
The key idea is to look for a defective item $i \in S$ with $M_i \leq c$ (for some $c$ to be determined), and an intruding possible defective item $j \in \TT$ with $N_j > c$. If we can find such items, then the set $S \setminus \{ i \} \cup \{ j \}$ explains the outcomes better than $S$ (since it requires at most $c + (\Npnd-N_j) < \Npnd$ errors to be made by the reverse Z-channel, yielding a higher likelihood assuming $\rho < 1/2$). As a result, the optimal maximum likelihood (ML)  algorithm will make a mistake.

Suppose that $n = \gamma k \log p$, which equates to a rate of $\frac{1-\theta}{\gamma \log 2}$. We first argue that there exists a defective item that is not the unique one in too many tests.

\begin{lem} \label{lem:lonedef} 
    Fix $\gamma > 0$ and $\Phi < 1$, and let $c = \Phi \nu e^{-\nu} \gamma \log p$ and $n = \gamma k \log p$. If $\theta > D_1(\Phi) \nu e^{-\nu} \gamma$, then with probability approaching one there exists a defective item $i \in S$ such that $M_i \leq c$.
\end{lem}
\begin{proof}
As in \cite{Ald14a}, the $M_i$ (together with other random variables corresponding to there being no defectives and multiple defectives in a test) are jointly multinomially distributed. Writing $r = \nu (1- \nu/k)^{k-1}  \leq \nu e^{-\nu} $,
the marginal distribution of each 
$M_i$ is  $\Binomial(n, r/k)$, so $\ep[M_i] =r  \gamma \log p$.  In addition, we have
\begin{align} \label{eq:binombd}  
\pr\left[ M_i \leq c \right]  &= \pr \left[ M_i \geq \Phi \nu e^{-\nu} \gamma \log p \right] \nonumber \\
    & \leq 
\pr \left[ M_i \geq \Phi r \gamma \log p \right]. 
\end{align}
  Using the concentration bound \eqref{eq:bino_le}
(with $N = n = \gamma k \log p$, $q =r/k$, and $\epsilon = 1 - \Phi$,  also implying $N q = r  \gamma  \log p$), we deduce that
\begin{equation} \label{eq:binombd2}  
    \pr\left[ M_i \leq c \right] \leq p^{- D_1(\Phi) r \gamma}. 
\end{equation}
Next, we use the key fact (see \cite[Section 3.1]{joag-dev})  that multinomial random variables satisfy the so-called negative association property, which implies for any $c$ that \cite[Property P3]{joag-dev}
\begin{align} \label{eq:multi-NA}
\pr \left[ \min_{i \in S} M_i \geq c+1 \right] 
    &= \pr \left[ \bigcap_{i \in S} \{ M_i \geq c+1 \} \right] \nonumber \\
    & \leq \prod_{i \in S} \pr[M_i \geq c+1] \nonumber \\    
    &= \left( 1 - \pr[M_i \leq c] \right)^k.
\end{align}
Using \eqref{eq:binombd} and \eqref{eq:multi-NA}, we deduce that
\begin{align}
    \pr \left[ \min_{i \in S} M_i \leq c \right]  & =  1-  \pr \left[ \min_{i \in S} M_i \geq c+1 \right] \\
    &\geq 1 - \left( 1 - \pr[M_i \leq c] \right)^k  \\
    & \geq 1 - \exp( - k \pr[M_i \leq c])  \\
    & = 1 - \exp( - p^{\theta - D_1(\Phi) r \gamma}). \label{eq:tobd1}
\end{align}
Recalling that $r =  \nu (1- \nu/k)^{k-1} = \big(\nu e^{-\nu}\big)(1+o(1))$, we deduce that if $\theta > D_1(\Phi) \nu e^{-\nu} \gamma$, then $\pr \left[ \min_{i \in S} M_i \leq c \right] \rightarrow 1$.
\end{proof}

Similarly, we can find an intruding possible defective that appears in sufficiently many false positive tests.

\begin{lem} \label{lem:intruding} 
Fix $\gamma > 0$ and $\Phi > \rho$, and let $c = \Phi \nu e^{-\nu} \gamma \log p$ and $n = \gamma k \log p$. If  $\gamma \nu e^{-\nu} (D_\rho(\Phi)+ 1 - \rho) < 1$, then with probability tending to 1 there exists a non-defective item $j$ such that $N_j > c$. 
\end{lem}
\begin{proof}
The probability of a particular test containing no defectives is $(1- \nu/k)^k$; hence, the probability that a test is negative is $(1-\rho)
(1- \nu/k)^k  $ and the probability that a test is a false positive is $\rho (1-\nu/k)^k$.  The total number of such tests (which we
refer to as $\Nneg$ and $\Npnd$ respectively)  both have a  binomial marginal distribution, due to the independence among tests. Hence,
since $(1-\nu/k)^k \rightarrow e^{-\nu}$ from below,
 using concentration results of the form \eqref{eq:bino_le} and \eqref{eq:bino_ge}, for arbitrarily small $\epsilon' > 0$ we may assume that 
$\Nneg \leq \ep[\Nneg](1+\epsilon') \leq n(1-\rho) e^{-\nu} (1+\epsilon')$ and $\Npnd \geq \ep[\Npnd](1-\epsilon'/2) \geq n \rho e^{-\nu} (1-\epsilon')$ for $p$ sufficiently large. % (The probability that this does not occur is exponentially small in $n$, so does not affect the statement of the lemma -- we will choose $\epsilon'$ in due course).

Conditioned on any such $\Nneg = \nneg$, the probability of a given $j \in S^c$ belonging to $\TT$ (i.e., being an intruding non-defective) is
\begin{align}
    &(1- \nu/k)^{\nneg} \nonumber \\
    &\geq  (1- \nu/k)^{n e^{-\nu} (1-\rho)(1+\epsilon')} \\
        &= \exp \left( - \nu e^{-\nu} (1-\rho)(1+\epsilon') \gamma \log p (1-o(1)) \right),
\end{align}
where we recall that $n = \gamma k \log p$.
Hence, again applying binomial concentration (with $p-k$ trials), we have with probability tending to one that the expected number of intruding possible defectives satisfies $| \TT | \geq  p^{\tau}/2$, where
\begin{equation}    
    \tau = 1 - \nu e^{-\nu} (1-\rho) (1+ 2\epsilon') \gamma. \label{eq:tau}
\end{equation}
We seek to find a possible defective item lying in at least $c$ of the $\Npnd$ false positive tests. Since (conditioned on $\Npnd = \npnd$) it holds that
 $N_j \sim \Binomial(\npnd, \nu/k)$, and since  by assumption $\npnd 
\geq n \rho e^{-\nu} (1-\epsilon')$, we know that $N_j$ is stochastically dominated by $N^{\dagger}_j$, defined to be $\Binomial(n \rho e^{-\nu} (1-\epsilon'), \nu/k)$.
Since $ \rho < \Phi$,  by a similar argument to the proof of Lemma \ref{lem:lonedef}, we have
\begin{align}
    \pr \left[ \max_{j \in \TT} N_j > c \right] 
        &\geq  \pr \left[ \max_{j \in \TT} N^{\dagger}_j > c \right] \nonumber \\
        & \geq 1 - \exp \left( -  \frac{1}{2} p^{\tau} \pr[ N^{\dagger}_j > c] \right). \label{eq:tobd2}
\end{align}
Taking $N = n e^{-\nu} \rho$ and $q = \nu/k$ and $\epsilon = \Phi/\rho - 1$, we know from the discussion following \eqref{eq:bino_gesharp} that \eqref{eq:bino_ge} has a matching lower bound up to a $p^{-\epsilon'}$ pre-factor.  Combining this observation with 
% in the lower bound \eqref{eq:bino_gesharp}, we obtain (using the fact that $N q = \gamma \nu e^{-\nu} \log p$) that the sharpness factor\footnote{See the discussion following \eqref{eq:bino_gesharp}.}is at least $p^{-\epsilon'}$ for $p$ sufficiently large, so that 
\eqref{eq:handy} gives
\begin{align} 
&p^{\tau} \pr[ N^{\dagger}_j > c] \nonumber \\ 
    & \geq p^{\tau - \epsilon'} \exp \left( - \gamma \nu e^{-\nu}  \rho \log p D_1( \Phi/\rho)  \right) \label{eq:ptau1} \\
    & = \exp \left( \log p \left( \tau - \epsilon' - \gamma e^{-\nu} \nu D_\rho(\Phi) \right) \right) \label{eq:ptau2} \\
    &=
    \exp \big( \log p \big( 1 -\epsilon' -\gamma \nu e^{-\nu}(D_\rho(\Phi) \nonumber \\
        &\hspace*{3cm} + (1 - \rho)(1+2 \epsilon') \big) \big), \label{eq:ptau3}
\end{align}
where \eqref{eq:ptau2} uses \eqref{eq:handy}, and \eqref{eq:ptau3} substitutes the definition of $\tau$ in \eqref{eq:tau}.  Hence, if  $\gamma \nu e^{-\nu}(D_\rho(\Phi)+ 1 - \rho) < 1$, then we can choose $\epsilon'$ sufficiently small such that in \eqref{eq:tobd2}, the probability $ \pr \left[ \max_{j \in \TT} N_j > c \right] $ tends to one.
\end{proof}

We now put Lemmas \ref{lem:lonedef} and \ref{lem:intruding} together.  Suppose that there 
exists $\gamma > 0$ and $\Phi \in ( \rho, 1)$ such that the rate $R= \frac{1-\theta}{\gamma \log 2}$ satisfies
\begin{multline} \label{eq:tohold2}
\frac{1-\theta}{\gamma \log 2}  \geq  
\max\bigg\{
\left( \frac{ (1-\theta) \nu e^{-\nu}}{\log 2} \right) \frac{ D_1(\Phi)}{\theta}, \\
\left( \frac{ (1-\theta) \nu e^{-\nu}}{\log 2} \right) (D_\rho(\Phi)+ 1 - \rho)
\bigg\}.
%R=    \frac{1-\theta}{\gamma \log 2}  \geq  
%\left\{ \begin{array}{l}
%\left( \frac{ (1-\theta) \nu e^{-\nu}}{\log 2} \right) \frac{ D_1(\Phi)}{\theta},   \\
%     \left( \frac{ (1-\theta) \nu e^{-\nu}}{\log 2} \right) (D_\rho(\Phi)+ 1 - \rho). \\
%\end{array} \right.
\end{multline}
Then, with probability approaching one, by Lemma \ref{lem:lonedef} there exists a defective item
$i$ in fewer than $c$ tests with no other defective, and by Lemma \ref{lem:intruding} a non-defective item $j$ which appears in more than $c$ false positive
tests. In this case, the set $S \setminus \{ i \} \cup \{ j \}$ is preferred by the ML decoder to the true defective set $S$, and an error occurs.

If we perform too few tests, then $\gamma$ will be small, and so $\frac{1}{\gamma}$ will be large, meaning \eqref{eq:tohold2} eventually holds. We would like to find the largest value for $\gamma$ that allows such a $\Phi$ to exist.  Similarly to the achievability proof, this final step amounts to somewhat tedious algebra, so the details are deferred to Appendix \ref{app:manip_conv_RZ}.

\section{Proof of Theorem \ref{thm:Zrate} (Z-Channel)} \label{sec:Z_proofs}

While the achievability and converse proofs below use similar ideas to those of the previous section, the details differ enough that we consider it worthwhile to include most steps.

\subsection{Achievability via Noisy DD}

Recall the Z-channel of Definition \ref{def:chanmod}. We first describe a noisy DD algorithm for this model. \medskip

\noindent\fbox{
    \parbox{0.95\columnwidth}{
        \textbf{Noisy DD algorithm for Z-channel noise:}
        \begin{enumerate}
            \item For each $j \in [p]$, let $\Nnegj$ be the number of negative tests in which item $j$ is included.  In the first step, we fix a constant $\alpha \in (\rho,1)$ and construct the following set of items that are believed to be non-defective:
            \begin{equation}
            \NDhat = \bigg\{ j \in [p] \,:\, \Nnegj \ge \frac{\alpha n \nu}{k} \bigg\}.  \label{eq:NDhat_Z}
            \end{equation}
            The remaining items, $\PDhat = [p] \setminus \NDhat$, are believed to be ``possibly defective''.
            \item For each $j \in \PDhat$, let $\Npposj$ be the number of positive tests that include item $j$ and no other item from $\PDhat$.  In the second step, we estimate the defective set as follows:
            \begin{equation}
            \Shat = \bigg\{ j \in \PDhat \,:\, \Npposj > 0 \bigg\}.  \label{eq:Shat_Z}
            \end{equation}
        \end{enumerate}
    }
} \medskip

\smallskip
Since positive tests must contain a defective item under the Z-channel model, we deduce that as long as the first step is correct (in the sense that $S \subseteq \PDhat$),  the second step will never add a defective item to $\Shat$.

\begin{thm} \label{thm:ndd_Z}
    Consider the Z-channel noisy group testing setup with parameter $\rho \in (0,1)$, number of defectives $k \asymp p^{\theta}$ (where $\theta \in (0,1)$), and i.i.d.~Bernoulli testing with parameter $\nu > 0$.  For any $\alpha \in (\rho,1)$, the noisy DD algorithm achieves $\pe \to 0$  as long as 
    \begin{equation}
    n \ge \max\big\{ \nid, \nind, \niid \big\} (1+\eta) \label{eq:ndd_bound_Z}
    \end{equation}
    for arbitrarily small $\eta > 0$, where defining $\zeta = e^{-\nu} +  \rho(1- e^{-\nu})$, we have
    \begin{align}
    \nid &= \frac{1}{\nu \rho D_1(\alpha/\rho)} %\cdot \frac{1}{(1+\epsid) \log (1+\epsid) - \epsid} 
 \cdot k \log k, \label{eq:nid_Z} \\
    \nind &= \frac{1}{\nu \zeta D_1(\alpha/\zeta)}
%\cdot \frac{1}{e^{-\nu} + \rho(1-e^{-\nu})}\cdot \frac{1}{(1-\epsind) \log (1-\epsind) + \epsind} 
\cdot k \log \frac{p}{k},  \label{eq:nind_Z} \\
    \niid &= \frac{1}{(1-\rho)\nu e^{-\nu}} \cdot  k \log k.  \label{eq:niid_Z}
    \end{align}
 %   and where
  %  \begin{gather}
  %  \epsid = \frac{\alpha - \rho}{\rho}, \quad \epsind = 1 - \frac{\alpha}{ e^{-\nu} + \rho(1-e^{-\nu}) }.
  %  \end{gather}
\end{thm}

Similarly to the previous section, the remaining details in proving the achievability part of Theorem \ref{thm:Zrate} using Theorem \ref{thm:ndd_Z} amount to tedious algebra.  We therefore defer these to Appendix \ref{app:manip_ach_Z}, and focus on proving Theorem \ref{thm:ndd_Z}.

Let $\Nneg$ and $\Npos$ respectively denote the number of negative and positive tests.  Both of these quantities follow a binomial distribution; the probability of a given test $i$ being positive is given by
\begin{align}
\PP[ Y_i = 1 ] 
&= (1-\rho)\bigg(1 - \Big( 1 - \frac{\nu}{k} \Big)^k \bigg) \\
&= \Big( (1-\rho) (1-e^{-\nu}) \Big) (1+o(1)),
\end{align}
since $\big( 1 - \frac{\nu}{k} \big)^k \to e^{-\nu}$ as $k \to \infty$, and we similarly have
\begin{align}
\PP[ Y_i = 0 ] = \Big( e^{-\nu}  + \rho (1-e^{-\nu}) \Big) (1+o(1)). \label{eq:n0_p_Z}
\end{align}
Hence, and using \eqref{eq:bino_le}--\eqref{eq:bino_ge}, we have with probability approaching one that
\begin{align}
\Nneg &= n\Big( e^{-\nu} + \rho(1-e^{-\nu}) \Big) (1+o(1)), \label{eq:n0_conc_Z} \\
\Npos &= n\Big( (1-\rho) (1-e^{-\nu}) \Big) (1+o(1)). \label{eq:n1_conc_Z}
\end{align}
We again henceforth condition on the defective set $S$ taking a fixed value, say $S = \{1,\dotsc,k\}$.  

{\bf Analysis of First Step.}
For the first step, we seek to ensure that, with probability approaching one, $\PDhat = [p] \setminus \NDhat$ contains all of the defective items, and at most $p^\xi$ non-defectives, for some $\xi \in (0,\theta)$ (this will later be taken arbitrarily close to $\theta$).   We denote the complements of these events by $\peid$ and $\peind$ respectively.

{\em Analysis of defective items.} For any defective item $j \in S$, the number of negative tests in which $j$ is included is distributed as $\Nnegj \sim \Binomial\big(n,\frac{\rho \nu}{k}\big)$. % since there is a probability $\frac{\nu}{k}$ of the item being included and a probability $\rho$ of the test being flipped from positive to negative.  
Moreover,
\begin{align}
\peid &= \PP\bigg[ \bigcup_{j \in S} \{ j \in \NDhat \} \bigg] \label{eq:pe1d_1_Z}\\
&\le k \PP\bigg[ \Nnegj \ge \frac{\alpha n \nu}{k} \bigg] \label{eq:pe1d_2_Z} \\
&=  k \PP\bigg[ \Nnegj \ge \bigg(1 + \frac{\alpha - \rho}{\rho} \bigg) \cdot \EE[\Nnegj]  \bigg],  \label{eq:pe1d_3_Z}
\end{align}
where \eqref{eq:pe1d_2_Z} holds for an arbitrary fixed $j \in S$ by the union bound and the definition of $\NDhat$ ({\em cf.}, \eqref{eq:NDhat_Z}), and \eqref{eq:pe1d_3_Z} follows since $\EE[\Nnegj] = \frac{n\rho\nu}{k}$ and $1 + \frac{\alpha - \rho}{\rho} = \frac{\alpha}{\rho}$.  Applying the binomial concentration bound \eqref{eq:bino_ge} to \eqref{eq:pe1d_3_Z}, we deduce that
\begin{equation}
\peid \le k \exp\Big( -\frac{n\rho\nu}{k} \cdot D_1(1+\epsid) \Big),
\end{equation}
where $\epsid = \frac{\alpha - \rho}{\rho}$.  As a result, we can achieve $\peid \to 0$ with a number of tests satisfying
\begin{equation}
n \ge \bigg(\frac{1}{\rho\nu} \cdot \frac{1}{ D_1(\alpha/\rho)} \cdot k \log k\bigg) (1+o(1)), \label{eq:n_final_1d_Z}
\end{equation}
since $1 + \epsid = \frac{\alpha}{\rho}$.

{\em Analysis of non-defective items.} Let $\peind(\nneg)$ be defined similarly to $\peind$, but conditioned on $\Nneg$ taking a given value $\nneg$.  It suffices to establish that  $\peind(\nneg) \to 0$ for all $\nneg$ satisfying \eqref{eq:n0_conc_Z}. %, since this bound holds with probability approaching one.

Since the test outcomes depend only on the defective items, one can envision the non-defective items as being placed in each test with probability $\frac{\nu}{k}$ {\em after} the test outcomes have been produced.  As a result, given $\Nneg = \nneg$, the number of negative tests $\Nnegj$ including a given item $j \notin S$ is distributed as $(\Nnegj | \nneg) \sim \Binomial\big(\nneg,\frac{\nu}{k}\big)$.  Hence, for any $j \notin S$, we have
\begin{align}
&\PP[j \notin \NDhat \,|\, \nneg] \nonumber \\
&\qquad= \PP\bigg[ \Nnegj < \frac{\alpha n \nu}{k} \,\Big|\, \nneg \bigg] \\
&\qquad\le \PP\bigg[ \Nnegj \le (1-\epsind) \nonumber \\
    &\qquad\qquad\times \EE[\Nnegj | \nneg] (1+o(1)) \,\Big|\, \nneg \bigg], \label{eq:pe1nd_2_Z}
\end{align}
where $\epsind$ is defined in such a way that $\frac{\alpha n \nu}{k} \le (1-\epsind) \frac{\nneg \nu}{k} (1+o(1))$ for all $\nneg$ satisfying \eqref{eq:n0_conc_Z}.  By some simple re-arrangements, we find that we can choose $\epsind = 1 - \frac{\alpha}{ e^{-\nu} + \rho(1-e^{-\nu}) }$.  Applying the concentration bound \eqref{eq:bino_le} to \eqref{eq:pe1nd_2_Z}, we find that
\begin{multline}
\PP[j \notin \NDhat \,|\, \nneg] \\ \le \exp\bigg( \frac{\nneg \nu}{k} \cdot D_1(1-\epsind) \cdot (1+o(1))  \bigg). \label{eq:pe1nd_4_Z}
\end{multline}
As a result, letting $G$ denote the number of non-defective items in $\PDhat = [p] \setminus \NDhat$, we find that $\EE[G \,|\,\nneg]$ is upper bounded by $p-k$ times the right-hand side of \eqref{eq:pe1nd_4_Z}.  Applying Markov's inequality, we obtain
\begin{multline}
\PP[G \ge p^\xi \,|\,\nneg] \\ \le p^{1-\xi} \exp\bigg( \frac{\nneg \nu}{k} \cdot D_1(1-\epsind) \cdot (1+o(1)) \bigg).
\end{multline}
Since $\nneg$ satisfies \eqref{eq:n0_conc_Z}, we find that we can achieve $\peind \to 0$ under the condition
\begin{multline}
\hspace*{-1.5ex} n \ge \bigg(\frac{1-\xi}{\nu} \cdot \frac{1}{ e^{-\nu} + \rho(1-e^{-\nu})}\cdot \frac{1}{D_1(\alpha/\zeta)} \cdot k \log p \bigg) \\ \times (1+o(1)), \label{eq:n_final_1nd_Z}
\end{multline}
since by definition $\epsind = 1 - \frac{\alpha}{ \zeta }$.

{\bf Analysis of Second Step.}
We follow a similar argument to the one following Lemma \ref{lem:multi_cond}; due to these similarities, we omit some details and focus on the differences.  Note that here we only need to focus on defective items: As long as $\PDhat$ contains all defectives, the the second step will never include a non-defective item in $\Shat$ ({\em cf.}, \eqref{eq:Shat_Z}), as positive tests must always contain a defective item.

Fix a defective $j \in S$ and consider the triplet $(\Nneg,\Ntilposj,\Nother)$, defined similarly to the RZ-channel case:
\begin{itemize}
    \item $\Nneg$ is the number of negative tests; for the Z-channel, the associated probability is $\qneg = \big(e^{-\nu} + \rho(1-e^{-\nu})\big)(1+o(1))$ (see \eqref{eq:n0_p_Z}).
    \item $\Ntilposj$ is the number of {\em positive} tests containing $j$ but no other defective item; the probability of a given test satisfying this condition is $\qtil_j = (1-\rho)\frac{\nu}{k}\big(1 - \frac{\nu}{k}\big)^{k-1} = \big( \frac{(1-\rho)\nu e^{-\nu}}{k} \big)(1+o(1))$.
    \item $\Nother$ is the number of remaining tests, with associated probability $\qother = 1 - \qneg - \qtil_j$.
\end{itemize}
Hence, $(\Nneg,\Ntilposj,\Nother)$ has a multinomial distribution with $n$ trials and parameters $(\qneg,\qtil_j,\qother)$.

We condition on $(\pdhat, g, \nneg)$ with $\nneg$ satisfying the concentration bound \eqref{eq:n0_conc_Z}, $\pdhat$ containing all of the defective items, and $g \le p^{\xi}$, in accordance with the first step.  Similarly to the RZ-channel case, we consider splitting the $\Ntilposj$ tests into sub-classes with counts $(\Npposj, \Ntilposj - \Npposj)$.  Similarly to \eqref{eq:gamma_calc}, the associated conditional probability of the first sub-class is $\gamma = 1-o(1)$.

Applying Lemma \ref{lem:multi_cond}, we deduce that given  $(\pdhat, g, \nneg)$, the triplet $(\Npposj, \Ntilposj - \Npposj, \Nother)$ is multinomial with $n - \nneg$ trials, with the first probability being 
\begin{align}
    \qpposj 
        &= \frac{\qtil_j \gamma}{1 - \qneg} \\
        &= \bigg( \frac{(1-\rho)\nu e^{-\nu}}{k (1 - \qneg)} \bigg)(1+o(1)).
\end{align}
This gives the following analog of \eqref{eq:Ntil1_dist_RZ}:
\begin{multline}
    \hspace*{-1.5ex} (\Npposj \,|\, \nneg,\pdhat,g) \sim \Binomial\bigg(  n(1-\qneg)(1+o(1)), \\ \frac{(1-\rho)\nu e^{-\nu}}{k (1 - \qneg)} (1+o(1)) \bigg). \label{eq:Ntil1_dist_Z}
\end{multline}
Recall from \eqref{eq:Shat_Z} that a given item $j$ is included in the final estimate $\Shat$ if $\Npposj > 0$.  Consequently, we deduce from \eqref{eq:Ntil1_dist_Z} that any defective item $j \in S$ yields
\begin{align}
&\PP[j \notin \Shat \,|\, \nneg,\pdhat,g] \nonumber \\
&= \bigg(1 - \frac{(1-\rho)\nu e^{-\nu}}{k (1 - \qneg)} (1+o(1))\bigg)^{n(1-\qneg)(1+o(1))} \\
&\le \exp\bigg( -\frac{n}{k} (1-\rho) \nu e^{-\nu} \cdot (1+o(1)) \bigg), \label{eq:step2d_8_Z} 
\end{align}
where we have applied $1 - \zeta \le e^{-\zeta}$. 
By the union bound, the probability of there existing some $j \in S$ failing to be included in $\Shat$ is at most $k$ times the right-hand side of \eqref{eq:step2d_8_Z}.  By re-arranging, we deduce that $\peiid \to 0$ as $p \to \infty$ under the condition
\begin{equation}
n \ge \bigg(\frac{1}{(1-\rho)\nu e^{-\nu}} \cdot k \log k\bigg) (1+o(1)). \label{eq:n_final_2d_Z}
\end{equation}

{\em Wrapping up.} Combining the conditions in \eqref{eq:n_final_1d_Z}, \eqref{eq:n_final_1nd_Z}, and \eqref{eq:n_final_2d_Z}, and noting that $\xi$ can be arbitrarily close to $\theta$ in \eqref{eq:n_final_1nd_Z} (yielding the logarithmic term $(1-\theta)\log p = \big(\log \frac{p}{k}\big)(1+o(1))$), we deduce Theorem \ref{thm:ndd}. 

\subsection{Algorithm-Independent Converse}

The proof of the converse part of Theorem \ref{thm:Zrate} proceeds along similar lines to the corresponding argument for Theorem \ref{thm:RZrate}.  We consider the optimal ML decoder of the form \eqref{eq:ml}, identify the suitable error event, analyze the error probability using concentration inequalities, and studying the intersection of relevant $D_\gamma$ functions.

Let $\Ntilposi$ denote the number of positive tests containing $i \in S$ and no other defectives, and $\Ntilnegi$ denote the number of negative tests containing $i \in S$ and no other defectives. The probability of a test containing no defectives is $(1- \nu/k)^k$, so we write
$\Nnnd \sim \Binomial(n, (1-\nu/k)^k)$ for the number of tests that contain no defectives (rightfully negative). By binomial concentration, we can assume that $\Nnnd = n e^{-\nu}(1 + o(1))$, as this occurs with probability approaching one.

The argument used to prove a converse is the following:
\begin{enumerate}
\item We look for a defective item $i \in S$ with $\Ntilposi = 0$ and $\Ntilnegi \geq d = \Psi \nu e^{-\nu} \gamma \log p$ (for some $\Psi$ to be determined). We call such a defective item ``strongly masked''.
\item Given a strongly masked item $i$, we look for an non-defective item  that  appears in fewer
than $d$  of the $\Nnnd + \Ntilnegi$ of the tests that would be rightfully negative if $i$ were removed
from the defective set. We call such an item ``weakly intruding''.
\end{enumerate}
If there exists such a strongly masked item $i \in S$ and a  weakly intruding item $j \notin S$, then since $d \leq \Ntilnegi$ the ML decoder will prefer the set $S \setminus \{i \} \cup \{ j \}$ to the true  defective set (assuming $\rho < 1/2$), and hence will make a mistake.  We first argue that there exists a strongly masked defective item.

\begin{lem} \label{lem:tamedefective}
Fix $\gamma > 0$, $\rho \leq \Psi < 1$, and $d = \Psi \nu e^{-\nu} \gamma \log p = \Psi \nu e^{-\nu} n/k$ with $n = \gamma k \log p$. If 
$ \theta > e^{-\nu} \nu \gamma (D_\rho(\Psi) + 1 - \rho)$, then with probability tending to one there exists a strongly masked item
 $i \in S$ (with $\Ntilposi = 0$ and $\Ntilnegi \geq d$).
\end{lem}
\begin{proof}
    For each defective $i \in S$, we write $V_i$ for the indicator of the event that  it is strongly masked, and $V = \sum_{i \in S} V_i$ for the
    total number of strongly masked items.
    As in \cite{Ald14a}, we know that $(\Ntilposi, \Ntilnegi)$ are components of a multinomial distribution, with respective parameters $(n, q_+, q_-)$, where $q_+ = (1-\nu/k)^{k-1} (\nu/k)(1-\rho)$ and $q_- = (1-\nu/k)^{k-1} (\nu/k) \rho$
    are the probabilities of a particular test containing  defective $i$ and no other defectives, and being positive or negative respectively. Further, observe that
    \begin{equation}
        (\Ntilnegi \,|\, \Ntilposi = 0) \sim \Binomial(n, q_-/(1-q_+)) \label{eq:bino1}
    \end{equation}
    by a simpler version of Lemma \ref{lem:multi_cond} with $\gamma = 1$.
    
    We will apply the binomial concentration bound \eqref{eq:bino_ge} and its matching lower bound \eqref{eq:bino_gesharp} to $\Ntilnegi$ conditioned on $\Ntilposi = 0$.  Since we consider the event $\Ntilnegi \ge d$ and we have assumed $d = \Psi \nu e^{-\nu} n/k$, we introduce a constant $\epsilon_p$ defined to satisfy
    \begin{equation}
        (1+\epsilon_p) \frac{n q_-}{1-q_+} = d = \Psi \nu e^{-\nu} \frac{n}{k}, \label{eq:eps_n}
    \end{equation}
    so that the left hand side is slightly above the conditional mean of $\Ntilnegi$.  Since  $k q_+ \rightarrow e^{-\nu} \nu (1-\rho)$ and $k q_- \rightarrow e^{-\nu} \nu \rho$ by definition, we readily deduce from \eqref{eq:eps_n} that $1 + \epsilon_p \to \Psi/\rho$ as $p \to \infty$, which is in the range $\big[ 1,\frac{1}{\rho} \big]$ by the assumption $\Psi \in [\rho,1]$.
    
    Combining the above and using binomial concentration, we obtain
    \begin{align}
    \ep[V_i] 
        &= \pr[ \Ntilposi = 0, \Ntilnegi \geq d] \\
        & =   \pr[ \Ntilposi = 0] \pr[ \Ntilnegi \geq d | \Ntilposi = 0] \\
        &= \left( 1 -  q_+ \right)^n \PP\Big[ \Binomial\big(n,{q_-}/({1-q_+})\big) \ge d \Big]  \label{eq:EVi1}  \\
        & \le  \left( 1 -  q_+ \right)^n \exp \left( - \frac{n q_-}{1-q_+} D_1 (1+\epsilon_p)  \right) \label{eq:EVi2} \\
        &\le \exp \left( - n \Big(q_+ + \frac{q_-}{1-q_+} D_1 (1+\epsilon_p) \Big) \right) \label{eq:EVi3} \\
        &= \exp \Big( - \frac{n}{k} (e^{-\nu} \nu) \left( 1 - \rho + \rho  D_1(\Psi/\rho) \right) \nonumber \\
            & \hspace*{4cm} \times (1+o(1)) \Big)  \label{eq:EVi4}\\
        &= p^{ - \gamma  e^{-\nu} \nu\left( 1- \rho + D_\rho(\Psi) \right) \cdot (1+o(1)) }, \label{eq:EVi5}
    \end{align}
    where \eqref{eq:EVi1} follows from \eqref{eq:bino1}, \eqref{eq:EVi2} follows from the concentration bound \eqref{eq:bino_ge},\eqref{eq:EVi3}  uses $1-q_+ \le e^{-q_+}$, \eqref{eq:EVi4} applies the above-mentioned asymptotics of $(q_+,q_-,\epsilon_p)$, and \eqref{eq:EVi5} follows since $n = \gamma k \log p$.
    
    We can apply the same argument to lower bound $\EE[V_i]$ via \eqref{eq:bino_gesharp}.  As in the proof of Lemma \ref{lem:intruding} above the sharpness factor  is at least $p^{-\epsilon'}$ for any $\epsilon' > 0$ and $p$ sufficiently large; in other words, it behaves as $p^{o(1)}$.  Combining this lower bound with the upper bound \eqref{eq:EVi5}, we conclude that $\ep[V_i] = p^{-\gamma e^{-\nu} \nu ( 1 - \rho + D_\rho(\Psi) ) \cdot (1+o(1)) + o(1)}$.  Hence, since $k \asymp p^{\theta}$, the expectation $\ep[V] = k \ep[V_i]$ tends to infinity under the assumed condition $\theta > \gamma e^{-\nu} \nu (1 - \rho + D_\rho(\Psi) )$.
    
    %In the limit, $k q_+ \rightarrow e^{-\nu} \nu (1-\rho)$ and $k q_- \rightarrow e^{-\nu} \nu \rho$ and thus
    %$(1+\epsilon_p) \rightarrow \Psi/\rho$, and applying \eqref{eq:handy} gives
    %\begin{align} p^{-\epsilon'}
    % \exp \left( - n \left( q_+ + \frac{q_-}{1-q_+} D_1(\Psi/\rho)
    %    \right) \right) 
    % & \sim p^{-\epsilon'} \exp \left( - \frac{n}{k} (e^{-\nu} \nu) \left( 1 - \rho + \rho  D_1(\Psi/\rho) \right) \right) \\
    %    & =  \exp \left( -  (\log p) \left( \epsilon'  + \gamma  e^{-\nu} \nu\left( 1- \rho + D_\rho(\Psi) \right) \right) \right) .
    %\end{align}
    %That is, the 
    %
    
    Furthermore, we can bound the variance of $V$ by a negative association argument (using \cite[Property P3]{joag-dev} as above) to deduce that for any two defectives $i \neq j$:
    \begin{align}
     & \ep[V_i V_j] \nonumber \\
     & =   \pr[ \Ntilposi = 0, \Ntilnegi \geq d, \Ntilposj = 0, \Ntilnegj \geq d] \\
    & \leq (1-q_+)^{2n} \left( \pr[ \Binomial(n, q_-/(1-2q_+)) \ge d] \right)^2. \label{eq:bino2} %  \\
    % & = p^{- 2 \gamma e^{-\nu} \nu (1 - \rho + D_\rho(\Psi)) \cdot (1+o(1))}. 
    \end{align}
    Here the factor of $1-2q_+$ can by understood as multinomial conditioning ({\em cf.}, Lemma \ref{lem:multi_cond}) on {\em both} $\Ntilposi = 0$ and $\Ntilposj = 0$.  Continuing, we write
    \begin{align}
        \frac{ \var(V)}{ (\ep[V])^2} 
            &= \frac{ k \var[V_i] +  k(k-1) \cov[V_i, V_j]}{ k^2 (\ep [V_i])^2} \\
            &= \frac{1}{k \EE[V_i]} + \frac{ \EE[V_i V_j] }{ (\EE[V_i])^2 } - 1, \label{eq:var_ratio}
    \end{align}
    where we used $\var[V_i] = \EE[V_i] - (\EE[V_i])^2$ for binomial random variables, and the fact that $\cov[V_i, V_j] = \EE[V_i V_j] - \EE[V_i] \EE[V_j] = \EE[V_i V_j] - (\EE[V_i])^2$.
    We have already established that $k \EE[V_i] \to \infty$, so the first term in \eqref{eq:var_ratio} tends to zero as $p \to \infty$.  In addition, from \eqref{eq:EVi1} and \eqref{eq:bino2}, we have
    \begin{align}
        \frac{ \EE[V_i V_j] }{ (\EE[V_i])^2 } 
            &\le \bigg( \frac{ \PP\big[ \Binomial\big(n,\frac{q_-}{1-2q_+}\big) \ge d \big] }{\PP\big[ \Binomial\big(n,\frac{q_-}{1-q_+}\big) \ge d \big]} \bigg)^2 \\
            &\le 1+o(1), \label{eq:lim_ratio}
    \end{align}
    where \eqref{eq:lim_ratio} is proved in Appendix \ref{app:lim_ratio}.  

    Combining the preceding observations with \eqref{eq:var_ratio} gives $\frac{ \var(V)}{ (\ep[V])^2}  \to 0$.  This implies that $V$ is sufficiently concentrated around its mean to deduce (via Chebyshev's inequality) that there exists a strongly masked defective (i.e., $V > 0$) with probability approaching one, as desired.
\end{proof}

We now argue that there exists a weakly intruding non-defective. First note that the analysis so far has only considered the columns of the test matrix corresponding to defective items, and the i.i.d.~design of the test matrix means that columns corresponding to non-defectives are independent of those. Recall that an item is weakly intruding if it appears in fewer than $d$ of the  $\Nnnd + \Ntilnegi$ tests that would be rightfully negative if $i$ were removed, and that $\Nnnd = n e^{-\nu}(1+o(1))$ ({\em cf.}, start of this subsection).  Again using the binomial concentration bound in \eqref{eq:bino_ge} along with $\Ntilnegi \sim \Binomial(n, q_-/(1-q_+))$ ({\em cf.}, proof of Lemma \ref{lem:tamedefective}), we can also assume that $\Ntilnegi \leq C e^{-\nu} \nu n/k$ for fixed $C > \rho$, as this holds with probability approaching one.

\begin{lem} \label{lem:weakintrude}
Fixing $\gamma > 0$, $\Psi < 1$, and $C > \rho$, and setting $d = \Psi \nu e^{-\nu} \gamma \log p$ and $n = \gamma k \log p$, conditioned on there existing a strongly masked $i \in S$ with $\Ntilnegi \le C e^{-\nu} \nu n/k$, there also exists a weakly intruding non-defective item with probability tending to one if $1 > \gamma \nu e^{-\nu} D_1(\Psi)$.
\end{lem}
\begin{proof}
Conditioned on $\Nnnd$ and $\Ntilnegi$, the number of rightfully negative tests that each non-defective item appears in is distributed as $\Binomial( \Nnnd + \Ntilnegi, \nu/k )$ (note that the relevant columns of $\Xv$ for non-defective items are independent of those that determine $\Nnnd$ and $\Ntilnegi$).  By the above-mentioned assumptions $\Nnnd = n e^{-\nu}(1+o(1))$ and $\Ntilnegi \leq C e^{-\nu} \nu n/k$, the probability of this binomial random variable being less than $d$ is lower bounded by that of $Z \sim \Binomial(n e^{-\nu} (1 + \delta_p + C/k), \nu/k)$, where $\delta_p = o(1)$. 
    
Defining $\epsilon'_p$ to be such that 
\begin{equation}
    (1-\epsilon'_p) \frac{n}{k} e^{-\nu} \nu (1 + \delta_p + C/k)  = d = \Psi \nu e^{-\nu} \frac{n}{k},
\end{equation}
we observe that $1-\epsilon'_p \rightarrow \Psi$.  Hence, we have 
$(1 + \delta_p + C/k)  D_1(1-\epsilon'_p) = D_1(\Psi) (1+o(1))$.  In addition, when bounding $\pr[ Z < d]$, the discussion following \eqref{eq:bino_gesharp} reveals that the concentration bound \eqref{eq:bino_le} has a matching lower bound with a $p^{-\epsilon'}$ pre-factor.  Combining the above observations, we obtain for any $j \notin S$ that
\begin{align}
 &\pr[ \mbox{ item $j$  is weakly intruding } ] \nonumber \\
     & \geq    \pr[ Z < d] \\
    & \ge p^{-\epsilon'} \exp\left( -  \frac{n}{k} e^{-\nu} \nu (1 + \delta_p + C/k)  D_1(1-\epsilon'_p) \right) \label{eq:weak_intr2} \\
     & \ge p^{-\epsilon'} \exp \left( - \gamma (\log p) e^{-\nu} \nu  D_1(\Psi) \cdot (1+o(1)) \right), \label{eq:weak_intr3}
\end{align}
where and \eqref{eq:weak_intr3} uses $n = \gamma k \log p$ in addition to the above observations.

By \eqref{eq:weak_intr3}, the expected number of weakly intruding non-defective items is at least
$$ (p-k) \cdot p^{-\epsilon'} \cdot p^{- \gamma \nu e^{-\nu} D_1(\Psi) \cdot (1+o(1))},$$
which grows to infinity as $\Omega(p^{\tau})$ for some $\tau > 0$, due to the assumption  $1 > \gamma \nu e^{-\nu} D_1(\Psi)$ and the fact that $\epsilon'$ may be arbitrarily small.  Finally, the i.i.d.~design of the matrix means that non-defective items are weakly intruding independently of one another, so by binomial concentration, the actual number of of weakly intruding items is positive with probability approaching one.
\end{proof}

By the definition of rate in \eqref{eq:defrateeq} along with $n = \gamma k \log p$ and $k \asymp p^{\theta}$, we can rephrase Lemma \ref{lem:tamedefective} to say that there exists a strongly masked defective with high probability when the
rate $R \sim \frac{1-\theta}{\gamma \log 2}$ satisfies
\begin{equation} \label{eq:Zrateconst}
    \frac{1-\theta}{\gamma \log 2} >  \left( \frac{ (1-\theta) \nu e^{-\nu}}{\log 2} \right) \frac{D_\rho(\Psi) + 1 - \rho}{\theta}.
\end{equation}
Similarly, Lemma \ref{lem:weakintrude} states that there exists a weakly intruding item with high probability if
\begin{equation} 
    \frac{1-\theta }{\gamma \log 2} > \left( \frac{ (1-\theta) \nu e^{-\nu}}{\log 2} \right) D_1(\Psi). \label{eq:neededA} 
\end{equation}

Combining \eqref{eq:Zrateconst} and \eqref{eq:neededA}, we see that the error event will occur if the rate is bigger
than
\begin{equation} \label{eq:maxZchan}
\frac{ (1-\theta) \nu e^{-\nu}}{\log 2} \max \left\{ \frac{D_\rho(\Psi) + 1 - \rho}{\theta}, D_1(\Psi) \right\}.
\end{equation}
We optimize this expression with respect to $\Psi$ in Appendix \ref{app:manip_conv_Z} to complete the proof of Theorem \ref{thm:Zrate}.

%\section{Numerical Results} \label{sec:experiments}
%
%The plot in Figure \ref{fig:rates_ndd} uses $\xi = \theta$, maybe the rates will go up when $\xi$ is fully optimized over the range $(0,\theta)$.  There are also some numerical errors, particularly in the red curve.
%
%\begin{figure}[!ht]
%    \vskip-0.1cm
%    \centering
%    \includegraphics[width=0.65\columnwidth]{Plots/RatesNDD.pdf}
%    \caption{Rates for the symmetric model. \label{fig:rates_ndd}}
%\end{figure}  

\section{Conclusion} \label{sec:conclusion}

We have introduced and analyzed variants of the definite defectives (DD) algorithm for noisy group testing, with an emphasis on the Z-channel and reverse Z-channel models.  Under RZ noise, our achievability result (part of which also uses the COMP algorithm) matches an algorithm-independent converse for Bernoulli testing for a broad range of dense scaling regimes, and matches a converse specific to DD for Bernoulli testing in sparse regimes.  While more significant gaps remain for the Z-channel and general binary channel (see Appendix \ref{app:symmetric}), the bounds are still matching or near-matching in several low-noise high-sparsity regimes.  Further closing these gaps, either by improved achievability or improved converse bounds (or both), poses an interesting direction for further research.

\appendix

\subsection{Technical Lemma Regarding Intersection of $D_{\gamma}$ Functions} \label{app:lemtech}
We present a result giving an explicit formula for the intersection of the $D_\gamma$ functions introduced in \eqref{eq:ddef}, i.e., $D_\gamma(t) = t\log\frac{t}{\gamma} - t + \gamma$.
\begin{lem} \label{lem:tech} Fixing $0 \leq \gamma_1 \leq \gamma_2$, $c \geq 0$, and $d \geq 0$, we have the following:
\begin{enumerate}
\item The equation
\begin{equation} \label{eq:tosolve} \varphi(t) := D_{\gamma_1}(t) -  c D_{\gamma_2}(t) +d = 0 \end{equation}
has a unique solution for $t \in [\gamma_1, \gamma_2]$ if and only if
\begin{equation} \label{eq:diffcond}  d \leq c D_{\gamma_2}(\gamma_1).\end{equation}
\item If \eqref{eq:diffcond} fails, the smallest value of $\varphi(t)$ is achieved by $t = \gamma_1$, and equals $- c D_{\gamma_2}(\gamma_1) + d  > 0$.
\item If \eqref{eq:diffcond} holds, the  solution to \eqref{eq:tosolve} is given by 
\begin{align}
t^* & = \frac{\gamma_2 - \gamma_1 - d}{\log(\gamma_2/\gamma_1)}  \mbox{\;\;\;\; for $c=1$}  \label{eq:tvalue1} \\
t^* & = \gamma_1^{1/(1-c)} \gamma_2^{-c/(1-c)} e^{z^*+1}  \mbox{\;\;\;\; for $c \neq 1$}  \label{eq:tvalue2} \\
 & = - \frac{ \gamma_1 + d - c \gamma_2}{(1-c) z^*}, \label{eq:tvalue3}
\end{align}
where $z^*$ is a solution to 
\begin{equation} \label{eq:zlambert} z e^z = -\frac{1}{e} \frac{1}{1-c} \left( 1 + \frac{d - c \gamma_2}{\gamma_1} \right) \left( \frac{\gamma_2}{\gamma_1} \right)^{c/(1-c)}.\end{equation}
This can be found using the appropriate branch of the Lambert $W$-function.
\end{enumerate}
\end{lem}
\begin{proof}
We have the following:
\begin{enumerate}
\item Observe that $D_{\gamma_1}(t)$  is strictly increasing for $t \in (\gamma_1,\gamma_2)$ and  $D_{\gamma_2}(t)$ is strictly decreasing for $t$ in this range,
so $\varphi(t)$  is strictly increasing in $t \in (\gamma_1,\gamma_2)$. Furthermore, $\varphi(\gamma_2) = D_{\gamma_1}(\gamma_2) + d \geq 0$ by definition, and $\varphi(\gamma_1) = - c D_{\gamma_2}(\gamma_1) + d$, which is non-positive if and only if \eqref{eq:diffcond} holds. In other words, \eqref{eq:diffcond} is equivalent to the existence of a change of sign of $\varphi$ at some point in $[\gamma_1,\gamma_2]$.
\item If \eqref{eq:diffcond} fails then $\varphi(\gamma_1) > 0$, so since $\varphi$ is increasing we know that it is minimized at this point.
\item For $c=1$, we can solve directly using the fact that 
\begin{align}
    \varphi(t) &= D_{\gamma_1}(t) -  D_{\gamma_2}(t) +d \nonumber \\  
    &= t \log \left( \frac{\gamma_2}{\gamma_1} \right) + \gamma_1 - \gamma_2 + d,
\end{align}
yielding \eqref{eq:tvalue1}.
For $c \neq 1$, we consider reparameterizing $t = \gamma_1^{1/(1-c)} \gamma_2^{-c/(1-c)} e^{z+1}$ by another variable $z$ to obtain
\begin{align}
 D_{\gamma_1}(t) & = t \log \left( \frac{\gamma_1^{1/(1-c)} \gamma_2^{-c/(1-c)} e^z}{\gamma_1} \right) + \gamma_1 \\
 D_{\gamma_2}(t) & = t \log \left( \frac{\gamma_1^{1/(1-c)} \gamma_2^{-c/(1-c)} e^z}{\gamma_2} \right) + \gamma_2.
\end{align}
Since $\frac{1}{1-c} -1 = \frac{c}{1-c}$ and $-\frac{c}{1-c} - 1 = -\frac{1}{1-c}$, we find that evaluating $ D_{\gamma_1}(t) - c D_{\gamma_2}(t)$ leads to a cancellation of powers of $\gamma_1$ and $\gamma_2$ inside the logarithm, and hence
\begin{align}  
    \varphi(t) &= D_{\gamma_1}(t) - c D_{\gamma_2}(t) +d  \\
        &= tz(1-c) + \gamma_1 - c\gamma_2 + d \label{eq:alt_phi} \\
        &= e (1-c) \gamma_1^{1/(1-c)} \gamma_2^{-c/(1-c)} \cdot z e^{z}  \nonumber \\  
            &\hspace*{3cm} + (d + \gamma_1 - c \gamma_2),
\end{align}
and \eqref{eq:zlambert} follows. We deduce \eqref{eq:tvalue2} from the definition of $t$, and \eqref{eq:tvalue3} by equating \eqref{eq:alt_phi} with zero. 
\end{enumerate}
\end{proof}

% \subsection{Proof of Theorem \ref{thm:ndd_RZ} (RZ Achievability)} \label{sec:pf_RZ}  

% \subsection{Proof of Theorem \ref{thm:ndd_Z}  (Z Achievability)} \label{sec:pf_Z} 

\subsection{Proving Theorem \ref{thm:RZrate} (achievability part, RZ noise) via Theorem \ref{thm:ndd_RZ}} \label{app:manip_ach_RZ}

We consider the rates achievable by COMP and NDD separately ({\em cf.}, Lemma \ref{lem:COMP_RZ} and Theorem \ref{thm:ndd_RZ}). Since $\frac{k \log_2(p/k)}{k \log p} \sim \frac{1-\theta}{\log 2}$, Lemma \ref{lem:COMP_RZ}  tells us that the limiting rate in \eqref{eq:defrateeq} for COMP becomes
\begin{align}  \label{eq:comprate} 
R  &\sim  \frac{k \log_2\frac{p}{k}}{\ncomp}  = \frac{k \log_2\frac{p}{k}}{k \log p}   \frac{k \log p}{\ncomp}  \nonumber \\
&\sim \frac{(1-\theta)}{\log 2} (1- \rho) \nu e^{-\nu}
= \frac{(1-\theta)(1- \rho)}{e \log 2} ,
\end{align}
taking the choice of $\nu = 1$ that maximizes $\nu e^{-\nu}$.

To find the rate achievable by NDD,
we rewrite \eqref{eq:nind_RZ}, \eqref{eq:niid_RZ} and \eqref{eq:niind_RZ}  respectively, using \eqref{eq:handy} in the third case,
to obtain
\begin{align}
\frac{k \log p}{n_1^{(D)}} & = \nu e^{-\nu} \frac{1-\rho}{1-\xi} \label{eq:rewriteRZ1} \\
\frac{k \log p}{n_2^{(D)}} & = %\nu e^{-\nu} \frac{ \beta \log \beta + 1 - \beta}{\theta} =  
\nu e^{-\nu} \frac{ D_1(\beta)}{\theta} \label{eq:rewriteRZ2} \\
\frac{k \log p}{n_2^{(ND)}} & = \nu e^{-\nu} \frac{ \rho D_1(\beta/\rho)}{\xi} =  \nu e^{-\nu} \frac{ D_\rho(\beta)}{\xi}, \label{eq:rewriteRZ3} 
\end{align}
since $k \log_2(p/k) = \big((1-\theta)k\log_2p\big) (1+o(1))$.  

Note that \eqref{eq:rewriteRZ1} is increasing in $\xi$, \eqref{eq:rewriteRZ3} is decreasing in $\xi$,  \eqref{eq:rewriteRZ2} is decreasing in $\beta
\in (\rho,1)$ and \eqref{eq:rewriteRZ3} is increasing in $\beta \in (\rho,1)$, so we need to choose the parameters $\xi$ and $\beta$ to balance these terms.
Using a similar argument to \eqref{eq:comprate} above, and again
making the optimal choice $\nu = 1$, the limiting rate in \eqref{eq:defrateeq} becomes
\begin{align} \label{eq:threeterm}
R &\sim \frac{k \log_2(p/k)}{n} \nonumber \\
    &\sim \frac{ (1-\theta) }{e \log 2} \min \left\{  \frac{1-\rho}{1-\xi}, \frac{ D_1(\beta)}{\theta},
\frac{ D_\rho(\beta)}{\xi} \right\}.
\end{align}

Since $D_1$ and $D_{\rho}$ in \eqref{eq:threeterm} are continuous in their arguments, we can consider choosing $\beta \in [\rho,1]$ and $\xi \in [0,\theta]$,\footnote{As $\xi \to 0$, the third term in \eqref{eq:threeterm} grows unbounded, so for $\xi = 0$ we simply lower bound the minimum of three terms by that of the first two terms.} where previously we excluded the endpoints.  If there exists a value $\beta^* \in [\rho,1]$  that makes 
\begin{equation} \label{eq:liketrue} \frac{ D_1(\beta^*)}{\theta}= D_\rho(\beta^*) + 1 - \rho, 
\end{equation}
then taking $\xi = \xi^* := \frac{D_{\rho}(\beta^*)}{D_\rho(\beta^*) + 1 - \rho}$ (if valid, i.e., in $[0,\theta]$) would make all bracketed terms in \eqref{eq:threeterm} become equal to $ D_\rho(\beta^*) + 1 - \rho$, suggesting that a putative rate of 
\begin{equation} \label{eq:putativerate} \frac{(1-\theta) (D_\rho(\beta^*) + 1 - \rho) }{e \log 2} \end{equation} 
might be possible.
We can consider whether there exists a solution to \eqref{eq:liketrue} by taking $\gamma_1 = \rho$, $\gamma_2 = 1$, $c=\frac{1}{\theta}$ and $d= 1-\rho$ in Lemma \ref{lem:tech}. Examining \eqref{eq:diffcond}, there exists a solution $\beta^* \in (\rho,1)$ if and only if
$(1-\rho) \theta \leq D_1(\rho) = \rho \log \rho + 1- \rho$, or equivalently if $\theta \leq \thetacrit = 1 + \frac{\rho \log \rho}{1-\rho}$ ({\em cf.}, \eqref{eq:theta_crit}).

{\bf Case 1 ($\theta \leq \thetacrit$).} Equation \eqref{eq:tvalue2} in Lemma \ref{lem:tech} shows that $\beta^* = \rho^{1/(1-1/\theta)} e^{z+1} = \rho^{\theta/(\theta-1)} e^{z+1}$, where   $z$ is a solution to $z e^z = - e^{-1} \rho^{\theta/(1-\theta)}$ (see \eqref{eq:zlambert}, and observe that $ \frac{1}{1-c} \big( 1 + \frac{d - c \gamma_2}{\gamma_1} \big) = \frac{1}{1 - 1/\theta}\big( 1 + \frac{1-\rho-1/\theta}{\rho} \big) = \frac{1}{\rho}$), i.e., $z = -\kappa$ in the notation of \eqref{eq:kappadef}. Substituting the
value of $\beta^*$ into the definition of $D_{\rho}$ and applying some algebra,\footnote{Write $z e^z = - e^{-1} \rho^{\theta/(1-\theta)}$ as $e^{z+1} \rho^{\theta/(\theta-1)} = -\frac{1}{z}$, or equivalently $\beta^* = \frac{1}{\kappa}$.  Then $D_{\rho}(\beta^*) + 1 - \rho = \beta^* \log \frac{\beta^*}{\rho} - \beta^* + 1 = \beta^* \log \frac{\beta^*}{\rho e} + 1
= \beta^* \frac{  - \log \rho}{1-\theta} - \beta^* \kappa  +1 = \beta^* \frac{  - \log \rho}{1-\theta} $,
since $\frac{\beta^*}{\rho e} = \rho^{-1/(1-\theta)} e^{-\kappa}$. } we obtain
\begin{equation} 
    \label{eq:bestvalue} (1-\theta)(D_\rho(\beta^*) + 1 - \rho) = \frac{ - \log \rho}{\kappa(\theta)}.
\end{equation}
We need to verify whether the corresponding parameter $\xi^*$ satisfies $\xi^* < \theta$, which is equivalent to $1 - \xi^* > 1 - \theta$ and in turn (recalling $\xi^* = \frac{D_{\rho}(\beta^*)}{D_\rho(\beta^*) + 1 - \rho}$) to
\begin{equation} 
    \label{eq:Dbd} (1-\theta)(D_\rho(\beta^*) + 1 - \rho) < 1-\rho, 
\end{equation}  
or equivalently $-\log \rho < \kappa (1-\rho)$ ({\em cf.}, \eqref{eq:bestvalue}). Direct calculation shows that this is satisfied if and only if $\theta > \thetaopt$, where $\thetaopt$ is defined in \eqref{eq:theta_opt}; this is deduced by substituting the ``endpoint'' value $\kappa = \frac{-\log \rho}{1-\rho}$ into $\kappa e^{-\kappa} =  - e^{-1} \rho^{\theta/(1-\theta)}$.\footnote{In more detail, this choice gives $\log\big( \kappa e^{-\kappa} \big) = \frac{\log \rho}{1 - \rho} + \log\big( \frac{-\log \rho}{1 - \rho} \big)$, and we can also rewrite $\kappa e^{-\kappa} =  - e^{-1} \rho^{\theta/(1-\theta)}$ as $\theta = \big(1 - \frac{\log \rho}{ 1 + \log(\kappa e^{-\kappa}) }\big)^{-1} = \frac{1+\log\big( \kappa e^{-\kappa} \big)}{1 + \log\big( \kappa e^{-\kappa} \big) - \log \rho}$ by applying $\frac{\theta}{1-\theta} = \frac{1}{1/\theta - 1}$ and re-arranging.  Combining these two facts gives $\theta = \theta_{\rm opt}$.} 

Hence for $\thetaopt < \theta \leq \thetacrit$,  using the fact that $(1-\theta)(D_\rho(\beta^*) + 1 - \rho) = \frac{- \log \rho}{\kappa(\theta)}$ as per \eqref{eq:putativerate}, we obtain a rate of $\frac{-\log \rho}{\kappa(\theta)  e \log 2}$ as claimed. 

For $\theta \leq \thetaopt$, $\xi^*$ is not a legitimate choice, but we can obtain a rate of $\frac{1-\rho}{e \log 2}$ by picking $\beta = \beta^*$ and $\xi = \theta$. In this case, the first term (i.e., $\frac{1-\rho}{1-\theta}$) provides the minimum in \eqref{eq:threeterm} since the reverse of \eqref{eq:Dbd} holds, implying 
\begin{gather}
 \frac{1-\rho}{1-\theta} \leq D_\rho(\beta^*) + 1 - \rho = \frac{D_1(\beta^*)}{\theta}, \\
 \frac{1-\rho}{1-\theta} \leq \frac{D_\rho(\beta^*)}{\theta},
\end{gather}
 where the equality in the first expression applies \eqref{eq:liketrue}, and the second expression follows by rewriting the reverse of \eqref{eq:Dbd} as $(1-\theta)D_{\rho}(\beta^*) \ge (1-\rho) - (1-\theta)(1-\rho) = \theta(1-\rho)$.

{\bf Case 2 ($\theta > \thetacrit$).} In this case, the second part of Lemma \ref{lem:tech} tells us that the optimal choice is to take $\beta = \rho$ and $\xi$ arbitrarily close to zero (to keep the third term in \eqref{eq:threeterm} zero while maximizing the first term). However, in analogy with \eqref{eq:liketrue}, for $\theta > \thetacrit$ it holds that 
$\frac{D_1(\rho)}{\theta} < 1- \rho$, so the minimum in \eqref{eq:threeterm} is strictly smaller than the first term. In other words, the optimized NDD rate is strictly less than $\frac{(1-\theta)(1-\rho)}{e \log 2}$, which we know from \eqref{eq:comprate} above is achievable by COMP.
In other words, for sufficiently dense problems, the NDD rate bound of Theorem \ref{thm:ndd_RZ} is worse than that attained by COMP.  Therefore, in this regime, we get the required rate in Theorem \ref{thm:ndd_RZ} from COMP instead of NDD.

\subsection{Proving Theorem \ref{thm:RZrate} (converse part, RZ noise) via \eqref{eq:tohold2}} \label{app:manip_conv_RZ}

By the discussion following \eqref{eq:tohold2}, we want to find the value of $\Phi$ that gives the smallest value of 
\begin{equation}    
    \max \left\{ \frac{ D_1(\Phi)}{\theta},  D_\rho(\Phi)+ 1 - \rho \right\}.
\end{equation}
This is precisely the problem considered in \eqref{eq:liketrue} (without any need to consider constraints on $\xi$) and recall the following observations that we established via Lemma \ref{lem:tech}.

{\bf Case 1 ($\theta \leq \thetacrit$).} We know that the smallest value is given by the intersection of the two curves.  In accordance with \eqref{eq:bestvalue}, if
\begin{equation}  
R > \frac{- \log \rho}{\kappa(\theta) e \log 2} \label{eq:optrate1}
\end{equation}
then the success probability of the ML algorithm tends to zero (again making the optimal choice $\nu = 1$)

{\bf Case 2 ($\theta \geq \thetacrit$).} Recall from the arguments following \eqref{eq:liketrue} that the maximum of $\frac{D_1(\Phi)}{\theta}$ and $D_\rho(\Phi)+ 1 - \rho$ is achieved
by $D_\rho(\Phi)+ 1 - \rho$, and the smallest such value  is attained when $\Phi = \rho$. Therefore, if the rate satisfies
\begin{equation}  
R \geq \left( \frac{ 1-\theta }{e \log 2} \right) (1-\rho), \label{eq:optrate2}
\end{equation}
then the success probability of the ML algorithm tends to zero (again using the fact that $\nu = 1$ maximizes $\nu e^{-\nu}$).

Finally, the presence of the (reverse) Z-channel capacity $C_{\mathrm{Z}}(\rho)$ in Theorem \ref{thm:RZrate} needs no further justification, as such a bound was proved in \cite{Bal13}.

\subsection{Proving Theorem \ref{thm:Zrate} (achievability part, Z noise) via Theorem \ref{thm:ndd_Z}} \label{app:manip_ach_Z}

In the following, recall that $\zeta = e^{-\nu} +  \rho(1- e^{-\nu}) \geq \rho$.  In a similar way to the proof of Theorem \ref{thm:RZrate}, we can rewrite \eqref{eq:nid_Z}, \eqref{eq:nind_Z}  and \eqref{eq:niid_Z} as
\begin{align}
\frac{k \log k}{n_1^{(D)}} &  %= \nu \left( \alpha \log(\alpha/\rho) - \alpha + \rho \right) 
= \nu D_\rho(\alpha) \label{eq:rewriteRZ4} \\
\frac{k \log \frac{p}{k}}{n_1^{(ND)}} &  % = \nu \frac{\theta}{1-\theta} \left( \alpha \log(\alpha/\zeta) - \alpha + \zeta \right)  
= \nu D_\zeta(\alpha) \label{eq:rewriteRZ5} \\
\frac{k \log k}{n_2^{(D)}} & = \nu e^{-\nu} (1-\rho). \label{eq:rewriteRZ6} 
\end{align}
Using the fact that $\frac{k \log_2 (p/k)}{k \log k} = \frac{1-\theta}{\theta \log 2} (1+o(1))$, we deduce that for any choice
of parameters $\nu$ and $\alpha$, an achievable
rate is given by
\begin{multline} \label{eq:Zachrate}  
\frac{1}{\log 2} \min \bigg\{  \frac{(1-\theta) \nu D_\rho(\alpha)}{\theta},  \\ \nu D_\zeta(\alpha), \frac{(1-\theta)(1-\rho) \nu e^{-\nu}}{\theta} \bigg\}.\end{multline}
Again, we maximize \eqref{eq:Zachrate}, first equating the first two terms using Lemma \ref{lem:tech} with
$\gamma_1 = \rho$, $\gamma_2 = \zeta$, $c = \frac{\theta}{1-\theta}$ and $d=0$, for which \eqref{eq:diffcond} trivially holds and so a unique $\alpha$ solving $ \frac{(1-\theta) \nu D_\rho(\alpha)}{\theta} =  \nu D_\zeta(\alpha)$ exists. If $\theta = \frac{1}{2}$, then \eqref{eq:tvalue1} gives 
$$\alpha^*(1/2) =  \frac{\zeta-\rho}{\log(\zeta/\rho)}  = \frac{ \rho s}{\log(1+s)},$$
since $s = \frac{\zeta}{\rho} - 1$ by the definitions of $s$ and $\zeta$.
Note that the required\footnote{We may again include the endpoints $\alpha = \rho$ and $\alpha = \zeta$ due to the continuity of $D_{\rho}$ and $D_{\zeta}$, similarly to the reverse Z-channel model.} bound $\rho \leq \alpha^*(1/2) \leq \zeta$ holds because it is equivalent to $\log(1+s) \leq s \leq (1+s) \log(1+s)$.
We obtain \eqref{eq:Zratedoable2} by substituting this value in $D_\rho(\alpha)$ in \eqref{eq:Zachrate} (note that $D_\rho(\alpha^*(1/2)) = \frac{\rho s}{ \log(1+s) } \cdot \log\frac{s}{\log(1+s)} - \frac{\rho s}{\log(1+s)} + \rho$ and $\rho s = \zeta - \rho = (1-\rho)e^{-\nu}$).

In the case $\theta \neq \frac{1}{2}$,   \eqref{eq:zlambert} can be expressed in the form\footnote{The relevant terms are evaluated as follows: (i) $\frac{1}{1-c} = \frac{1-\theta}{1-2\theta}$, (ii) $\big( 1 + \frac{d - c \gamma_2}{\gamma_1} \big) = 1 - \frac{\theta}{1-\theta} \frac{\zeta}{\rho}$, (iii)  $\big( \frac{\gamma_2}{\gamma_1} \big)^{c/(1-c)} = \big( \frac{\zeta}{\rho} \big)^{\frac{\theta}{1-2\theta}}$.  The latter two of these are further simplified using $\frac{\zeta}{\rho} = 1+s$. }
\begin{equation} e^{z^*} z^* = -\frac{1}{e} g(s, \theta), \end{equation}
for the function $g$ given in \eqref{eq:gratio}. This
is solved by taking $z^* = \lambda(\theta)$ with $\lambda(\theta)$ given in \eqref{eq:lambdadef}. Our choice of branch of the Lambert $W$-function is justified in Remark \ref{rem:branchchoice} below. The value of $\alpha^*$ given in \eqref{eq:alphadef} then follows via \eqref{eq:tvalue3}.\footnote{In more detail, $- \frac{ \gamma_1 + d - c \gamma_2}{(1-c) z^*}$ evaluates to $-\frac{\rho - \frac{\theta}{1-\theta}\zeta}{(1 - \frac{\theta}{1-\theta})\lambda}$.  Multiplying and dividing by $\rho$ and using $\frac{\zeta}{\rho} = 1 + s$, this simplifies to $-\frac{\rho}{\lambda}\big(\frac{1-\theta-\theta(1+s)}{ 1-2\theta }\big)$ and in turn to \eqref{eq:alphadef}.}

Next, note the following two facts:
\begin{itemize}
    \item By \eqref{eq:tvalue2}, we have $\frac{\alpha^*}{\rho e} = \big(\frac{1}{1+s}\big)^{c/(1-c)} e^{z^*}$ (this is established via $\rho^{\frac{1}{1-c}} = \rho \cdot \rho^{\frac{c}{1-c}}$ and $\frac{\rho}{\zeta} = \frac{1}{1+s}$).
    \item By \eqref{eq:tvalue3}, we have $\alpha^* z^* = - \rho - \frac{\rho s \theta}{2 \theta -1}$ (this is established via $-\frac{ \rho - \frac{\theta}{1-\theta} \zeta }{ 1 - \frac{\theta}{1-\theta}} = -\frac{ \rho(1-2\theta) - (\zeta - \rho) }{ 1 - 2\theta }$ and $\zeta - \rho = \rho s$).
\end{itemize}
Using the definition of $D_{\rho}$ followed by these two facts, the first (and therefore also the second) bracketed term in \eqref{eq:Zachrate} equates to
\begin{align}
&\frac{(1-\theta) \nu D_\rho(\alpha^*) }{\theta}\nonumber \\
 & = \frac{(1-\theta) \nu}{\theta} \left( \alpha^* \log \left( \frac{ \alpha^*}{\rho e} \right) + \rho \right) \\
& = \frac{(1-\theta) \nu}{\theta} \left( (\alpha^* z^* + \rho) - \alpha^* \frac{c}{1-c} \log(1+s) \right) \\
& = \frac{(1-\theta) \nu}{\theta} \left( - \frac{\rho  s \theta}{2 \theta -1}  + \alpha^* \frac{\theta}{2 \theta-1} \log(1+s)  \right) \label{eq:commonrate} \\
& = \frac{(1-\theta) (1-\rho) \nu e^{-\nu}}{2 \theta-1} \left( - 1  +  \frac{\alpha^* \log(1+s)}{\rho s} \right), \label{eq:commonrate2}
\end{align}
where \eqref{eq:commonrate2} uses $\rho s = \zeta - \rho = (1-\rho)e^{-\nu}$.  We obtain the desired rate in \eqref{eq:Zratedoable1} upon substituting into \eqref{eq:Zachrate} and noting that $\frac{\log(1+s)}{\rho s} = \alpha^*(1/2)^{-1}$.

\begin{rem} \label{rem:branchchoice}
    The choice of branch of the Lambert $W$-function in \eqref{eq:lambdadef} follows from the fact that (using the parametrization
$t = \gamma_1^{1/(1-c)} \gamma_2^{-c/(1-c)} e^{z+1}$) we require
    $\rho \leq  \rho (1+s)^{\theta/(2 \theta-1)} e^{z+1} \leq \zeta$ in Lemma \ref{lem:tech}. 
    Rearranging, and using $\frac{\zeta}{\rho} = 1+s$, the lower bound is equivalent to the fact that
    \begin{equation} \label{eq:branchchoice1}
    z \geq - 1 - \frac{\theta}{2 \theta -1} \log (1+s), \end{equation}
    and the upper bound is equivalent to the fact that
    \begin{equation}\label{eq:branchchoice2}
    z \leq - 1 - \frac{1-\theta}{2 \theta -1} \log (1+s). \end{equation}
    Hence, for $\theta < \frac{1}{2}$, \eqref{eq:branchchoice1} tells us that $z \geq -1$, so we need to take the $W_0$ branch. Similarly,
    for $\theta > \frac{1}{2}$, \eqref{eq:branchchoice2} tells us that $z \leq -1$, so we take the $W_{-1}$ branch.
\end{rem}

We can now justify in more detail some of the claims made earlier in Remark \ref{rem:Zratedetailsshort} regarding the behavior of the achievable rate.

\begin{rem} \label{rem:Zratedetails}
    \begin{enumerate}
        \item For given $s$, the function $\alpha^*(\theta)$ is continuous at $\theta = \frac{1}{2}$. This follows because as $\theta \rightarrow \frac{1}{2}$ from below, $-e^{-1} g(s, \theta) \rightarrow \infty$, and \eqref{eq:Wlimit} gives $W_0(x) \sim \log x$ for $x$ large, yielding
        \begin{multline}    
            \lambda(\theta) \sim \log( -e^{-1} g) \\ = - \frac{\theta}{2 \theta - 1} \log (1+s) - 1 + \log \left( \frac{ \theta s}{1- 2 \theta} - 1 \right)
        \end{multline}
        by the definition of $g$. (Here $\sim$ means that the ratio of the terms tends to 1). Similarly, as $\theta \rightarrow \frac{1}{2}$ from above, $-e^{-1} g(s, \theta) \rightarrow 0$,  and \eqref{eq:Wlimit2} gives $W_{-1}(x) \sim \log(-x)$ for $x$ close to zero, yielding
        \begin{multline}
            \lambda(\theta) \sim \log( e^{-1} g) \\ = - \frac{\theta}{2 \theta - 1} \log (1+s) - 1 + \log \left( \frac{ \theta s}{2 \theta -1} + 1 \right).
        \end{multline}
        In either case, we deduce that $(2 \theta - 1) \lambda(\theta) \rightarrow - \frac{\log(1+s)}{2}$, and $\alpha^*(\theta) \rightarrow \frac{\rho s}{\log(1+s)}$ so the definition in \eqref{eq:alphadef} is continuous at this point.
        Since $\alpha^*(\theta)$ is continuous at $\theta = 1/2$, the definition of the rate $\Runder{Z}(\theta, \rho)$ given in (32) is also continuous at $\theta = 1/2$ (since $\Runder{Z}$ is obtained by substituting $\alpha^*$ into a continuous
        function $D_\rho(\alpha)$).
        By \eqref{eq:Wlimit}, we know that $W_0(x) \leq \log x$ for $x$ sufficiently large, and hence $\alpha^*(\theta) < \alpha^*(1/2)$ for $\theta$ in some left-neighborhood of $1/2$. Similarly, $W_1(x) \leq \log(-x)$ for all $x$ sufficiently large (and negative), and hence $\alpha^*(\theta) > \alpha^*(1/2)$ for $\theta$ in some right-neighborhood of $1/2$.
\item In the limit as $\rho$ tends to zero, a sub-optimal but useful choice of $\alpha$ is
\begin{equation}    
    \alpha = \frac{ e^{-\nu} \theta}{1-\theta} \frac{1}{(-\log \rho)}, \label{eq:alpha}
\end{equation}
        which is greater than $\rho$ (as required for $\alpha > \rho$) when $\rho$ is sufficiently small. Using $a D_{b}(t) = D_{ab}(at)$ (a simple generalization of \eqref{eq:handy}), the first two bracketed terms of \eqref{eq:Zachrate} simplify to
        \begin{align*}
            &\frac{\nu e^{-\nu}}{\log 2} \min\bigg\{  \frac{ (1-\theta) e^\nu}{\theta} D_\rho(\alpha), e^\nu D_\zeta(\alpha) \bigg\} \\
            &=  \frac{\nu e^{-\nu}}{\log 2} \min\bigg\{  D_{\rho (1-\theta) e^\nu/\theta} \left( \frac{\alpha (1-\theta) e^\nu}{\theta}  \right), \\
            &\hspace*{5cm} D_{e^\nu\zeta}(\alpha e^\nu) \bigg\}
        \end{align*}
and under the above choice of $\alpha$, both $D_{(\cdot)}$ terms tend to 1 as $\rho \rightarrow 0$.\footnote{For the first term, this is established by substituting \eqref{eq:alpha} into the definition of $D_{(\cdot)}$ to obtain a ratio of the form $\frac{-\log \rho + o(\log \rho)}{- \log \rho + o(\log \rho)}$.  For the second term, simply note that $\zeta e^{\nu} \to 1$ and $\alpha \to 0$.} Setting $\nu = 1$, we obtain $\frac{1}{e \log 2}$, and 
taking into account the  third term of  \eqref{eq:Zachrate} we recover the noiseless rate of \eqref{eq:ddnoiselessrate} in the limit as $\rho \to 0$.
    \end{enumerate}
\end{rem}

\subsection{Proving Theorem \ref{thm:Zrate} (converse part, Z noise) via \eqref{eq:maxZchan}} \label{app:manip_conv_Z}

% \begin{proof}[Proof of Theorem \ref{thm:Zrate} (converse)]

Recall from \eqref{eq:maxZchan} that the error probability cannot vanish when the rate is above
\begin{equation}
    \frac{ (1-\theta) \nu e^{-\nu}}{\log 2} \max \left\{ \frac{D_\rho(\Psi) + 1 - \rho}{\theta}, D_1(\Psi) \right\}. \label{eq:R_repeated_Z}
\end{equation}
The analysis of this term is simpler than in the reverse Z-channel case.
We again use Lemma \ref{lem:tech}, taking $\gamma_1 = \rho$, $\gamma_2 = 1$, $c= \theta$ and $d= 1-\rho$.  In this case, we find that \eqref{eq:diffcond} does not hold, since
\begin{equation}
    1- \rho \geq \theta  (1-\rho) \geq \theta( 1- \rho + \rho \log \rho) = c D_1(\rho).
\end{equation}
In other words, the maximum in \eqref{eq:R_repeated_Z} is always provided by the first term, and by choosing $\Psi$ arbitrarily close to $\rho$, we deduce (also using $D_{\rho}(\rho) = 0$) that any rate above
\begin{equation}
    \frac{ (1-\theta) \nu e^{-\nu}}{\log 2} \frac{(1-\rho)}{\theta}
\end{equation}
will ensure the error probability $\pe$ does not converge to zero. Finally, using the fact that $\nu e^{-\nu} \leq 1/e$, we deduce that
using any rate larger than
\begin{equation}
    \frac{ (1-\theta)(1-\rho)}{\theta e \log 2}
\end{equation}
will ensure that $\pe$ does not converge to zero.

\subsection{Proof of Equation \eqref{eq:lim_ratio} in Lemma \ref{lem:tamedefective}} \label{app:lim_ratio}

It suffices to show that
\begin{equation}
    \frac{ \PP\big[ \Binomial\big(n,\frac{q_-}{1-q_+}\big) \ge d \big] }{\PP\big[ \Binomial\big(n,\frac{q_-}{1-2q_+}\big) \ge d \big]} \ge 1 + o(1).
\end{equation}
To simplify the notation, we write $P$ and $Q$ for binomial PMFs (with $n$ trials) with probabilities $p = \frac{q_-}{1-q_+}$ and $q = \frac{q_-}{1-2q_+}$.  We are interested in lower bounding $\frac{\sum_{x \ge d} P(x)}{ \sum_{x \ge d} Q(x) }$.  To do this, we write $\sum_{x \ge d} Q(x) = \sum_{x \ge d} P(x) + \sum_{x \ge d} \big( Q(x) - P(x) \big)$ and seek to show that
\begin{align}
    \sum_{x \ge d} \big( Q(x) - P(x) \big) \ll \sum_{x \ge d} P(x), \label{eq:goal}
\end{align}
in the sense that the left-hand side is upper bounded by a vanishing fraction of the right-hand side.

Towards establishing \eqref{eq:goal}, note that
\begin{align}
    &\sum_{x \ge d} \big( Q(x) - P(x) \big) \nonumber \\
        &\quad= \sum_{x \ge d} P(x) \bigg( \frac{Q(x)}{P(x)} - 1 \bigg) \\
        &\quad= \sum_{x \ge d} P(x) \bigg( \Big( \frac{q}{p} \Big)^x \Big( \frac{1-q}{1-p} \Big)^{n-x} - 1 \bigg). \label{eq:bino_ratio}
\end{align}
Observe that under the choices $p = \frac{q_-}{1-q_+}$ and $q = \frac{q_-}{1-2q_+}$, we have
\begin{align}
    \frac{p}{q} = \frac{1-2q_+}{1-q_+} = 1 - \frac{q_+}{1-q_+} < 1.
\end{align}
We also have $\frac{1-p}{1-q} > 1$ (since $p < q$), so we can upper bound  $\big( \frac{1-q}{1-p} \big)^{n-x}$ in \eqref{eq:bino_ratio} by one.  Combining these findings gives
\begin{align}
\sum_{x \ge d} \big( Q(x) - P(x) \big) 
    &\le \sum_{x \ge d} P(x) \bigg( \Big( \frac{q}{p} \Big)^x - 1 \bigg) \\
    &= \sum_{x \ge d} P(x) \bigg( \Big( \frac{1}{1 - \frac{q_+}{1-q_+}}  \Big)^x - 1 \bigg) \\
    &= \sum_{x \ge d} P(x) \big( e^{x q_+ (1+o(1))} - 1 \big), \label{eq:sum_x}
\end{align}
since $q_-$ and $q_+$ are both $o(1)$.

To complete the proof, recall that $d = O(\log p)$, whereas $q_+ = O\big(\frac{1}{k}\big)$ with $k \asymp p^{\theta}$ and $\theta > 0$.  We fix $\xi < \theta$ and split the summation in \eqref{eq:sum_x} as follows:
\begin{align}
    &\sum_{x \ge d} P(x) \big( e^{x q_+ (1+o(1))} - 1 \big) \nonumber \\
        &= \sum_{d \le x \le p^{\xi}} P(x) \big( e^{x q_+ (1+o(1))} - 1 \big) \nonumber \\
        &~~ +  \sum_{p^{\xi} \le x \le n} P(x) \big( e^{x q_+ (1+o(1))} - 1 \big). \label{eq:two_sums}
\end{align}
The first summation behaves as $o\big( \sum_{x \ge d} P(x) \big)$, since the conditions $x \le p^{\xi}$, $q_+ = O\big(\frac{1}{k}\big)$, $k \asymp p^{\theta}$, and $\xi \in (0,\theta)$ collectively imply $e^{x q_+ (1+o(1))} = 1+o(1)$.  For the second summation, we recall that $P(x)$ is the PMF of a $\Binomial(n,q_-/(1-q_+))$ random variable, and make use of the concentration bound \eqref{eq:bino_le}.  Equating $x$ with $n q_-/(1-q_+) (1+\epsilon)$ gives $\epsilon = \frac{x(1-q^+)}{nq_-} -1$, which yields $\epsilon = \Omega\big( \frac{p^{\xi}}{\log p} \big)$ due to the fact that $nq_- \asymp \log p$ and $x \ge p^{\xi}$.  The concentration bound in \eqref{eq:bino_le} gives an exponent of $\frac{nq_-}{1-q_+} D_1(1+\epsilon)$, and we have $D_1(1+\epsilon) \asymp \epsilon \log \epsilon$ when $\epsilon \to \infty$; combining these observations yields an exponent on the order of $x \log \epsilon$, or equivalently, on the order of $x \log p$.  Substituting $P(x) \le e^{- \Omega(x \log p)}$ into the second summation of \eqref{eq:two_sums}, we readily obtain $\sum_{p^{\xi} \le x \le n} P(x) \big( e^{x q_+ (1+o(1))} - 1 \big) \le e^{- \Omega( p^{\xi} \log p )}$.  This is much faster decay than the $p^{-O(1)}$ decay of $\sum_{x \ge d} P(x)$ shown following \eqref{eq:EVi1}.

Recalling that \eqref{eq:two_sums} is equal to \eqref{eq:sum_x}, we deduce that \eqref{eq:goal} holds with  $p = \frac{q_-}{1-q_+}$ and $q = \frac{q_-}{1-2q_+}$, as desired.

\subsection{DD-Specific Converse Under RZ Noise} \label{sec:dd_spec_conv}

Consider the reverse Z-channel with parameter $\rho$.  We claim that under i.i.d.~Bernoulli testing with parameter $\nu > 0$, if the number of tests satisfies
\begin{equation}
n = \frac{k \log \frac{p}{k}}{\nu e^{-\nu} (1-\rho)} (1-\eta) \label{eq:n_dd_conv}
\end{equation}
for some fixed $\eta \in (0,1)$, then the error probability of DD tends to one.  We proceed by proving this claim.

{\bf First step.} Recall that the first step removes all items that appear in any negative test.  Given that there are $\nneg = \big(n e^{-\nu} (1-\rho) \big)(1+o(1))$ negative tests ({\em cf.}, \eqref{eq:n0_conc_RZ}), the probability of a given non-defective $j$ being kept is
\begin{equation}
\psi := \bigg( 1 - \frac{\nu}{k} \bigg)^{\nneg} = e^{-\big(\frac{n}{k} \nu e^{-\nu} (1-\rho)\big)(1+o(1))}. \label{eq:psi}
\end{equation}
Given $\nneg$, the number of intruding non-defectives (i.e., non-defectives appearing in no negative tests) is distributed as $G \sim \Binomial(p-k,\psi)$, so if $\psi p \to \infty$ then we have by the binomial concentration bounds \eqref{eq:bino_le}--\eqref{eq:bino_ge} that $G = p\psi (1+o(1))$ with probability approaching one.  Using \eqref{eq:n_dd_conv} and \eqref{eq:psi}, we obtain
\begin{equation}
\psi = e^{-(1-\eta)\log\frac{p}{k} \cdot (1+o(1))} = \Big( \frac{k}{p} \Big)^{(1-\eta)(1+o(1))},
\end{equation}
and hence
\begin{equation}
p\psi = k^{(1-\eta)(1+o(1))} \cdot p^{\eta (1+o(1))}.
\end{equation}
Since $k \asymp p^{\theta}$ with $\theta \in (0,1)$, we deduce that $p\psi \ge 2 k^{1+\epsilon}$ for sufficiently large $p$ and some $\epsilon > 0$ depending only on $\theta$.  Then, since $G = p\psi (1+o(1))$ with probability approaching one, we conclude that $G \ge k^{1+\epsilon}$ with probability approaching one.

{\bf Second step.} We proceed similarly to the argument following Lemma \ref{lem:multi_cond}, but now the key probability $\gamma$ analyzed in \eqref{eq:gamma_calc} behaves very differently due to the fact that $g \ge k^{1+\epsilon}$.  Specifically, we have
\begin{align}
    \gamma 
        &= \bigg(1 - \frac{\nu}{k}\bigg)^g \\
        &= \bigg( \Big( 1 - \frac{\nu}{k} \Big)^{\frac{k}{\nu}} \bigg)^{\frac{\nu g}{k}} \\
        &= \Big( \big(e^{-1}\big)(1+o(1)) \Big)^{ \frac{\nu g}{k} } \\
        &\le e^{- \nu k^{\epsilon} (1+o(1))}, \label{eq:g_large}
\end{align}
where we have applied $g \ge k^{1+\epsilon}$.  Hence, in the binomial distribution found in \eqref{eq:Ntil1_dist_RZ}, the associated probability decreases by this factor accordingly:
\begin{multline}
    \hspace*{-1.5ex} (\Npposj \,|\, \pdhat, g, \nneg) \sim \Binomial\bigg( n(1-\qneg)(1+o(1)), \\ \frac{\nu e^{- \nu \gamma}}{k (1 - \qneg)} (1+o(1)) \bigg). \label{eq:Ntil1_modified}
\end{multline}
By \eqref{eq:g_large}, the mean of this distribution is at most $\frac{\nu n}{k} e^{- \nu k^{\epsilon} (1+o(1))}$, which vanishes under the choice of $n$ in \eqref{eq:n_dd_conv} due to the fact that $k \asymp p^{\theta}$ with $\theta > 0$.  By Markov's inequality, it follows that the probability of the given defective item being the unique element from $\pdhat$ in {\em any} positive test vanishes as $p \to \infty$, and hence the probability that the DD algorithm successfully identifies a given defective item also vanishes.

% Applying the union bound, we deduce that the probability of the given defective item being alone in {\em any} positive test is at most $\frac{\nu n}{k} e^{- \nu k^{\epsilon} (1+o(1))}$, which vanishes under the choice of $n$ in \eqref{eq:n_dd_conv} due to the fact that $k \asymp p^{\theta}$ with $\theta > 0$.  Therefore, the probability that the DD algorithm successfully identifies a given defective item vanishes, establishing the desired claim.

%Now consider the $\npos \approx n(\rho e^{-\nu} + 1-e^{-\nu})$ positive tests.  As before, the property of $j$ being alone is conditionally independent across these tests, so we get
%\begin{equation}
%    \PP[j\text{ alone in some positive test}] = 1 - \bigg(1 - o\Big( \frac{1}{k} \Big)\bigg)^{\npos} = 1 - e^{-\delta n/k},
%\end{equation}
%for some $\delta = o(1)$.  Finally, we have
%\begin{align}
%    \Pr[\mathrm{error}] 
%        &\ge 1 - \PP\bigg[ \bigcap_{j \in S} \big\{ j \text{ alone in some positive test} \} \bigg] \\
%        &\ge 1 - \bigg( 1 - e^{-\delta n/k} \bigg)^k, \label{eq:dd_conv_final}
%\end{align}
%where the last line follows since the probability of the intersection is upper bounded by what we would get if the events were independent; this is because for each test $j_1$ is alone in, that leaves one less test $j_2$ can be alone in, making the event for $j_2$ less likely. 
%
%Finally, since $\delta \to 0$, we observe that \eqref{eq:dd_conv_final} tends to one whenever $n \asymp k \log k \asymp k \log p$.

\subsection{High-$\rho$ Low-$\theta$ Optimality Result for Reverse Z-Channel} \label{sec:high_noise_opt}

    In this appendix, we consider the alternative formulation of the achievability part of Theorem \ref{thm:RZrate} given in Theorem \ref{thm:ndd_RZ}. 
    % We claim that for fixed $\rho \in (0,1)$, it holds for sufficiently small $\theta > 0$ that the maximum in \eqref{eq:ndd_bound_RZ} is attained by the first condition \eqref{eq:nind_RZ}.  To see this, 
    We let $\xi$ be arbitrarily close to $\theta$, and set $\beta = \frac{1 + \rho}{2} \in (\rho,1)$.  The conditions \eqref{eq:niid_RZ} and \eqref{eq:niind_RZ} both have a dependence on $(k,p,\theta)$ scaling as $O(k \log k) = O(\theta k \log p)$, which is dominated by the $k \log p$ term in \eqref{eq:nind_RZ} for sufficiently small $\theta$.  Hence, the condition \eqref{eq:nind_RZ} dominates for sufficiently small $\theta$.
    
    When \eqref{eq:nind_RZ} dominates, the required number of tests simplifies to
    \begin{align}
    n = \bigg( \frac{e}{1-\rho} \cdot k \log p \bigg) (1+o(1)),
    \end{align}
    which yields a rate of $\frac{1-\rho}{e \log 2}$ bits per test.  By performing a Taylor expansion of \eqref{eq:zcap} at $\rho = 1$, we find that the reverse Z-channel capacity also behaves as $C = \frac{1-\rho}{e \log 2} + O\big( (1-\rho)^2 \big)$ as $\rho \to 1$, so in fact we have asymptotic optimality in this high-noise regime. 
    
    Stated more precisely, the rate of Theorem \ref{thm:ndd_RZ} is asymptotically optimal when the order of limits is first $n \to \infty$, then $\theta \to 0$, and finally $\rho \to 1$.

\subsection{Results for the General Binary Noise Model} \label{app:symmetric}

{\bf Achievability.} Under the general binary noise model (which includes symmetric noise as a special case), we consider the following noisy version of the DD algorithm.

\medskip
\noindent\fbox{
    \parbox{0.95\columnwidth}{
        \textbf{Noisy DD algorithm for general binary noise:}
        \begin{enumerate}
            \item For each $j \in [p]$, let $\Nnegj$ be the number of negative tests in which item $j$ is included.  In the first step, we fix a constant $\alpha \in (\rho_{10},1-\rho_{01})$ and construct the following set of items that are believed to be non-defective:
            \begin{equation}
            \NDhat = \bigg\{ j \in [p] \,:\, \Nnegj \ge \frac{\alpha n \nu}{k} \bigg\}. \label{eq:NDhat}
            \end{equation}
            The remaining items, $\PDhat = [p] \setminus \NDhat$, are believed to be ``possibly defective''.
            \item For each $j \in \PDhat$, let $\Npposj$ be the number of positive tests that include item $j$ and no other item from $\PDhat$.  In the second step, we fix a constant $\beta \in (\rho_{01},1-\rho_{10})$, and estimate the defective set as follows:
            \begin{equation}
            \Shat = \bigg\{ j \in \PDhat \,:\, \Npposj \ge \frac{\beta n \nu e^{-\nu}}{k} \bigg\}. \label{eq:Shat}
            \end{equation}
        \end{enumerate}
    }
} \medskip

\begin{thm} \label{thm:ndd}
    {\em (General binary noise achievability)}
    Consider the general binary noisy group testing setup with crossover probabilities $\rho_{01}$ and $\rho_{10}$ ({\em cf.}, \eqref{eq:gt_gen_model}), number of defectives $k \asymp p^{\theta}$ (where $\theta \in (0,1)$), and i.i.d.~Bernoulli testing with parameter $\nu > 0$.  For any $\alpha \in (\rho_{10},1-\rho_{01})$, $\beta \in (\rho_{01},1-\rho_{10})$, and $\xi \in (0,\theta)$, we have $\pe \to 0$ if
    \begin{equation}
        n \ge \max\big\{ \nid, \nind, \niid, \niind \big\} (1+\eta) \label{eq:ndd_bound}
    \end{equation}
    for arbitrarily small $\eta > 0$, where defining $w = (1-\rho_{01}) e^{-\nu} + \rho_{10}(1- e^{-\nu})$, we have
    \begin{align}
    \nid &= \frac{1}{\nu \rho_{10} D_1(\alpha/\rho_{10})} %\cdot \frac{1}{(1+\epsid) \log (1+\epsid) - \epsid}
    \cdot k \log k, \label{eq:nid} \\
    \nind &= \frac{1-\xi}{\nu w D_1(\alpha/w)} 
    %\cdot \frac{1}{(1-\rho)e^{-\nu} + \rho(1-e^{-\nu})}\cdot \frac{1}{(1-\epsind) \log (1-\epsind) + \epsind} 
    \cdot k \log p,  \label{eq:nind} \\
    \niid &= \frac{1}{\nu e^{-\nu} (1-\rho_{10}) D_1(\beta/(1-\rho_{10}))} % \cdot \frac{1}{(1-\epsiid) \log (1-\epsiid) + \epsiid} 
    \cdot k \log k,  \label{eq:niid} \\
    \niind &= \frac{\xi}{\nu e^{-\nu} \rho_{01} D_1(\beta/\rho_{01})} %\cdot \frac{1}{(1+\epsiind) \log (1+\epsiind) - \epsiind} 
    \cdot k \log p.   \label{eq:niind}
    \end{align}
    %    and where
    %    \begin{gather}
    %    \epsid = \frac{\alpha - \rho}{\rho}, \quad \epsind = 1 - \frac{\alpha}{ (1-\rho)e^{-\nu} + \rho(1-e^{-\nu}) }, \\
    %    \epsiid = \frac{(1-\rho) - \beta}{1-\rho}, \quad
    %    \epsiind = \frac{\beta - \rho}{\rho}. 
    %    \end{gather}
\end{thm}
\begin{proof}
    The first step of the algorithm is analyzed in the same way as the proof of Theorem \ref{thm:ndd_Z}, and the second step is analyzed in the same way as the proofs of Theorems \ref{thm:ndd_RZ} and  \ref{thm:ndd_Z}.  For the former,  $\PDhat$ contains all $k$ defective items and $o(k)$ non-defective items under the conditions \eqref{eq:nid}--\eqref{eq:nind}, and for the latter, $\Shat$ contains all of the defective items and no non-defective items under the conditions \eqref{eq:niid}--\eqref{eq:niind}.   The details are omitted to avoid repetition.
\end{proof}

%While the theorem statement is somewhat complex, the constants are easy to compute numerically, only requiring an optimization over $\alpha$, $\beta$, $\xi$, and $\nu$ (or any subset thereof).  The terms in \eqref{eq:nid}--\eqref{eq:nind} are required for the first step to be successful in the sense that $\PDhat$ contains all $k$ defective items and $o(k)$ non-defective items.  Similarly, the terms in \eqref{eq:niid}--\eqref{eq:niind} are required for the second step to be successful, in the sense that $\Shat$ contains all of the defective items and no non-defective items (i.e., $\Shat = S$). 
% The superscripts $\mathrm{D}$ and $\mathrm{ND}$ indicate whether the corresponding ``error events'' concern defectives or non-defectives.

As before, we can rephrase Theorem \ref{thm:ndd} to give a statement in terms of rates. That is, we can rewrite \eqref{eq:nid}, \eqref{eq:nind}, \eqref{eq:niid} and \eqref{eq:niind} using \eqref{eq:handy} to show that for given  $\alpha$, $\beta$, $\xi$, and $\nu$ an achievable rate is: 
\begin{multline} \label{eq:bscrate}
\frac{1-\theta}{\log 2} \min \bigg\{ \frac{\nu}{\theta} D_{\rho_{01}}(\alpha),
\frac{\nu}{1-\xi} D_w(\alpha), \\  \frac{\nu e^{-\nu}}{\theta} D_{1-\rho_{10}}(\beta),
\frac{\nu e^{-\nu}}{\xi} D_{\rho_{01}}(\beta)
\bigg\}.
\end{multline}
% here we write $w = (1-\rho) e^{-\nu} + \rho(1- e^{-\nu})$, which is $\geq \rho$ when $\rho \leq 1/2$. Note that for a given $\xi$ and $\nu$ we can minimize the first two terms over $\alpha$ by finding the unique intersection of these curves in the range $(\rho,w)$, and we can minimize the first second terms over $\beta$ by finding the unique intersection of these curves in the range $(\rho,1-\rho)$.  This could be done explicitly as before using Lemma \ref{lem:tech}, but the resulting expressions are not likely to provide particular insight. Having done this, we can numerically optimize over $\xi$ (based on the reverse Z-channel and Z-channel, it appears that a good rate will be obtained by taking $\nu = 1$).

We claim that in the limit as $\rho_{01} \to 0$ and $\rho_{10} \to 0$ for fixed $\theta \in (0,1)$, we recover the noiseless DD guarantee given in \cite{Ald14a}.  To see this, we let both $\alpha$ and $\beta$ equal an arbitrarily small constant $c > 0$.  We observe that $D_1(\alpha/\rho)$ and $D_1(\beta/\rho)$ both scale as $\frac{1}{\rho} \log \frac{1}{\rho}$, which implies that $\nid$ and $\niind$ behave as $o(k \log k)$, and their contributions are insignificant.  On the other hand, $\alpha/w$ and $\beta/(1-\rho_{10})$ can be made arbitrarily close to zero by suitable choice of $c$, which means (using $D_1(0) = 1$ and $w|_{\rho_{01}=\rho_{10}=0} = e^{-\nu}$) that $\nind$ and $\niid$ can be made arbitrarily close to $(1-\xi)\frac{k \log p}{\nu e^{-\nu}}$ and $\frac{k \log k}{\nu e^{-\nu}}$, respectively.  Since $\xi$ can be chosen arbitrarily close to $\theta$ and $(1-\theta)k\log p = \big(k\log\frac{p}{k} \big)(1+o(1))$, the final condition on $n$ is
\begin{equation}
n \ge \max\bigg\{ \frac{k \log \frac{p}{k}}{\nu e^{-\nu}}, \frac{k \log k}{\nu e^{-\nu}} \bigg\} (1+\eta)
\end{equation}
for arbitrarily small $\eta > 0$.  This bound is minimized by the choice $\nu = 1$, which recovers the bound in \cite{Ald14a}.

\begin{figure}
    \begin{centering}
        \includegraphics[width=0.95\columnwidth]{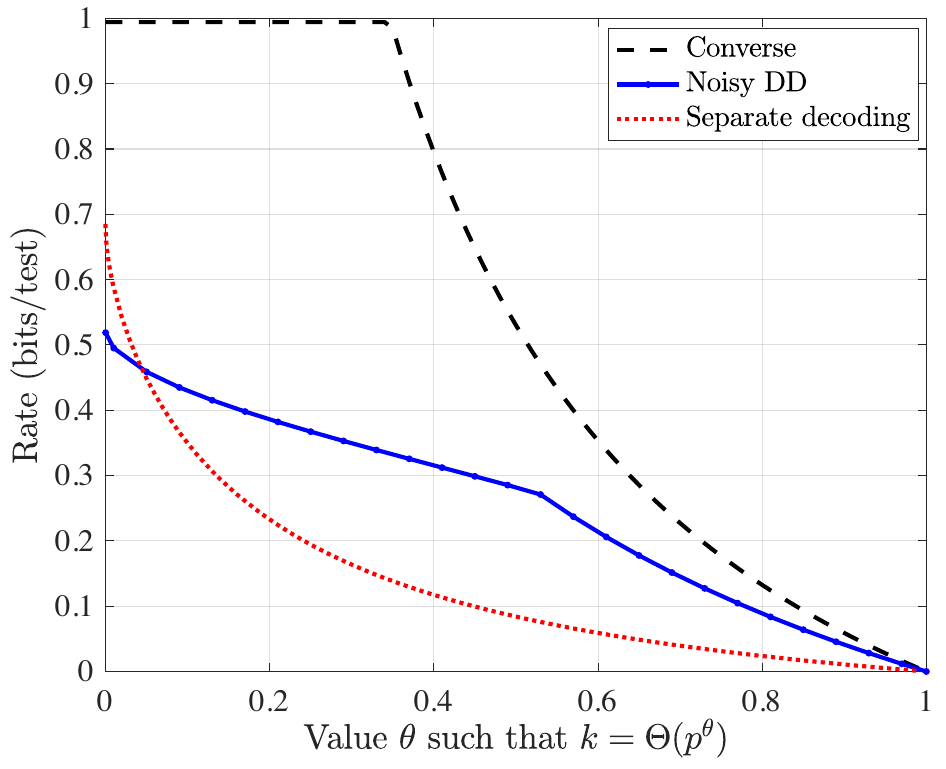}
        \par
    \end{centering}
    
    \caption{Achievable and converse rates under the symmetric noise model and i.i.d.~Bernoulli testing with noise level $\rho = 0.001$. \label{fig:RatesBSC}}
\end{figure}

{\bf Converse.} We can construct a general binary channel with noise levels $\rho_{01} + \rho_{10} \leq 1$ by the composition of a reverse Z-channel followed by a Z-channel. To be precise, analyzing the relevant conditional probabilities,  using a reverse Z-channel with noise level $\rho_{01}/(1-\rho_{10}) \leq 1$ followed by a Z-channel with noise level $\rho_{10}$ gives an general binary channel with noise levels  $\rho_{01}$ and $\rho_{10}$. 

Hence, a standard genie argument shows that any converse that applies to the reverse Z-channel must also apply to the general binary channel. In particular, evaluating the converse bound $\Rbar{RZ} \left(\theta, \rho_{01}/(1-\rho_{10}) \right)$ of Theorem \ref{thm:RZrate}  gives a converse that applies to the general binary channel. % Similarly, we can deduce that $\Rbar{Z}(\theta, \rho_{10})$ from Theorem \ref{thm:Zrate} will also bound the rate of the general binary channel. 

We can use a similar argument based on taking a Z-channel with noise level $\rho_{10}/(1-\rho_{01}) \leq 1$ followed by a reverse Z-channel with noise level $\rho_{01}$. The same genie argument shows that we can bound the rate of the general binary channel by $\Rbar{Z} \left( \theta, \rho_{10}/(1-\rho_{01}) \right)$. % and by $\Rbar{RZ} \left( \theta, \rho_{01} \right)$.
Putting these bounds together, %using the fact that $\Rbar{Z}(\theta, \rho)$ and $\Rbar{RZ}(\theta, \rho)$ are decreasing in $\rho$, 
we deduce the following.
\begin{cor} \label{cor:general_conv}
For the general binary noisy group testing problem with noise levels $\rho_{01} + \rho_{10} \leq 1$, in the regime $k \asymp p^{\theta}$ with $\theta \in (0,1)$,  no algorithm can achieve $P_e \rightarrow 0$ under Bernoulli testing with a rate higher than
    \begin{multline} \label{eq:symconv}
    \Rbar{gen}(\theta, \rho_{01}, \rho_{10}) \\ = \min \left\{ \Rbar{RZ} \left(\theta,  \frac{\rho_{01}}{1-\rho_{10}} \right), \Rbar{Z} \left( \theta, \frac{\rho_{10}}{1-\rho_{01}} \right) \right\}.
    \end{multline}
\end{cor}

In the degenerate case of $\rho_{01} + \rho_{10} = 1$ the initial reverse Z-channel has a noise level of $1$, meaning that all inputs are deterministically mapped to $1$ by this step. In this case, as expected, the converse bound \eqref{eq:symconv} is given by $\Rbar{RZ}(\theta,1) = 0$ (where this value follows by \eqref{eq:RZconverse}).

For comparison purposes, we discuss the symmetric case $\rho_{01} = \rho_{10} = \rho$.
In this case, while Corollary \ref{cor:general_conv} does not exactly match Theorem \ref{thm:ndd} for any $\theta \in (0,1)$ and $\rho \in (0,1)$, the two become increasingly close for high $\theta$ and low $\rho$; see Figure \ref{fig:RatesBSC} for an example with $\rho = 0.001$.  In this figure, we also observe a strict improvement over the best previously known rate attained by separate decoding of items \cite{Sca17b} unless $\theta$ is very small. 

The difficulty in establishing a tight bound when both $\rho_{01}$ and $\rho_{10}$ are positive appears to stem from the requirement of handling all four possible error types (false positive vs.~false negative, and first stage vs.~second stage); in contrast, for the Z and RZ models, only a strict subset of these is relevant.  The fact that these error events do not occur independently of one another poses a significant challenge for a tight analysis.  Similarly, in the converse proof, constructing a set $S \setminus \{i\} \cup \{j\}$ with a higher likelihood than $S$ is complicated by the fact that the likelihood depends on all four combinations of $P_{Y|U}(y|u)$ with $u,y \in \{0,1\}$.  Further closing the remaining gaps remains an interesting direction for future research.

\section*{Acknowledgment}

J.~Scarlett was supported by an NUS Startup Grant.

\bibliographystyle{IEEEtran}
\bibliography{JS_References,papers}

\begin{thebibliography}{10}
\providecommand{\url}[1]{#1}
\csname url@samestyle\endcsname
\providecommand{\newblock}{\relax}
\providecommand{\bibinfo}[2]{#2}
\providecommand{\BIBentrySTDinterwordspacing}{\spaceskip=0pt\relax}
\providecommand{\BIBentryALTinterwordstretchfactor}{4}
\providecommand{\BIBentryALTinterwordspacing}{\spaceskip=\fontdimen2\font plus
\BIBentryALTinterwordstretchfactor\fontdimen3\font minus
  \fontdimen4\font\relax}
\providecommand{\BIBforeignlanguage}[2]{{%
\expandafter\ifx\csname l@#1\endcsname\relax
\typeout{** WARNING: IEEEtran.bst: No hyphenation pattern has been}%
\typeout{** loaded for the language `#1'. Using the pattern for}%
\typeout{** the default language instead.}%
\else
\language=\csname l@#1\endcsname
\fi
#2}}
\providecommand{\BIBdecl}{\relax}
\BIBdecl

\bibitem{Ald19}
M.~Aldridge, O.~Johnson, and J.~Scarlett, ``Group testing: An information
  theory perspective,'' \emph{Found. Trend. Comms. Inf. Theory}, vol.~15, no.
  3--4, pp. 196--392, 2019.

\bibitem{Dor43}
R.~Dorfman, ``The detection of defective members of large populations,''
  \emph{Ann. Math. Stats.}, vol.~14, no.~4, pp. 436--440, 1943.

\bibitem{Ant11}
A.~Fern\'andez~Anta, M.~A. Mosteiro, and J.~Ram\'on Mu\~{n}oz, ``Unbounded
  contention resolution in multiple-access channels,'' in \emph{Distributed
  Computing}.\hskip 1em plus 0.5em minus 0.4em\relax Springer Berlin
  Heidelberg, 2011, vol. 6950, pp. 225--236.

\bibitem{erlich2}
Y.~Erlich, A.~Gilbert, H.~Ngo, A.~Rudra, N.~Thierry-Mieg, M.~Wootters,
  D.~Zielinski, and O.~Zuk, ``Biological screens from linear codes: Theory and
  tools,'' 2015, https://www.biorxiv.org/content/10.1101/035352v1.article-info.

\bibitem{goodrich2}
M.~T. Goodrich, M.~J. Atallah, and R.~Tamassia, ``Indexing information for data
  forensics,'' in \emph{Conf. Appl. Crypt. Net. Security}.\hskip 1em plus 0.5em
  minus 0.4em\relax Springer, 2005, pp. 206--221.

\bibitem{Cli10}
R.~Clifford, K.~Efremenko, E.~Porat, and A.~Rothschild, ``Pattern matching with
  don't cares and few errors,'' \emph{J. Comp. Sys. Sci.}, vol.~76, no.~2, pp.
  115--124, 2010.

\bibitem{Cor05}
G.~Cormode and S.~Muthukrishnan, ``What's hot and what's not: Tracking most
  frequent items dynamically,'' \emph{ACM Trans. Database Sys.}, vol.~30,
  no.~1, pp. 249--278, March 2005.

\bibitem{Gil08}
A.~Gilbert, M.~Iwen, and M.~Strauss, ``Group testing and sparse signal
  recovery,'' in \emph{Asilomar Conf. Sig., Sys. and Comp.}, Oct. 2008, pp.
  1059--1063.

\bibitem{Gil07}
A.~C. Gilbert, M.~J. Strauss, J.~A. Tropp, and R.~Vershynin, ``One sketch for
  all: Fast algorithms for compressed sensing,'' in \emph{Proc. ACM-SIAM Symp.
  Disc. Alg. (SODA)}, New York, 2007, pp. 237--246.

\bibitem{Ati12}
G.~Atia and V.~Saligrama, ``Boolean compressed sensing and noisy group
  testing,'' \emph{IEEE Trans. Inf. Theory}, vol.~58, no.~3, pp. 1880--1901,
  March 2012.

\bibitem{Bal13}
L.~Baldassini, O.~Johnson, and M.~Aldridge, ``The capacity of adaptive group
  testing,'' in \emph{IEEE Int. Symp. Inf. Theory}, July 2013, pp. 2676--2680.

\bibitem{Ald14a}
M.~Aldridge, L.~Baldassini, and O.~Johnson, ``Group testing algorithms: Bounds
  and simulations,'' \emph{IEEE Trans. Inf. Theory}, vol.~60, no.~6, pp.
  3671--3687, June 2014.

\bibitem{Sca15b}
J.~Scarlett and V.~Cevher, ``Phase transitions in group testing,'' in
  \emph{Proc. ACM-SIAM Symp. Disc. Alg. (SODA)}, 2016.

\bibitem{Ald15}
M.~Aldridge, ``The capacity of {B}ernoulli nonadaptive group testing,''
  \emph{IEEE Trans. Inf. Theory}, vol.~63, no.~11, pp. 7142--7148, 2017.

\bibitem{Joh16}
O.~Johnson, M.~Aldridge, and J.~Scarlett, ``Performance of group testing
  algorithms with near-constant -item,'' \emph{IEEE Trans. Inf. Theory},
  vol.~65, no.~2, pp. 707--723, Feb. 2019.

\bibitem{Coj19}
A.~Coja-Oghlan, O.~Gebhard, M.~Hahn-Klimroth, and P.~Loick,
  ``Information-theoretic and algorithmic thresholds for group testing,'' in
  \emph{Int. Colloq. Aut., Lang. and Prog. (ICALP)}, 2019.

\bibitem{Cha14}
C.~L. Chan, S.~Jaggi, V.~Saligrama, and S.~Agnihotri, ``Non-adaptive group
  testing: Explicit bounds and novel algorithms,'' \emph{IEEE Trans. Inf.
  Theory}, vol.~60, no.~5, pp. 3019--3035, May 2014.

\bibitem{Laa14}
T.~Laarhoven, ``Asymptotics of fingerprinting and group testing: Tight bounds
  from channel capacities,'' \emph{IEEE Trans. Inf. Forens. Sec.}, vol.~10,
  no.~9, pp. 1967--1980, 2015.

\bibitem{Sca17b}
J.~Scarlett and V.~Cevher, ``Near-optimal noisy group testing via separate
  decoding of items,'' \emph{IEEE Trans. Sel. Topics Sig. Proc.}, vol.~2,
  no.~4, pp. 625--638, 2018.

\bibitem{Sca18}
J.~{Scarlett}, ``Noisy adaptive group testing: Bounds and algorithms,''
  \emph{IEEE Trans. Inf. Theory}, vol.~65, no.~6, pp. 3646--3661, June 2019.

\bibitem{hwang}
F.~K. Hwang, ``A method for detecting all defective members in a population by
  group testing,'' \emph{J. Amer. Stat. Assoc.}, vol.~67, no. 339, pp.
  605--608, 1972.

\bibitem{Ald14}
M.~Aldridge, L.~Baldassini, and K.~Gunderson, ``Almost separable matrices,''
  \emph{J. Comb. Opt.}, pp. 1--22, 2015.

\bibitem{Cha11}
C.~L. Chan, P.~H. Che, S.~Jaggi, and V.~Saligrama, ``Non-adaptive probabilistic
  group testing with noisy measurements: Near-optimal bounds with efficient
  algorithms,'' in \emph{Allerton Conf. Comm., Ctrl., Comp.}, Sep. 2011, pp.
  1832--1839.

\bibitem{tallini2002capacity}
L.~G. Tallini, S.~Al-Bassam, and B.~Bose, ``On the capacity and codes for the
  z-channel,'' in \emph{IEEE Int. Symp. Inf. Theory}, 2002.

\bibitem{Mal80}
M.~B. Malyutov and P.~S. Mateev, ``Screening designs for non-symmetric response
  function,'' \emph{Mat. Zametki}, vol.~29, pp. 109--127, 1980.

\bibitem{Mal78}
M.~Malyutov, ``\BIBforeignlanguage{English}{The separating property of random
  matrices},'' \emph{\BIBforeignlanguage{English}{Math. Notes Acad. Sci.
  {USSR}}}, vol.~23, no.~1, pp. 84--91, 1978.

\bibitem{Sca15}
J.~Scarlett and V.~Cevher, ``Limits on support recovery with probabilistic
  models: An information-theoretic framework,'' \emph{IEEE Trans. Inf. Theory},
  vol.~63, no.~1, pp. 593--620, 2017.

\bibitem{Sca16b}
------, ``Converse bounds for noisy group testing with arbitrary measurement
  matrices,'' in \emph{IEEE Int. Symp. Inf. Theory}, Barcelona, 2016.

\bibitem{Sed10}
D.~Sejdinovic and O.~Johnson, ``Note on noisy group testing: Asymptotic bounds
  and belief propagation reconstruction,'' in \emph{Allerton Conf. Comm.,
  Control and Comp.}, 2010.

\bibitem{Mal12}
D.~Malioutov and M.~Malyutov, ``Boolean compressed sensing: {LP} relaxation for
  group testing,'' in \emph{IEEE Int. Conf. Acoust. Sp. Sig. Proc. (ICASSP)},
  March 2012, pp. 3305--3308.

\bibitem{Cai13}
S.~Cai, M.~Jahangoshahi, M.~Bakshi, and S.~Jaggi, ``Efficient algorithms for
  noisy group testing,'' \emph{IEEE Trans. Inf. Theory}, vol.~63, no.~4, pp.
  2113--2136, 2017.

\bibitem{Lee15a}
K.~Lee, R.~Pedarsani, and K.~Ramchandran, ``{SAFFRON}: A fast, efficient, and
  robust framework for group testing based on sparse-graph codes,'' 2015,
  http://arxiv.org/abs/1508.04485.

\bibitem{Ina19}
H.~A. {Inan}, P.~{Kairouz}, M.~{Wootters}, and A.~{\"Ozg\"ur}, ``On the
  optimality of the {K}autz-{S}ingleton construction in probabilistic group
  testing,'' \emph{IEEE Trans. Inf. Theory}, vol.~65, no.~9, pp. 5592--5603,
  Sept. 2019.

\bibitem{Bon19a}
S.~Bondorf, B.~Chen, J.~Scarlett, H.~Yu, and Y.~Zhao, ``Sublinear-time
  non-adaptive group testing with {$O(k \log n)$} tests via bit-mixing
  coding,'' 2019, https://arxiv.org/abs/1904.10102.

\bibitem{Che13a}
M.~Cheraghchi, ``Noise-resilient group testing: Limitations and
  constructions,'' \emph{Disc. App. Math.}, vol. 161, no.~1, pp. 81--95, 2013.

\bibitem{Ngo11}
H.~Q. Ngo, E.~Porat, and A.~Rudra, ``Efficiently decodable error-correcting
  list disjunct matrices and applications,'' in \emph{Int. Colloq. Automata,
  Lang., and Prog.}, 2011.

\bibitem{Ind10}
P.~Indyk, H.~Q. Ngo, and A.~Rudra, ``Efficiently decodable non-adaptive group
  testing,'' in \emph{ACM-SIAM Symp. Disc. Alg. (SODA)}, 2010.

\bibitem{Mac97}
A.~J. Macula, ``Error-correcting nonadaptive group testing with de-disjunct
  matrices,'' \emph{Disc. App. Math.}, vol.~80, no. 2-3, pp. 217--222, 1997.

\bibitem{madej}
T.~Madej, ``An application of group testing to the file comparison problem,''
  in \emph{Distributed Computing Systems, 1989., 9th International Conference
  on}.\hskip 1em plus 0.5em minus 0.4em\relax IEEE, 1989, pp. 237--243.

\bibitem{wolf}
J.~K. Wolf, ``Born again group testing: Multiaccess communications,''
  \emph{IEEE Trans. Inf. Theory}, vol.~31, no.~2, pp. 185--191, 1985.

\bibitem{berger}
T.~Berger and V.~I. Levenshtein, ``Asymptotic efficiency of two-stage
  disjunctive testing,'' \emph{IEEE Trans. Inf. Theory}, vol.~48, no.~7, pp.
  1741--1749, 2002.

\bibitem{corless}
R.~M. Corless, G.~H. Gonnet, D.~E.~G. Hare, D.~J. Jeffrey, and D.~E. Knuth,
  ``On the {Lambert $W$} function,'' \emph{Adv. Comp. Math.}, vol.~5, no.~1,
  pp. 329--359, 1996.

\bibitem{Mot10}
R.~Motwani and P.~Raghavan, \emph{Randomized Algorithms}.\hskip 1em plus 0.5em
  minus 0.4em\relax Chapman \& Hall/CRC, 2010.

\bibitem{ash}
R.~B. Ash, \emph{Information Theory}.\hskip 1em plus 0.5em minus 0.4em\relax
  Dover Publications Inc., New York, 1990.

\bibitem{joag-dev}
K.~Joag-Dev and F.~Proschan, ``Negative association of random variables with
  applications,'' \emph{Ann. Stats.}, vol.~11, pp. 286--295, 1983.

\end{thebibliography}

\newpage
 \begin{IEEEbiographynophoto}{Jonathan Scarlett}
     (S'14 -- M'15) received 
     the B.Eng. degree in electrical engineering and the B.Sci. degree in 
     computer science from the University of Melbourne, Australia. 
     From October 2011 to August 2014, he
     was a Ph.D. student in the Signal Processing and Communications Group
     at the University of Cambridge, United Kingdom. From September 2014 to
     September 2017, he was post-doctoral researcher with the Laboratory for
     Information and Inference Systems at the \'Ecole Polytechnique F\'ed\'erale
     de Lausanne, Switzerland. Since January 2018, he has been an assistant
     professor in the Department of Computer Science and Department of Mathematics,
     National University of Singapore. His research interests are in
     the areas of information theory, machine learning, signal processing, and
     high-dimensional statistics. He received the Singapore National Research Foundation (NRF) fellowship, and the NUS Early Career Research Award.
 \end{IEEEbiographynophoto}

\begin{IEEEbiographynophoto}{Oliver Johnson}
received the B.A. degree in 1995, Part III Mathematics in 1996, and the Ph.D. degree in 2000, all from the University of Cambridge, Cambridge, U.K. He was Clayton Research Fellow at Christ's College and Max Newman Research Fellow at Cambridge University until 2006, during which time he published the book Information Theory and the Central Limit Theorem (Singapore: World Scientific, 2004). Since 2006, he has been at University of Bristol, Bristol, U.K, where he is Professor of Information Theory. Within Bristol Mathematics he is Director of the Institute for Statistical Science and Programme Director for the MSc in the Mathematics of Cybersecurity.
\end{IEEEbiographynophoto}
 
\end{document}